\documentclass[11pt]{amsart}

\usepackage{lmodern}

\allowdisplaybreaks 

\usepackage{tikz-cd}

\usepackage[draft=false]{hyperref}

\newtheorem{theorem}{Theorem}[section]
\newtheorem{lemma}[theorem]{Lemma}
\newtheorem{proposition}[theorem]{Proposition}
\newtheorem{corollary}[theorem]{Corollary}

\theoremstyle{definition}
\newtheorem{defn}[theorem]{Definition}

\newtheorem{example}[theorem]{Example}
\newtheorem{remark}[theorem]{Remark}

\newcommand{\CQ}{\cl{C\mkern-1.5mu Q}}

\newcommand\ran{\mathop{\rm ran}}
\DeclareMathOperator{\id}{id}
\DeclareMathOperator{\tr}{Tr}

\newcommand\nph{\varphi}
\newcommand\eps{\epsilon}

\newcommand\rank{\mathop{\rm rank}}

\newcommand{\tightmath}{%
  \thinmuskip=0mu plus 3mu
  \medmuskip=1mu plus 4mu
  \thickmuskip=2mu plus 5mu}

\newcommand{\cl}[1]{\mathcal{#1}}
\newcommand{\bb}[1]{\mathbb{#1}}
\DeclareMathOperator{\Tr}{Tr}

\newcommand{\la}{\langle}
\newcommand{\ra}{\rangle}
\newcommand{\ten}{\otimes}

\newcommand{\om}{\omega}
\newcommand{\ep}{\varepsilon}
\newcommand{\vphi}{\varphi}
\newcommand{\lm}{\lambda}
\providecommand{\norm}[1]{\lVert#1\rVert}
\providecommand{\bignorm}[1]{\bigg\lVert#1\bigg\rVert}

\newcommand{\N}{\mathbb{N}}
\newcommand{\BH}{\mathcal{B}(H)}

\newcommand{\ip}[2]{\left\la #1,#2\right\ra}

\begin{document}

\title
{Values of cooperative quantum games}


\author[J. Crann]{Jason Crann$^{1}$}
\address{$^1$School of Mathematics \& Statistics, Carleton University, Ottawa, ON, Canada H1S 5B6}
\email{jasoncrann@cunet.carleton.ca}

\author[R. H. Levene]{Rupert H. Levene$^{2,3}$}
\address{$^2$School of Mathematics and Statistics, University College Dublin, Belfield, Dublin 4, Ireland}
\email{rupert.levene@ucd.ie}

 \address{$^3$Centre for Quantum Engineering, Science, and Technology, University College Dublin, Belfield, Dublin 4, Ireland}

\author[I. G. Todorov]{Ivan G. Todorov$^{4}$}
\address{$^{4}$School of Mathematical Sciences, University of Delaware, 501 Ewing Hall, Newark, DE 19716, USA}
\email{todorov@udel.edu}

\author[L. Turowska]{Lyudmila Turowska$^{5}$}
\address{$^{5}$Department of Mathematical Sciences, Chalmers University of Technology and the University of Gothenburg, Gothenburg SE-412 96, Sweden}
\email{turowska@chalmers.se}

\begin{abstract} We develop a resource-theoretical approach that allows us 
to quantify values of two-player, one-round cooperative games with quantum inputs and outputs, 
as well as values of quantum probabilistic hypergraphs. 
We analyse the quantum game values 
arising from the type hierarchy of quantum no-signalling correlations, 
establishing tensor norm expressions for each of the correlation types.  
As a consequence, we provide metric characterisations of state convertibility via 
LOSR and LOCC.
En route, we obtain an 
alternative description of the maximal tensor products of ternary rings of operators. 
\end{abstract}

\date{28 September 2023}

\maketitle

\tableofcontents


\section{Introduction}\label{s_intro}

The past decade witnessed a plethora of fruitful interactions between 
operator space theory, initiated in the 1980's mostly 
intrinsically as an area of functional analysis \cite{blm, er, Pa, pisier_intr, Pi2}, 
and quantum information theory, whose advent in the 1990's 
was motivated by developments in quantum computing \cite{nielsen-chuang, petz, watrous}. 
Among the first samples of this interaction are \cite{jppvw, pwpv-junge}, where 
violations of the fundamental in physics
Bell inequalities were placed in an operator space framework.  
Further connections between operator spaces and operator systems on one hand, 
and quantum information on the other, have been pursued in 
\cite{ghj, jppvw2, jp-2011, jkppg, jkpv, lmprsstw, musat-rordam, pr, psstw, pt_QJM}, among others.
In \cite{jnppsw}, a close connection between Tsirelson's Problem  
in theoretical physics \cite{tsirelson}
and the Connes Embedding Problem in operator algebra theory
\cite{connes} was made, leading to the proof of their equivalence in \cite{oz}.

A negative answer to the Connes Embedding Problem was announced in \cite{jnvwy},
with (classical) non-local games lying at the base of the authors' approach therein. 
A non-local game is a 
cooperative game played by two parties, say Alice and Bob, 
against a verifier; each round of the game consists of the verifier selecting a pair of inputs, represented by 
elements $x$ and $y$ of finite sets $X$ and $Y$, respectively, according to some 
probability distribution on $X\times Y$, 
and sending $x$ to Alice, and $y$ to Bob. 
The players return, without communicating, 
outputs $a$ and $b$, represented by elements of finite sets $A$ and $B$, respectively. 
The round is won by the tandem Alice-Bob if the quadruple $(x,y,a,b)$ satisfies a previously known 
rule predicate, which can be represented as a subset 
$E\subseteq X\hspace{-0.05cm}\times \hspace{-0.05cm}Y \hspace{-0.05cm}\times\hspace{-0.05cm} 
A\hspace{-0.05cm}\times\hspace{-0.05cm} B$, capturing the allowed responses of the players. 
A strategy for Alice and Bob is a family 
$p = \{(p(a,b|x,y))_{a,b} : x\in X, y\in Y\}$ 
of conditional 
probability distributions, where the value
$p(a,b|x,y)$ represents the likelihood of the players producing an output pair $(a,b)$ given an input pair $(x,y)$. 

There are several, successively stronger, types of strategies that can be used by Alice and Bob,
depending on the physical model they employ. Each of these 
strategy types leads to a corresponding optimal winning probability -- or \emph{game value} --  
arriving at an increasing chain of such values. 
In \cite{jnvwy}, the authors show the existence of a non-local game $\cl G$, whose 
\emph{quantum value}
$\omega_{\rm q}(\cl G)$, which (corresponds to the tensor model of quantum entanglement and) 
a priori does not exceed the
\emph{quantum commuting value} 
$\omega_{\rm qc}(\cl G)$ (arising from the commuting operator model of quantum physics), can be in fact strictly smaller than the latter
(we refer the reader to \cite{psstw} for the definitions of the aforementioned quantum mechanical models).

Classical non-local games are combinatorial objects, leading to substantial variation restrictions:
there are only finitely many games of 
fixed size and a given probability distribution over the input sets. 
Therefore, at least on a theoretical level, 
considering \emph{quantum games}, where the inputs and the outputs are  
quantum states as opposed to letters from finite alphabets, 
is expected to provide a significant advantage 
(the very set of quantum games of given size with a fixed probability distribution over the inputs is 
a manifold). 
Many types of quantum games have been considered in the literature, including
cooperative quantum games (e.g. rank-one quantum games) \cite{junge,ltw}, semi-quantum games \cite{bus}, extended non-local games \cite{fritz,johnetal,tometal}, quantum XOR games \cite{regev-vidick} and 
quantum non-local games \cite{tt-QNS} (see also \cite{bhtt, bhtt2}). 
In \cite{junge}, a game tensor was associated to a given rank-one quantum game, and 
the quantum value was expressed via its norm in the minimal tensor product of two copies of 
the operator space of the trace class. 
No similar expression is however known for the quantum commuting value of quantum games
in any of their variants.

This was the starting point of the present work. 
We consider two interrelated contexts: an extension of the setting employed in \cite{junge} beyond the 
rank one projection case 
(these can be thought of as games with a fuzzy rule 
predicate that assigns weights and not a win/lose status to the player's outputs), 
and measurable extensions of the games in \cite{tt-QNS, bhtt, bhtt2}, called 
herein quantum hypergraph games.  
As one of our main results (see Theorems \ref{th_qcval} and \ref{th_hypqhval}), 
we show that the quantum commuting value of such games
coincides with the norm of associated game tensors in 
a maximal tensor product of operator spaces arising from the trace class. 
En route, we re-examine the maximal tensor product of TRO's introduced in \cite{kr}, 
and show that it can be obtained by requiring seemingly weaker conditions on the ternary 
representations used in the computation of the matricial norms, involving  
commutation of suitable corners of the associated linking algebras only, as opposed to 
the commutation of the entire linking algebras
(see Theorem \ref{t:maxnorm}).
We achieve operator space tensor norm expressions for 
the local and the general no-signalling value, 
as well as for the one-way classical communication value of a given game. 
We formalise the definitions of quantum game values, placing them in a 
systematic framework, where the different types of values are defined using 
corresponding classes of strategies, making an explicit parallel to the 
definitions of values of classical games \cite{palazuelos-vidick}.
As strategy classes, we use the quantum no-signalling correlation types introduced in \cite{dw, tt-QNS}, 
which are natural quantum versions of the classical no-signalling correlation types (see e.g. \cite{lmprsstw}), 
generalising the aforementioned classical strategies $p = \{(p(a,b|x,y))_{a,b} : x\in X, y\in Y\}$. 

Each of the quantum no-signalling 
correlation classes can be viewed as a resource, of which the players avail during the course of the game,
and which they are allowed to use to produce their outputs. 
The tensor norm expressions for the quantum game values of different types are thus 
obtained as 
special cases of a more general framework, designed to host quantum hypergraph values, following a 
resource-theoretic approach.
For us, a quantum hypergraph over the pair $(X,A)$ of finite sets
is a Borel map $\nph$ from the pure state space $\bb{P}_X$ of the (indexed by $X\times X$)
matrix algebra $M_X$ into the projection lattice $\cl P_A$ of the matrix algebra $M_A$. 
When a Borel probability measure $\mu$ on $\bb{P}_X$ is given, we refer to the pair $(\nph,\mu)$ as a 
\emph{quantum probabilistic hypergraph.} 
Every \emph{resource} $\cl R$ -- a collection of 
$A\times X$-block operator isometries satisfying natural closure properties -- 
yields a collection of channels from $M_X$ into $M_A$ which, in its own turn, 
gives rise to the \emph{$\cl R$-value} $\omega_{\cl R}(\nph,\mu)$ of $(\nph,\mu)$. At the core of 
the subsequent applications to quantum games is an interpretation of $\omega_{\cl R}(\nph,\mu)$ as a norm 
of a canonical element of an operator space, based upon the trace class 
$\cl S_1(\bb{C}^A,\bb{C}^X)$, whose 
operator space structure is induced by the resource $\cl R$ in a canonical way. 
We refer to Section \ref{ss_hyper} for connections with quantum resource theories \cite{chitambar}.

The paper is organised as follows. After collecting some necessary operator space 
preliminaries in Section \ref{s:prelim}, in Section \ref{ss_copsp} we revisit the maximal tensor product 
of ternary rings of operators (TRO's) \cite{kr}. We prove a factorisation result for 
ternary representations of the algebraic tensor product of two TRO's (Proposition \ref{p:representations}), 
which allows us to see that the matricial norms in the maximal tensor product 
of the TRO's $\cl U$ and $\cl V$ are achieved using ternary representations 
$\phi : \cl U\to \cl B(K,L)$ and $\psi : \cl V\to \cl B(H,K)$
with the property that the sets of operators $\psi(\cl V)\psi(\cl V)^*$ and $\phi(\cl U)^*\phi(\cl U)$ (on the Hilbert 
space $K$) mutually commute. 

In Section \ref{s_hyper}, we introduce the general framework of quantum probabilistic hypergraphs and 
the notion of a resource. Using an identification of the trace class 
$\cl S_1(\bb{C}^A,\bb{C}^X)$
as an operator subspace of the universal TRO of a block operator isometry indexed by $A\times X$, 
we show that every resource $\cl R$ induces an operator space structure 
$\cl O_{X,A}^{\cl R}$ on $\cl S_1(\bb{C}^A,\bb{C}^X)$, 
for which the formal identity $\cl S_1(\bb{C}^A,\bb{C}^X)\to \cl O_{X,A}^{\cl R}$ is a complete contraction. 
The $\cl R$-value expressions for a quantum probabilistic hypergraph and a
generalised unipartite counterpart of rank one games is
established
in Theorems \ref{th_hypvge} and \ref{th_R-val}, respectively.

In Section \ref{s_qv}, after a reminder of the quantum no-signalling correlation chain 
$\cl Q_{\rm loc}\subseteq \cl Q_{\rm q} \subseteq \cl Q_{\rm qc}\subseteq \cl Q_{\rm ns}$
from \cite{dw, tt-QNS}, the setup from Section \ref{s_hyper} is applied to the 
bipartite context, which naturally hosts quantum games in the sense of \cite{junge,ltw} and \cite{tt-QNS}. 
For each type ${\rm t}\in \{{\rm loc}, {\rm q}, {\rm qc}, {\rm ns}\}$, we identify a resource $\cl R_{\rm t}$
that gives rise to the correlation class $\cl Q_{\rm t}$, thus leading, in view of Section \ref{s_hyper}, to 
tensor norm expressions  for the local, quantum, quantum commuting and no-signalling values of 
the underlying quantum game. 
As a consequence of our results in Section \ref{ss_copsp}, we show
in Corollary~\ref{c_adv}
that {for finite rank quantum games} the quantum commuting resource offers a bounded advantage over the quantum one,
with the bound on the violation depending on the rank, but independent on the size of the game.

We show that the local value is expressed in terms of the weak-{\rm cb} Hilbert-Schmidt norm, 
recently introduced in \cite{jkppg}. We provide a characterisation of the game value, 
in the case where the players
avail of the resource of local operations and one-way classical communication.
Since quantum hypergraph games can be thought of as the operational task of 
producing a state with prescribed support (depending on the input state), our results 
lead to a metric characterisation of pure state convertibility via local operations and shared randomness (LOSR) 
\cite{watrous}, and via local operations and (any amount of) classical communication (LOCC) \cite{clmow}. 
We note that state convertibility has been a topic of substantial interest in the past 
(see \cite{chitambar, cklt-CMP} the references therein); 
up to our knowledge, ours are the 
first such characterisations in the literature. 

In Section \ref{s_ctoq}, 
we provide an alternative characterisation of the quantum and the local value of 
quantum games, closer to the descriptions of the quantum correlation classes in terms 
of states on tensor products of universal C*-algebras presented in \cite{tt-QNS}. 
We apply these results to the case of games with classical inputs and quantum outputs, obtaining in 
Theorem \ref{qv_classtoquant} tensor norm expressions for the local, quantum, 
quantum commuting and no-signalling value of such a game. 
Finally, in Section \ref{s_examples} we consider some examples, providing
formulas for the value of the quantum homomorphism game, whose perfect strategies were previously 
characterised in \cite{tt-QNS} and \cite{bhtt}.

\subsection*{Acknowledgements}
The authors would like to thank Marius Junge, Carlos Palazuelos and Andreas Winter for useful discussions
on the topic of the paper. The present research was initiated during the AIM Workshop \lq\lq Non-local games in 
quantum information theory'' held virtually in May 2021. J.C. was supported by the NSERC Discovery Grants RGPIN-2017-06275 and RGPIN-2023-05133, and the NSERC Alliance International Catalyst Quantum Grant ALLRP 580899-22.
I.T. was supported by NSF grants CCF-2115071 and DMS-2154459. L.T. would like to thank Wenner-Gren Foundation which supported the visit of I.T. to Gothenburg in May 2022, also Stiftelsen Lars Hiertas Minne which supported the visit of L.T. to  University of Delaware in October 2022.\\

\subsection*{Note}
Upon the completion of this paper, we were informed by Roy Araiza, Marius Junge and Carlos Palazuelos 
of their work \cite{ajp}.
We are however not aware of any intersections of results between \cite{ajp} and the present work.


\section{Preliminaries}\label{s:prelim}

In this section, we set notation and review the relevant concepts from operator space theory, referring the reader to \cite{blm,er,pisier_intr} for details.

Let $H$ and $K$ be Hilbert spaces. We denote by $\cl B(H,K)$ the Banach space of all 
bounded linear operators from $H$ to $K$, write $\cl B(H) = \cl B(H,H)$, and 
let $I_H$ be the identity operator in $\cl B(H)$. 
We further denote by $\cl S_1(H,K)$ the space of trace class operators in $\cl B(H,K)$, let $\cl S_1(H):=\cl S_1(H,H)$, and write ${\rm Tr}:\cl S_1(H)\to \mathbb C$ for the trace 
functional on $\cl S_1(H)$. 
For vectors $\xi\in K$ and $\eta\in H$, we write $\xi\eta^*$ for the rank one operator 
in $\cl S_1(H,K)$, given by 
$(\xi\eta^*)(\zeta) = \langle \zeta,\eta\rangle \xi$. 

Throughout the paper, $X$, $Y$, $A$ and $B$ will denote non-empty finite sets.
For a Hilbert space $H$, let $H^X = \oplus_{x\in X} H$; we will sometimes write
$\ell_2^X = \bb{C}^X$.
We write
$M_{A,X}$ for the space of all $|A|\times |X|$ matrices with complex entries,
and identify it with $\cl B(\bb{C}^X,\bb{C}^A)$ 
in a standard fashion. 
We set $M_X = M_{X,X}$, let $\cl D_X$ be the subalgebra of $M_X$ of all diagonal matrices and write
$I_X = I_{\bb{C}^X}$. We will freely abbreviate the Cartesian product $A\times X$ as $AX$, so that $M_{AX}=M_{AX,AX}$.  
We let ${\rm Tr}_A : M_{AX}\to M_X$ be the partial trace. 
For an element $\omega\in M_{A,X}$, we let $\omega^{\rm t}\in M_{X,A}$ be its transpose. 
We let $(e_x)_{x\in X}$ be the canonical orthonormal basis of $\bb{C}^X$ and set 
$\epsilon_{x,x'} = e_x e_{x'}^*$, $x,x'\in X$. 
For a vector $\xi = \sum_{x\in X}\lambda_x e_x$, we write $\bar\xi = \sum_{x\in X}\overline{\lambda}_xe_x$. 
We set $[n] = \{1,2,\dots,n\}$, and 
abbreviate $M_{m,n} = M_{[m],[n]}$, $M_n = M_{[n]}$ and $I_n = I_{[n]}$.

Let $\cl P_X$ be the projection lattice of $M_X$ and $\cl P_X^{\rm cl}$ be the projection lattice of $\cl D_X$.
Let $\bb{P}_{X}$ be the pure state space on $\bb{C}^{X}$, equipped with the usual topology
(equivalently, $\bb{P}_{X}$ is the rank one projective space in $\bb{C}^{X}$), 
and $\bb{P}_{X}^{\rm cl}$ be the set of all pure states in $\cl P_X^{\rm cl}$; 
thus, $\bb{P}_X^{\rm cl} = \{\epsilon_{x,x} : x\in X\}$. 
Note the correspondence $\alpha\to P_{\alpha}$ between subsets $\alpha$ of $X$ and elements $P_{\alpha}$ 
of $\cl P_X^{\rm cl}$. 
If $X$ and $A$ are finite sets, a linear map $\Phi : M_X\to M_A$ is called 
a \emph{quantum channel}, or simply a channel,
if $\Phi$ is completely positive and trace preserving; 
we let $\mathsf{QC}(M_X,M_A)$ be the (convex) set of all quantum channels from $M_X$ to $M_A$.

We write
$\cl X\otimes \cl Y$ for the algebraic tensor product of complex vector spaces $\cl X$ and $\cl Y$, 
except when $\cl X$ and $\cl Y$ are Hilbert spaces,
in which case this notation
is reserved for their Hilbertian tensor product. 
We write $M_{m,n}(\cl X)$ for the vector space of all $m\times n$ matrices with entries in $\cl X$, and 
make the canonical identification $M_{m,n}\otimes \cl X = M_{m,n}(\cl X)$. If $m=n$, we denote 
$M_{m,n}(\cl X)$ by $M_n(\cl X)$.

A family of \textit{matricial norms} on a complex vector space $\cl X$ is a sequence of norms $(\norm{\cdot}_n)_{n\in\N}$,
where $\norm{\cdot}_n$ is a norm on $M_n(\cl X)$, $n\in\N$. The pair $(\cl X,(\norm{\cdot}_n)_{n\in\N})$ is an \textit{(abstract) operator space} if 
\begin{itemize}
\item[(i)] $\norm{v\oplus w}_{m+n}=\text{max}\{\norm{v}_m,\norm{w}_n\}$;
\item[(ii)] $\norm{\alpha v\beta}_n\leq\norm{\alpha}\norm{v}_m\norm{\beta}$,
\end{itemize}
for all $m,n\in\N, v\in M_m(\cl X), w\in M_n(\cl X), \alpha\in M_{n,m}(\mathbb{C}), \beta\in M_{m,n}(\mathbb{C})$. Any closed subspace $\cl X$ of $\BH$ inherits a canonical operator space structure via the canonical identification $M_n(\BH)=\cl B(H^n)$. In fact, by virtue of Ruan's Theorem,
every (abstract) operator space is completely isometrically isomorphic to a concrete operator space $\cl X\subseteq\BH$ for some Hilbert space $H$ \cite[Theorem 2.3.4]{er}.

If $\cl X$ is an operator space and an $A$-$B$-bimodule over unital $C^*$-algebras $A$ and $B$, we say that $\cl X$ is an \textit{operator $A$-$B$-bimodule} if the matricial norms on $\cl X$ satisfy
$$\norm{a\cdot v\cdot b}_n\leq\norm{a}\norm{v}_m\norm{b}$$
for all $m,n\in\N$, $v\in M_m(\cl X)$,
$a\in M_{n,m}(A)$,
$b\in M_{m,n}(B)$, where 
$$a\cdot v\cdot b=[\sum_{k,l}a_{i,k}\cdot v_{k,l}\cdot b_{l,j}]\in M_n(\cl X).$$
When $B=\bb{C}$, we simply refer to $\cl X$ as an operator $A$-module.

Given operator spaces $\cl X$ and $\cl Y$, and a linear map $\phi : \cl X\rightarrow \cl Y$, the 
$n$-th amplification of $\phi$ is the linear map $\phi^{(n)} : M_n(\cl X)\rightarrow M_n(\cl Y)$, defined by
$$\phi^{(n)}(v)=[\phi(v_{i,j})]_{i,j}, \quad v = [v_{i,j}]_{i,j} \in M_n(\cl X).$$
The map $\phi:\cl X\rightarrow \cl Y$ is
\begin{itemize}
\item \textit{completely bounded} if $\norm{\phi}_{\rm cb}:=\text{sup}\{\norm{\phi^{(n)}} : n\in\N\} < \infty$;
\item \textit{completely contractive} if $\norm{\phi}_{\rm cb} \leq 1$;
\item a \textit{complete isometry} if $\phi^{(n)}$ is an isometry for all $n\in\N$.
\end{itemize}

We denote by ${\rm CB}(\cl X,\cl Y)$ the space of completely bounded linear maps from 
$\cl X$ to $\cl Y$, which becomes an operator space under the canonical identification 
$M_n({\rm CB}(\cl X,\cl Y)) = {\rm CB}(\cl X,M_n(\cl Y))$. For example, if $\cl X$ is an operator space, its \textit{dual operator space} structure on $\cl X^*$ is given by ${\rm CB}(\cl X,\bb{C})$, and if $H$ is a Hilbert space, the \textit{column} and \textit{row} operator space structures on $H$ 
are given by the identifications 
$H_c:= {\rm CB}(\mathbb{C},H)$ and $H_r:= {\rm CB}(H^*,\mathbb{C})$. Unless otherwise stated, all Hilbert spaces will be given their column operator space structure. We record the following formula for future use (see \cite[3.4.2]{er})
\begin{equation}\label{e:column} \norm{[\xi_{i,j}]}_{M_n(H_c)}=\bignorm{\bigg[\sum_{k=1}^n\la \xi_{k,j},\xi_{k,i}\ra\bigg]}^{1/2}, \ \ \ [\xi_{i,j}]\in M_n(H_c).\end{equation}

An operator space $\cl X$ comes naturally equipped with norms on rectangular matrices $M_{m,n}(\cl X)$ via the canonical inclusion into $M_{\max\{m,n\}}(\cl X)$. This remains true more generally
when $m,n\in\N\cup\{\infty\}$, where
$$\norm{[v_{i,j}]}_{M_{\infty}(\cl X)}:=\sup_{n\in\N}\{\norm{[v_{i,j}]} : i,j=1,\dots,n\}.$$

Given operator spaces $\cl X$ and $\cl Y$, and elements
$v = [v_{i,j}]  \in M_{n,k}(\cl X)$ and $w = [w_{i,j}] \in M_{k,n}(\cl Y)$, we write
\[
v\odot w = \left[ \sum_{l=1}^k v_{i,l}\otimes w_{l,j} \right]_{i,j}
\] 
and recall that the \textit{Haagerup norm} of an element $u = [u_{ij}] \in M_n(\cl X\otimes \cl Y)$ is defined by 
letting
\[
\displaystyle \Vert u\Vert_{\rm h} = \inf
\{\Vert v\Vert _{n,k}\ \Vert w\Vert _{k,n} : u = v\odot w, v  \in M_{n,k}(\cl X), 
w \in M_{k,n}(\cl Y), k \in {\mathbb{N}}\}.
\]
The \emph{Haagerup tensor product} $\cl X\otimes_{\rm h} \cl Y$  is the completion of 
$\cl X\otimes \cl Y$ with respect to 
$\|\cdot\|_{\rm h}$. 
For $C^*$-algebras $\cl A$ and $\cl B$, we have that 
\[
\| u\|_{\rm h} = \inf \left\{
\left\| \sum_{k=1}^n a_k a_k^* \right\|^{\frac{1}{2}} \left\| \sum_{k=1}^n b_k^* b_k \right\|^{\frac{1}{2}}: \;  
u = \sum_{k=1}^n a_k \otimes b_k \right\}, \ \ u\in \cl X\ten \cl Y.
\] 
The symmetrized Haagerup tensor norm was defined  in \cite{oikhberg-pisier} by 
\[
\bignorm{\sum_{k=1}^n a_k\otimes b_k}_\mu=\sup\bignorm{\sum_{k=1}^n T(a_k)S(b_k)}_{\cl B(H)}
\] 
where the supremum runs over all Hilbert spaces $H$ and possible pairs of complete contractions $T:\cl X\to \cl B(H)$, $S:\cl Y\to \cl B(H)$ with commuting ranges, i.e. $T(a)S(b)=S(b)T(a)$, $a\in\cl X$, $b\in\cl Y$. By \cite[Theorem 1]{oikhberg-pisier},
\begin{equation}\label{eq_hht+}
\|u\|_\mu=\inf\{\|v\|_{\rm h}+\|w\|_{\rm h^t}: u=v+w\}
\end{equation}
where $\|\cdot\|_{\rm h^t}$ is the transpose version of the Haagerup norm, which is given by $\|\sum_{k=1}^na_k\otimes b_k\|_{\rm h^t}=\|\sum_{k=1}^nb_k\otimes a_k\|_{\rm h}$. One defines $\cl X\otimes_\mu\cl Y$ to be the completion of $\cl X\otimes\cl Y$  with respect to $\|\cdot\|_\mu$.
The Banach space $\cl X\otimes_\mu\cl Y$ has a natural operator space structure associated with the embedding $\cl X\otimes_\mu\cl Y\to \oplus_{\sigma}\cl B(H_\sigma)\subset\cl B(\oplus_\sigma H_\sigma)$, where the direct sum is taken over all pairs $\sigma=(T,S)$ of complete contractions  $T:\cl X\to \cl B(H_\sigma)$, $S:\cl Y\to\cl B(H_\sigma)$ with commuting ranges.

If $\cl X\subseteq \cl B(H)$ and $\cl Y\subseteq \cl B(K)$, 
the minimal tensor product $\cl  X \ten_{\min} \cl  Y$ of operator spaces is defined as the norm closure of 
$\cl X\ten \cl Y$ in $\cl B(H\ten K)$; we note that, up to complete isometry, 
$\cl X\ten_{\min} \cl Y$ is independent of the specific concrete realisations of $\cl X$ and $\cl Y$. 
We will also use the fact that, if $u=\sum_{i=1}^m x_i\otimes y_i$, with $x_i\in \cl X$ and $y_i\in \cl Y$, $i\in [m]$, then 
$$\|u\|_{\cl X\otimes_{\rm min}\cl Y}
= \sup_{(T,S)}\left\|\sum_{i=1}^m T(x_i)\otimes S(y_i)\right\|_{\cl B(H\otimes K)},$$ where the supremum is taken over all {Hilbert spaces $H$, $K$} and complete contractions $T:\cl X\to\cl B(H)$, $S:\cl Y\to \cl B(K)$. Moreover, {if $\cl X$ and $\cl Y$ are finite-dimensional}, we can assume that $H$ and $K$ are finite-dimensional
\cite[(8.1.7)]{er}. When $\cl X$ or $\cl Y$ is finite-dimensional,  we also have $\cl X^*\otimes_{\rm min}\cl Y\simeq {\rm CB}(\cl X,\cl Y)$, where the correspondence is given by $u=\sum_{i=1}^m x_i\otimes y_i\mapsto T_u$, $T_u(x)=\sum_{i=1}^mx_i(x)y_i$.

If $\cl X$ and $\cl Y$ are Banach spaces, the injective tensor norm of 
an element $z\in \cl X\otimes \cl Y$ is defined by letting
\begin{equation}\label{eq_inno3}
\|z\|_{\cl X\otimes_{\varepsilon} \cl Y} = 
\sup \left\{|(f\otimes g)(z)| : f\in \cl X^*,\; g\in \cl Y^* \mbox{ contractive}\right\}.
\end{equation}
Similarly, if $\cl X$ or $\cl Y$ are finite dimensional Banach spaces then 
$\cl X^*\otimes_\varepsilon \cl Y\simeq \cl B(\cl X,\cl Y)$, where $\cl B(\cl X,\cl Y)$ is the 
(Banach) space of all bounded linear operators from $\cl X$ into $\cl Y$.


\section{Tensor products of TRO's}\label{ss_copsp}

The main result of this section is a description of the maximal tensor 
product of two TRO's \cite{kr} in terms of semi-commuting 
ternary representations (see Definition \ref{d_semic} and Theorem \ref{t:maxnorm});
the result will be applied in Section \ref{ss_qc} to obtain an expression 
of the quantum commuting value of a non-local game, but it may be interesting in its own right. 

Recall \cite{hest,z} that a \textit{ternary ring} is a complex vector space $\cl U$, equipped with a ternary operation $[\cdot,\cdot,\cdot]: \cl U\times\cl U\times\cl U\to \cl U$, linear in 
the outer variables and conjugate linear in the middle variable, such that
$$[s,t,[u,v,w]] = [s,[v,u,t],w] = [[s,t,u],v,w], \quad s,t,u,v,w\in \cl U.$$
A \textit{ternary representation/morphism} of $\cl U$ is a linear map $\theta:\cl U\to\cl B(H,K)$ for some Hilbert spaces $H$ and $K$, such that
$$\theta([u, v, w]) = \theta(u)\theta(v)^*\theta(w), \quad u,v,w\in \cl U.$$
We call $\theta$ \textit{left non-degenerate} if 
$\overline{{\rm span}(\theta(\cl U)^*K)} = H$. 
A (concrete) \textit{ternary ring of operators (TRO)} \cite{z} is a norm-closed subspace $\cl U\subseteq \cl B(H,K)$ for some Hilbert spaces $H$ and $K$ such that $R,S,T\in\cl U \Longrightarrow RS^*T\in \cl U$. 
TRO's can be characterised abstractly as \emph{ternary $C^*$-rings} 
(see \cite[Theorem 3]{z}).

Every TRO $\cl U$ has a canonical operator space structure under which $M_n(\cl U)$ is a TRO for every $n\in\mathbb{N}$ \cite[\S1,\S2]{kr}; namely, if $\cl U\subseteq\cl B(H,K)$, 
the matrix norms over $\cl U$ 
are inherited from the identification $M_n(\cl B(H,K))=\cl B(H^n,K^n)$. By \cite[Proposition 2.1]{h} (see also \cite[Proposition 2.1]{kr}), 
a linear ternary isomorphism is automatically a complete isometry. 
Hence, the operator space structure on $\cl U$ thus defined is independent 
of the concrete isometric embedding $\cl U\subseteq\cl B(H,K)$, and
\begin{equation}\label{eq_TROopsp}\tightmath
\norm{[u_{i,j}]}_{M_n(\cl U)} = 
\sup\{\norm{\phi^{(n)}([u_{i,j}])} \ : \  \phi:\cl U\to\cl B(H,K) \ \textnormal{ternary morphism}\}.
\end{equation}

For a TRO $\cl U$, we will denote by $\cl L_{\cl U}$ (resp. $\cl R_{\cl U}$) its left
(resp. right) C*-algebra; 
if $\theta : \cl U\to \cl B(H,K)$ is a faithful (that is, injective) ternary morphism then, up to 
a *-isomorphism, we have that
$$\cl R_{\cl U} = \overline{{\rm span}(\theta(\cl U)^*\theta(\cl U))}
\ \ \mbox{ and } \ \ 
\cl L_{\cl U} = \overline{{\rm span}(\theta(\cl U)\theta(\cl U)^*)}.$$
We note that
every ternary morphism $\phi : \cl U\to \cl B(H,K)$ gives rise to *-representations
$\pi^R_{\phi} : \cl R_{\cl U}\to \cl B(H)$ and $\pi^L_{\phi} : \cl L_{\cl U}\to \cl B(K)$, 
such that 
$$\pi^R_{\phi}(u_1^* u_2) = \phi(u_1)^*\phi(u_2) \mbox{ and }
\pi^L_{\phi}(u_1 u_2^*) = \phi(u_1)\phi(u_2)^*, \ u_i\in \cl U, i = 1,2,$$
where the products $u_1^* u_2$ and $u_1 u_2^*$ are computed in any faithful ternary representation of $\cl U$. 
The block operator space 
$$\cl D_{\cl U} = 
\left(
\begin{matrix}
\cl L_{\cl U} & \cl U\\
\cl U^* & \cl R_{\cl U}
\end{matrix}
\right)
$$
is a C*-algebra in a canonical fashion, called the \emph{linking algebra} of $\cl U$, 
and every ternary morphism $\phi : \cl U\to \cl B(H,K)$ gives rise to a (unique)
*-representation $\tilde{\pi}_{\phi} : \cl D_{\cl U}\to \cl B(H\oplus K)$, given by 
$$
\tilde{\pi}_{\phi}
\left(\left(
\begin{matrix}
l & u\\
v^* & r
\end{matrix}
\right)\right)
= 
\left(
\begin{matrix}
\pi^L_{\phi}(l) & \phi(u)\\
\phi(v)^* & \pi^R_{\phi}(r)
\end{matrix}
\right), \ \ \ u,v\in \cl U, r\in \cl R_{\cl U}, l\in \cl L_{\cl U}$$
(see \cite[Proposition 2.1]{h} for details).

We recall the maximal ternary tensor product defined in 
\cite[Section 5]{kr}. 
If $H$ is a Hilbert space, and 
$\phi : \cl U\to \cl B(H)$ and $\psi : \cl V\to \cl B(H)$ are ternary morphisms, 
the pair $(\phi,\psi)$ is called \emph{commuting} if
$$\phi(u)\psi(v) = \psi(v)\phi(u) \ \mbox{ and } \ 
\phi(u)\psi(v)^* = \psi(v)^*\phi(u), \ \ 
u\in \cl U, v\in \cl V;$$
given a commuting pair $(\phi,\psi)$, let
$\phi\cdot\psi : \cl U\otimes \cl V\to \cl B(H)$ be the linear map, defined by letting
$$(\phi\cdot\psi)\left(u\ten v\right) = \phi(u)\psi(v), \ \ \  u\in\cl U, v\in \cl V.$$
The space $\cl U\otimes_{\rm tmax}\cl V$ \cite{kr} is the completion of $\cl U\ten\cl V$ 
with respect to the norm 
$$\|w\|_{\rm tmax} := \sup\left\{\|(\phi\cdot\psi)(w)\| : (\phi,\psi) \mbox{ commuting pair}\right\}.$$

If $(\phi,\psi)$ is a commuting pair then 
the induced representations $\tilde{\pi}_{\phi}$ and $\tilde{\pi}_{\psi}$ have commuting ranges.
Conversely, 
if $\pi : \cl D_{\cl U}\to \cl B(L)$ and $\rho : \cl D_{\cl V}\to \cl B(L)$ are 
*-representations with commuting ranges, their restrictions $\phi$ and $\psi$ 
to $\cl U$ and $\cl V$, respectively, are ternary morphisms that form a commuting pair; 
thus, we have a completely isometric embedding
\begin{equation}\label{inclusion}
\cl U \otimes_{\rm tmax} \cl V\subseteq \cl D_{\cl U}\otimes_{\max} \cl D_{\cl V}.
\end{equation}

Given TRO's $\cl U$ and $\cl V$, we equip the algebraic tensor product $\cl U\ten\cl V$ 
with the ternary product, given by 
\begin{equation}\label{eq_terpwhy}
(u_1\hspace{-0.05cm}\otimes \hspace{-0.05cm} v_1)
(u_2\hspace{-0.05cm}\otimes \hspace{-0.05cm}v_2)^*
(u_3\hspace{-0.05cm}\otimes \hspace{-0.05cm}v_3) 
:= u_1u_2^*u_3 \hspace{-0.05cm} \otimes \hspace{-0.05cm} v_1v_2^*v_3, \ 
u_i\in \cl U, v_i\in \cl V, i = 1,2,3
\end{equation}
(in order to make sense of (\ref{eq_terpwhy}), we employ the embedding 
$\cl U\otimes\cl V\subseteq \cl D_{\cl U}\otimes\cl D_{\cl V}$ and 
make use of the fact that 
$\cl D_{\cl U}\otimes\cl D_{\cl V}$ is a *-algebra in a canonical fashion \cite[p. 203]{tak1}).
For $w\in\cl U\ten \cl V$, we set
$$\norm{w}_{\rm max}:=\sup\{\norm{\theta(w)} \ : \ \theta : \cl U\ten\cl V\rightarrow\cl B(H,K)
\mbox{ ternary morphism}\}$$
(note that the Hilbert spaces $H$ and $K$ in the displayed formula are also allowed to vary).
If $\phi$ and $\psi$ are ternary morphisms of $\cl U$ and $\cl V$, respectively, 
that form a commuting pair, then $\phi\cdot\psi$ is a ternary morphism of $\cl U\otimes\cl V$;
thus, $\norm{\cdot}_{\rm tmax}\leq\norm{\cdot}_{\rm max}$ and, in 
particular, $\norm{\cdot}_{\rm max}$ is a norm. 
We let $\cl U\ten_{\max} \cl V$ be the completion of 
$\cl U\ten \cl V$ under $\norm{\cdot}_{\rm max}$, 
and note that 
$\cl U\ten_{\max}\cl V$ is a TRO with a canonical isometric ternary morphism 
$$\cl U\ten_{\rm max}\hspace{-0.05cm}\cl V\hookrightarrow
\bigoplus \left\{\cl B(H_\theta,K_\theta) \ : \ 
\theta:\cl U\hspace{-0.05cm} \ten \hspace{-0.05cm} \cl V\to\cl B(H_\theta,\hspace{-0.05cm} K_\theta)  \ \textnormal{ternary morphism}\right\}\hspace{-0.1cm}.$$
(We note the set theory abuse 
which can be avoided by picking one element from each unitary equivalence representation 
class on Hilbert spaces of
dimension
at most the maximum of the cardinalities of $\cl U$ and $\cl V$; 
similar arrangements can be made in Section \ref{s_qv} below.) 

\begin{defn}\label{d_semic}
Let $\cl U$ and $\cl V$ be TRO's. We say that the ternary morphisms $\phi:\cl U\rightarrow\cl B(K,L)$ and $\psi:\cl V\to\cl B(H,K)$ \textit{semi-commute} if $\pi^R_{\phi}(\cl R_{\cl U})$ and $\pi^L_{\psi}(\cl L_{\cl V})$ commute in $\cl B(K)$.
\end{defn}

We will say that the TRO's $\cl U\subseteq \cl B(K,L)$ 
and $\cl V\subseteq \cl B(H,K)$ semi-commute if their identity representations do so. 
In the next example, we note that, not surprisingly, 
semi-commutation is strictly weaker than commutation.

\begin{example}\label{ex:semi-commuting} 
Let $G$ be a discrete group, and let $\ell_\infty(G)$ denote the C*-algebra of all bounded (complex-valued) functions on $G$, viewed as multiplication operators in $\cl B(\ell_2(G))$. Let 
$\{\delta_s\}_{s\in G}$ be the canonical basis of $\ell_2(G)$ and 
$\lm$ denote the left regular representation of $G$; thus, $\lm(s)\delta_t =\delta_{st}$, $s,t\in G$. Fix $s,t\in G$ with $s\neq t^{-1}$. Then $\cl U:=\lm(s)\ell_\infty(G)$ and $\cl V:=\ell_\infty(G)\lm(t)$ are concrete TRO's in $\cl B(\ell_2(G))$ with $\cl R_{\cl U}=\ell_\infty(G)=\cl L_{\cl V}$, so $\cl U$ and $\cl V$ semi-commute. However they do not commute, since 
$$\lm(s)\delta_e \delta_e^* \cdot \delta_e \delta_e^* \lm(t) = \delta_s \delta_t^*\neq 0 = \delta_e \delta_e^* \lm(t)\cdot \lm(s) \delta_e \delta_e^*.$$
\end{example}

Notwithstanding Example \ref{ex:semi-commuting}, we now show that  $\norm{\cdot}_{\max}$ coincides with $\norm{\cdot}_{\rm tmax}$. 
We will require the following instance of Arveson's Commutant Lifting 
\cite[Theorem 1.3.1]{a}.
As usual,
we write $\cl R'$ for the  commutant of $\cl R\subseteq \cl B(H)$.

\begin{lemma}\label{lifting} 
Let $\cl U$ and $\cl V$ be TRO's, and let $\phi:\cl U\rightarrow\cl B(K,L)$ and $\psi:\cl V\to\cl B(H,K)$ be semi-commuting ternary morphisms with $\psi$ left non-degenerate.
Then there exists a $*$-homomorphism $\rho:\cl R_{\phi(\cl U)}\to \cl R_{\psi(\cl V)}'\subseteq\cl B(H)$,
such that
$$\rho(b)\psi(v)^* = \psi(v)^*b, \ \ \ b\in \cl R_{\phi(\cl U)}, \ v\in\cl V.$$
\end{lemma}

\begin{proof}
We identify $\cl R_{\phi(\cl U)}$ with the concrete C*-algebra 
$\overline{{\rm span}(\phi(\cl U)^*\phi(\cl U))}^{\|\cdot\|}$, acting on the Hilbert space $K$.
Let $b\in \cl R_{\phi(\cl U)}$, $v_k\in \cl V$ and $\xi_k\in K$, $k=1,\ldots, n$. Write $V=[\pi^L_{\psi}(v_lv_k^*)]_{k,l=1}^n\in M_n(\cl B(K))$ and $\xi=[\xi_k]_{k=1}^n\in \mathbb C^n\otimes K$. We have
\begin{align*}
\left\|\sum_{k=1}^n\psi(v_k)^*b\xi_k\right\|^2
&= 
\sum_{k,l=1}^n\langle \psi(v_k)^*b\xi_k,\psi(v_l)^*b\xi_l\rangle
=
\sum_{k,l=1}^n\langle \pi^L_{\psi}(v_lv_k^*)b\xi_k, b\xi_l\rangle\\
&=
\langle V(I_n \otimes b)\xi,(I_n \otimes b)\xi\rangle
=
\|V^{1/2}(I_n\otimes b)\xi\|^2\\
&=
\|(I_n\otimes b)V^{1/2}\xi\|^2\leq\|b\|^2\langle V\xi,\xi\rangle\\
&=
\|b\|^2\sum_{k,l=1}^n\langle \pi^L_{\psi}(v_lv_k^*)\xi_k,\xi_l\rangle
= \|b\|^2 \left\|\sum_{k=1}^n\psi(v_k)^*\xi_k \right\|^2.
\end{align*}
Hence, the mapping 
$\tilde b: \sum_{k=1}^n\psi(v_k)^*\xi_k\mapsto\sum_{k=1}^n\psi(v_k)^*b\xi_k$ is a (well-defined)
bounded linear operator on ${\rm span}(\psi(\cl V)^*K)$ with $\|\tilde b\|\leq\|b\|$. Since we have 
$\overline{{\rm span}(\psi(\cl V)^*K)} = H$,
it extends to an element (denoted in the same way) 
$\tilde b\in\cl B(H)$ satisfying $\tilde b \psi(v)^*=\psi(v)^*b$ for any $v\in \cl V$. 
Define $\rho:\cl R_{\phi(\cl U)}\mapsto\cl B(H)$ by letting 
$\rho(b) = \tilde b$; clearly, $\rho$ is linear. It is also multiplicative as, for $\xi\in K$ and $v\in \cl V$, 
we have 
\begin{align*}
\rho(b_1b_2)\psi(v)^*\xi&=\psi(v)^*b_1b_2\xi=\psi(v)^*b_1(b_2\xi)=\rho(b_1)\psi(v)^*b_2\xi\\
&=\rho(b_1)\rho(b_2)\psi(v)^*\xi.
\end{align*}
Moreover, for $b\in \cl R_{\cl U}$, $v, w\in\cl V$ and $\xi, \eta\in K$, we have
\begin{eqnarray*}
& & \langle \rho(b)^*\psi(v)^*\xi,\psi(w)^*\eta\rangle
= \langle \psi(v)^*\xi, \psi(w)^*b\eta\rangle=\langle \xi, \pi^L_\psi(vw^*)b\eta\rangle\\
& = & \langle \xi, b\pi^L_\psi(vw^*)\eta\rangle=\langle \psi(v)^*b^*\xi,\psi(w)^*\eta\rangle
=\langle\rho(b^*)\psi(v)^*\xi,\psi(w)^*\eta\rangle,
\end{eqnarray*}
showing that $\rho(b)^*=\rho(b^*)$; thus, $\rho$ is a $*$-homomorphism. 
For $v,w\in \cl V$ and $b\in \cl R_{\phi(\cl U)}$, we further have
\begin{align*}
\rho(b)\pi^R_\psi(v^*w)
&= \psi(v)^*b\psi(w) = \psi(v)^*(\psi(w)^*b^*)^* = \psi(v)^*(\rho(b^*)\psi(w)^*)^*\\
&= \pi^R_\psi(v^*w)\rho(b),
\end{align*}
and hence $\rho(b)\in \cl R_{\cl \psi(\cl V)}'$ for every $b\in \cl R_{\phi(\cl U)}$.
The proof is complete.
\end{proof}

In the following proposition, the definition of the map $\phi\cdot\psi$ is 
extended naturally to the case where $\phi:\cl U\rightarrow\cl B(K,L)$ and $\psi:\cl V\to\cl B(H,K)$, 
for possibly different Hilbert spaces $H$, $K$ and $L$. 

\begin{proposition}\label{p:representations} 
Let $\cl U$ and $\cl V$ be TRO's. Then $\theta:\cl U\ten\cl V\rightarrow\cl B(H,L)$ is a ternary morphism if and only if there exist 
a Hilbert space $K$, and 
semi-commuting ternary morphisms $\phi:\cl U\rightarrow\cl B(K,L)$ and $\psi:\cl V\to\cl B(H,K)$, 
such that $\theta = \phi\cdot\psi$.
\end{proposition}

\begin{proof} 
Let $\phi : \cl U\rightarrow\cl B(K,L)$ and $\psi:\cl V\to\cl B(H,K)$ be semi-commuting ternary morphisms. 
If $u_{i,j}\in \cl U$ and $v_{i,j}\in \cl V$, 
$j = 1,\dots,p$, $i = 1,2,3$, and $x_i = \sum_{j=1}^p u_{i,j}\ten v_{i,j}$, then 
\begin{align*}
& 
(\phi\cdot\psi) (x_1x_2^*x_3) 
=\sum_{j,k,l=1}^p \phi(u_{1,j}u_{2,k}^*u_{3,l})\psi(v_{1,j}v_{2,k}^*v_{3,l})\\
&=\sum_{j,k,l=1}^p \phi(u_{1,j})\pi^R_\phi(u_{2,k}^*u_{3,l})\pi^L_\psi(v_{1,j}v_{2,k}^*)\psi(v_{3,l})\\
&=\sum_{j,k,l=1}^p\phi(u_{1,j})\pi^L_\psi(v_{1,j}v_{2,k}^*)\pi^R_\phi(u_{2,k}^*u_{3,l})\psi(v_{3,l})\\
&=\sum_{j,k,l=1}^p (\phi\cdot\psi)(u_{1,j}\ten v_{1,j})(\phi\cdot\psi)(u_{2,k}\ten v_{2,k})^* (\phi\cdot\psi)(u_{3,l}\ten v_{3,l})\\
&= (\phi\cdot\psi)(x_1)
(\phi\cdot\psi)(x_2)^* (\phi\cdot\psi)(x_3).
\end{align*}
Hence, $\phi\cdot\psi$ is a ternary morphism $\cl U\ten\cl V\to\cl B(H,L)$.

Conversely, if $\theta:\cl U\ten \cl V\rightarrow\cl B(H,L)$ is a ternary morphism,
then it extends to a ternary morphism of $\cl U\ten_{\max}\cl V$. As in the proof of \cite[Proposition 2.1(iv)]{h}, $\theta $ induces a *-homomorphism $\pi^R_\theta:\cl R_{\cl U}\ten\cl R_{\cl V}\to\cl B(H)$, given by 
\begin{equation}\label{eq_thetate}
\pi^R_{\theta}\bigg(\sum_{j=1}^p u_{j}^*u_{j}'\ten v_{j}^*v_{j}'\bigg)
= \sum_{j=1}^p \theta(u_{j}\ten v_{j})^*\theta(u_{j}'\ten v_{j}').
\end{equation}
Then $\pi^R_\theta = \pi^R_{\cl U}\cdot\pi^R_{\cl V}$ for commuting representations of $\cl R_{\cl U}$ and $\cl R_{\cl V}$ \cite[Theorem 3.2.6]{bo} (for the degenerate case, see \cite[Exercise 11.4]{pisier_intr}).  

Equip the algebraic tensor product $\cl V\otimes H$ with the sesquilinear pairing given by
\begin{equation}\label{e:module}
\la v_1\ten_{\cl V}\xi_1,v_2\ten_{\cl V}\xi_2\ra
= \la\pi^R_{\cl V}(v_2^*v_1)\xi_1,\xi_2\ra, \ \ \ v_1,v_2\in\cl V, \ \xi_1,\xi_2\in H;
\end{equation}
after passing to the quotient space with respect to its kernel $\cl N$, we obtain a 
pre-Hilbert space, whose completion will be denoted by 
$\cl V\ten_{\cl V} H$. 
Define $\psi:\cl V\rightarrow\cl B(H,\cl V\ten_{\cl V} H)$ by letting
$\psi(v)\xi=v\ten_{\cl V}\xi$, $\xi\in H$.
It is straightforward that $\psi$
is a ternary morphism; further, by (\ref{e:module}), 
for $\xi_i\in H$ and $v_i\in \cl V$, $i = 1,2$, we have 
$$\langle\pi^R_{\psi}(v_2^*v_1) \xi_1,\xi_2\rangle = 
\langle\psi(v_1)\xi_1, \psi(v_2)\xi_2\rangle = \la\pi^R_{\cl V}(v_2^*v_1)\xi_1,\xi_2\ra,$$
implying that $\pi^R_{\psi} = \pi^R_{\cl V}$.
Thus, if $v,w\in \cl V$ and $\xi,\eta\in H$, then 
\begin{eqnarray*}
\langle(v_2 v_1^*) v\otimes_{\cl V} \xi,w\otimes_{\cl V} \eta\rangle
& = & 
\langle \pi^R_{\cl V}(w^*v_2 v_1^* v) \xi,\eta\rangle\\
& = & 
\langle \psi(w)^*\psi(v_2) \psi(v_1)^* \psi(v) \xi,\eta\rangle\\
& = & 
\langle \pi^L_{\psi}(v_2v_1^*) (v\otimes_{\cl V} \xi), w \otimes _{\cl V}\eta\rangle, 
\end{eqnarray*}
showing that 
\begin{equation}\pi^L_\psi(a)(v\ten_{\cl V}\xi)=av\ten_{\cl V}\xi, \ \ \ a\in\cl L_{\cl V}, \ v\in \cl V, \ \xi\in H.\label{eq:LVH}
\end{equation}
For $b\in \cl R_{\cl U}$, let 
$\tilde{\pi}^R_{\cl U}(b) : \cl V\otimes H \to \cl V\otimes_{\cl V} H$ be the 
linear map, given by 
$$\tilde{\pi}^R_{\cl U}(b)\bigg(v \otimes \xi\bigg) = 
v \ten_{\cl V} \pi^R_{\cl U}(b)\xi, \ \ \ v\in \cl V, \xi\in H.$$ 
Now let $\sum_{i=1}^k v_i\otimes\xi_i\in \cl V\otimes H$, and set 
$\xi := [\xi_i]_{i=1}^k$ and $\tilde{b} := I_k\otimes \pi^R_{\cl U}(b)$. We have that
\begin{eqnarray*}
& & \left\|\tilde{\pi}^R_{\cl U}(b)\bigg(\sum_{i=1}^kv_i \otimes \xi_i\bigg)\right\|^2
= 
\sum_{i,j=1}^k \la \pi^R_{\cl V}(v_j^*v_i)\pi^R_{\cl U}(b)\xi_i,\pi^R_{\cl U}(b)\xi_j\ra\\
& = &    
\left\|[\pi^R_{\cl V}(v_j^*v_i)]_{i,j}^{1/2}\tilde{b}\xi\right\|^2
= 
\left\|\tilde{b}[\pi^R_{\cl V}(v_j^*v_i)]_{i,j}^{1/2}\xi\right\|^2\\
& \leq & 
\|\pi^R_{\cl U}(b)\|^2\left\| [\pi^R_{\cl V}(v_j^*v_i)]_{i,j}^{1/2}\xi\right\|^2
=
\|\pi^R_{\cl U}(b)\|^2\sum_{i,j=1}^k\langle\pi_{\cl V}^R(v_j^*v_i)\xi_i,\xi_j\rangle\\
& = & 
\|\pi^R_{\cl U}(b)\|^2 
\left\langle\sum_{i=1}^kv_i\otimes_{\cl V}\xi_i,\sum_{j=1}^kv_j\otimes_{\cl V}\xi_j\right\rangle
=
\|\pi^R_{\cl U}(b)\|^2\left\|\sum_{i=1}^kv_i\otimes_{\cl V}\xi_i\right\|^2;
\end{eqnarray*}
thus, the map $\tilde{\pi}^R_{\cl U}(b)$ induces a bounded 
linear map (denoted in the same way)
$\tilde{\pi}^R_{\cl U}(b)\in \cl B(\cl V\ten_{\cl V} H)$. 
Since the ranges of $\pi^R_{\cl U}$ and $\pi^R_{\cl V}$ commute, 
it follows that the resulting map
$\tilde{\pi}^R_{\cl U}:\cl R_{\cl U}\to\cl B(\cl V\ten_{\cl V} H)$ is a $*$-homomorphism. We include the details of $*$-preservation; the multiplicativity follows similarly. 
For $b\in \cl R_{\cl U}$, $v_1,v_2\in\cl V$ and $\xi_1,\xi_2\in H$, we have
\begin{eqnarray*}
& & 
\la\tilde{\pi}^R_{\cl U}(b^*)(v_1\ten_{\cl V}\xi_1),v_2\ten_{\cl V}\xi_2\ra
=\la v_1\ten_{\cl V}\pi^R_{\cl U}(b^*)\xi_1,v_2\ten_{\cl V}\xi_2\ra\\
& = & 
\la \pi^R_{\cl V}(v_2^*v_1)\pi^R_{\cl U}(b^*)\xi_1,\xi_2\ra
=
\la \pi^R_{\cl V}(v_2^*v_1)\xi_1,\pi^R_{\cl U}(b)\xi_2\ra\\
& = & 
\la v_1\ten_{\cl V}\xi_1,v_2\ten_{\cl V}\pi^R_{\cl U}(b)\xi_2\ra
=
\la v_1\ten_{\cl V}\xi_1, \tilde{\pi}^R_{\cl U}(b)(v_2\ten_{\cl V}\xi_2)\ra\\
&= & 
\la \tilde{\pi}^R_{\cl U}(b)^*(v_1\ten_{\cl V}\xi_1), v_2\ten_{\cl V}\xi_2\ra.
\end{eqnarray*}

Now, let 
$\phi:\cl U\rightarrow\cl B(\cl V\ten_{\cl V} H,\cl U\ten_{\cl U}(\cl V\ten_{\cl V} H))$ be given by 
$$\phi(u)(v\ten_{\cl V}\xi)=u\ten_{\cl U}(v\ten_{\cl V}\xi), \ \ \ u\in \cl U, \ v\in\cl V, \ \xi\in H,$$
where the Hilbertian tensor product $\ten_{\cl U}$ is defined relative to 
$\tilde{\pi}^R_{\cl U}$, similarly to (\ref{e:module}). 
Then $\phi$ is a ternary morphism with $\pi^R_{\phi} = \tilde{\pi}^R_{\cl U}$. 
By~\eqref{eq:LVH}, the algebras $\tilde{\pi}^R_{\cl U}(\cl R_{\cl U})$ and $\pi^L_{\psi}(\cl L_{\cl V})$ commute 
(in $\cl B(\cl V\ten_{\cl V}H)$), hence we have that $\phi$ and $\psi$ semi-commute. 
Moreover, for all $u_1,u_2\in \cl U$, $v_1,v_2\in\cl V$ and $\xi_1,\xi_2\in H$, using (\ref{eq_thetate}), we have
\begin{eqnarray*}
& & \la\phi(u_1)\psi(v_1)\xi_1,\phi(u_2)\psi(v_2)\xi_2\ra
=\la u_1\ten_{\cl U}(v_1\ten_{\cl V}\xi_1), u_2\ten_{\cl U}(v_2\ten_{\cl V}\xi_2)\ra\\
& = & 
\la\tilde{\pi}^R_{\cl U}(u_2^*u_1)(v_1\ten_{\cl V}\xi_1),v_2\ten_{\cl V}\xi_2\ra
=\la v_1\ten_{\cl V}\pi^R_{\cl U}(u_2^*u_1)\xi_1,v_2\ten_{\cl V}\xi_2\ra\\
& = & \la\pi^R_{\cl V}(v_2^*v_1)\pi^R_{\cl U}(u_2^*u_1)\xi_1,\xi_2\ra
=\la\theta(u_1\ten v_1)\xi_1,\theta(u_2\ten v_2)\xi_2\ra. 
\end{eqnarray*}
It follows that the map $W : \cl U \ten_{\cl U} (\cl V\ten_{\cl V} H) \to L$, given by 
$$W\left(\sum_{i=1}^k u_i\ten_{\cl U}v_i\ten_{\cl V}\xi_i \right) = 
\sum_{i=1}^k \theta(u_i\ten v_i)\xi_i,$$
is a well-defined isometry.
Let
$W\circ\phi : \cl U\rightarrow\cl B(\cl V\ten_{\cl V}H,L)$ 
be defined by letting $(W\circ\phi)(u) = W\phi(u)$. Then $W\circ\phi$ 
is a ternary morphism which semi-commutes with $\psi$, and 
it follows from the definition of $W$ that 
$\theta = (W\circ\phi)\cdot\psi$ on $\cl U\ten\cl V$. The proof is complete.
\end{proof}

The proof of Proposition \ref{p:representations} implies the following 
corollary, which we isolate for further reference.

\begin{corollary}\label{c_maxtose}
Let $\cl U$ and $\cl V$ be TRO's. If $\pi : \cl R_{\cl U}\to \cl B(H)$ and 
$\rho : \cl R_{\cl V}\to \cl B(H)$ are *-representations with commuting ranges, 
then there exist Hilbert spaces $K$ and $L$, and 
semi-commuting ternary morphisms $\phi : \cl U\to \cl B(K,L)$ and 
$\psi : \cl V\to \cl B(H,K)$, such that $\pi\cdot\rho = \pi^R_{\phi\cdot\psi}$.
\end{corollary}

\begin{lemma}\label{c:maxnorm} 
Let $\cl U$ and $\cl V$ be TRO's. For $w\in \cl U\ten \cl V$,
$$\norm{w}_{\max}=\sup\{\norm{(\phi\cdot\psi)(w)} :
(\phi,\psi) \mbox{ semi-commuting with } \psi \mbox{ left non-deg.}\}.$$
\end{lemma}

\begin{proof} 
Let $w\in \cl U\ten \cl V$. 
Denoting the right hand side of the identity in the statement by $r$, we have, by Proposition \ref{p:representations}, 
that 
$r\leq \norm{w}_{\max}$. For the reverse inequality, suppose that 
$\phi : \cl U\rightarrow\cl B(K,L)$ and $\psi:\cl V\to\cl B(H,K)$ are semi-commuting ternary morphisms.
Let $H' = \overline{{\rm span}(\psi(\cl V)^* K)}$
and $\psi' : \cl V\to \cl B(H',K)$ be given by 
$\psi'(v) = \psi(v)|_{H'}$. 
With respect to the decomposition
$H = H'\oplus (H')^{\perp}$, we have that 
$\psi(v) = \left(\psi'(v), 0\right)$, $v\in \cl V$, and hence 
$\psi'$ is a left non-degenerate ternary morphism. Since $\psi(v_1)\psi(v_2)^* = \psi'(v_1)\psi'(v_2)^*$, we have that 
$\phi$ and $\psi'$ are semi-commuting. 
For $U \in \phi(\cl U)$ and 
$V = \left(V_1, 0\right)\in \psi(\cl V)$ with respect to the decomposition 
$H = H'\oplus (H')^{\perp}$, we have that 
$UV = (UV_1, 0)$. It follows that 
$$\|(\phi\cdot\psi')(w)\| = \|(\phi\cdot\psi)(w)|_{H'}\| = \|(\phi\cdot\psi)(w)\|,$$
showing (by Proposition~\ref{p:representations})
that $\norm{w}_{\max}\leq r$.
\end{proof}

We are now in a position to prove the main result of this subsection. 

\begin{theorem}\label{t:maxnorm} 
Let $\cl U$ and $\cl V$ be TRO's. Then 
$\norm{\cdot}_{\max}=\norm{\cdot}_{\rm tmax}$ on $\cl U\ten\cl V$.
\end{theorem}

\begin{proof} 
The inclusion $b\mapsto \left(\smallmatrix 0 & 0\\0 & b\endsmallmatrix\right)$
of $\cl R_{\cl U}$ into $\cl D_{\cl U}$
is complemented by the contractive completely positive map
$$\cl D_{\cl U}\ni \begin{bmatrix}a & u\\v^* & b\end{bmatrix}\mapsto \begin{bmatrix}0 & 0\\0 & 1\end{bmatrix}\begin{bmatrix}a & u\\v^* & b\end{bmatrix}\begin{bmatrix}0 & 0\\0 & 1\end{bmatrix}\in\cl D_{\cl U},$$
where $1$ is the identity in the multiplier algebra $M(\cl R_{\cl U})$ (one can assume without loss of generality that $\cl U$ is a right $M(\cl R_{\cl U})$-module, see \cite[\S2]{kr}). By \cite[Proposition 3.6.6]{bo}, we have a
canonical inclusion
$\cl R_{\cl U}\ten_{\rm max}\cl R_{\cl V}\subseteq \cl D_{\cl U}\ten_{\rm max}\cl R_{\cl V}.$
By symmetry, 
$$\cl R_{\cl U}\ten_{\rm max}\cl R_{\cl V}\subseteq \cl D_{\cl U}\ten_{\rm max}\cl D_{\cl V}.$$

Now, suppose that $\phi:\cl U\to \cl B(K,L)$  and $\psi:\cl V\to \cl B(H,K)$ are semi-commuting ternary morphisms with $\overline{{\rm span}(\psi(\cl V)^*K)} = H$. By Lemma \ref{lifting} there exists a $*$-homomorphism $\rho:\cl R_{\phi(\cl U)}\rightarrow \cl R_{\psi(\cl V)}'\subseteq\cl B(H)$ satisfying $\rho(b)\psi(v)^*=\psi(v)^*b$ for all $b\in \cl R_{\phi(\cl U)}$ and all $v\in \cl V$. 
If $b\in \cl R_{\phi(\cl U)}$ and $v_1,v_2\in \cl V$, we have 
\begin{eqnarray*}
\rho(b)\pi^R_{\psi}(v_1^*v_2)
& = & 
\rho(b)\psi(v_1)^*\psi(v_2) = \psi(v_1)^*b\psi(v_2) = 
\psi(v_1)^*(\psi(v_2)^*b^*)^*\\
& = & 
\psi(v_1)^*(\rho(b^*)\psi(v_2)^*)^*
= 
\psi(v_1)^*\psi(v_2)\rho(b) = \pi^R_{\psi}(v_1^*v_2)\rho(b),
\end{eqnarray*}
showing that the $*$-homomorphisms
$\rho\circ\pi^R_{\phi}:\cl R_{\cl U}\to\cl B(H)$ and $\pi^R_{\psi}:\cl R_{\cl V}\to\cl B(H)$ have 
commuting ranges. Thus, using (\ref{inclusion}), for $w=\sum_{i=1}^n u_i\ten v_i\in \cl U\ten\cl V$
we have 
\begin{align*}
\norm{(\phi\cdot\psi)(w)}
&= \left\|\sum_{i=1}^n\phi(u_i)\psi(v_i)\right\|
= \left\|\sum_{i,j=1}^n\psi(v_j)^*\phi(u_j)^*\phi(u_i)\psi(v_i)\right\|^{1/2}\\
&= \left\|\sum_{i,j=1}^n\rho(\phi(u_j)^*\phi(u_i))\psi(v_j)^*\psi(v_i)\right\|^{1/2}\\
&= \left\|\sum_{i,j=1}^n\rho(\pi_{\phi}^R(u_j^*u_i))\pi^R_{\psi}(v_j^*v_i)\right\|^{1/2}\\
&\leq \left\|\sum_{i,j=1}^n u_j^*u_i\ten v_j^*v_i\right\|_{\cl R_{\cl U}\ten_{\rm max}\cl R_{\cl V}}^{1/2}
\hspace{-0.2cm} = \left\|\sum_{i,j=1}^n u_j^*u_i\ten v_j^*v_i\right\|_{\cl D_{\cl U}\ten_{\rm max}\cl D_{\cl V}}^{1/2}\\
&= \left\|\sum_{i=1}^n u_i\ten v_i\right\|_{\cl D_{\cl U}\ten_{\rm max}\cl D_{\cl V}}
= \left\|\sum_{i=1}^n u_i\ten v_i \right\|_{\cl U\ten_{\rm  tmax}\cl V}.
\end{align*}
By Lemma \ref{c:maxnorm}, it follows that $\norm{w}_{\max}\leq\norm{w}_{\rm tmax}$. Since the reverse inequality is 
trivial, the proof is complete.
\end{proof}

\begin{corollary}\label{c:h} 
  Let $\cl U$ and $\cl V$ be TRO's. The identity map
  on ${\cl U}\ten{\cl V}$
  extends to a complete contraction $\cl U\ten_{\rm h} \cl V\rightarrow\cl U\ten_{\rm tmax}\cl V$.
\end{corollary}

\begin{proof} By \cite[Theorem 5.13]{pisier_intr} and 
the injectivity of the Haagerup tensor product \cite[Proposition 9.2.5]{er}, there exist 
canonical completely 
isometric inclusions
$$\cl U\otimes_{\rm h}\cl V\subseteq \cl D_{\cl U}\otimes_{\rm h}\cl D_V
\subseteq \cl D_{\cl U}\ast\cl D_V,$$
where the last of the three spaces is the C*-algebraic 
free product of $\cl D_{\cl U}$ and $\cl D_{\cl V}$. 
Ternary morphisms 
$\phi$ and $\psi$ of $\cl U$ and $\cl V$
correspond precisely to *-representations $\pi$ and $\rho$ of $\cl D_{\cl U}$ and
$\cl D_{\cl V}$; thus, if $[w_{i,j}]\in M_n(\cl U\otimes \cl V)$ then
$$\|[w_{i,j}]\|_{\rm h} = 
\sup\left\{\|(\phi\cdot\psi)^{(n)}([w_{i,j}])\| : \phi,\psi \mbox{ ternary morphisms}\right\}.$$
By (\ref{eq_TROopsp}) and Theorem \ref{t:maxnorm}, the identity map 
$\cl U\otimes_{\rm h}\cl V \to \cl U\otimes_{\rm tmax}\cl V$ is completely contractive. 
\end{proof}


\section{Hypergraph values via resources}\label{s_hyper}

In this section we develop a general resource theory that ascribes operator space structures on the trace class through families of block operator isometries. Such families correspond to families of quantum channels, which, in the context of quantum games, are families of strategies 
that include local operations and shared randomness. 
The maximal success probability when using such strategies is then governed 
by the associated operator space norms on canonical game tensors. 

We fix non-empty finite sets $X$ and $A$.
We begin by identifying the trace class inside the universal TRO of an isometry.


\subsection{The operator space \texorpdfstring{$\cl S_1$}{S1} revisited}
\label{s_S1}

A \emph{block operator isometry (over $(X,A)$)} is an isometry $U=(U_{a,x})_{a,x}\in M_{A,X}\otimes \cl B(H,K)=\cl B(H^X,K^A)$, where $H$ and $K$ are Hilbert spaces.
Following \cite{tt-QNS}, let $\cl V_{X,A}$ be the universal TRO 
of an isometry $(u_{a,x})_{a,x}$; 
thus, $\cl V_{X,A}$ is generated, as a TRO, by a family $\{u_{a,x} : x\in X, a\in A\}$, and 
has the property that 
its ternary representations $\theta : \cl V_{X,A}\to \cl B(H,K)$ correspond 
to block operator isometries 
$U = (U_{a,x})_{a,x}\in M_{A,X}\otimes \cl B(H,K)$, via the assignments $\theta(u_{a,x}) = U_{a,x}$, $x\in X$, $a\in A$. 
In the latter case, we write $\theta = \theta_U$. 
Let $\cl C_{X,A}$ be the right C*-algebra corresponding to $\cl V_{X,A}$; 
thus, $\cl C_{X,A}$ is generated, as a C*-algebra, by the elements 
$e_{x,x',a,a'} := u_{a,x}^*u_{a',x'}$, $x,x'\in X$, $a,a'\in A$.
More concretely, if  
$\theta : \cl V_{X,A}\to \cl B(H,K)$ is a faithful ternary representation then 
$\cl C_{X,A}$ is *-isomorphic to the C*-algebra $\overline{{\rm span}}(\theta(\cl V_{X,A})^*\theta(\cl V_{X,A}))$. 
If $U \in M_{A,X}\otimes \cl B(H,K)$ is a block operator isometry, let
$\pi_U : \cl C_{X,A}\to \cl B(H)$ be the unital *-homomorphism  with 
$$\pi_U(e_{x,x',a,a'}) = \theta_U(u_{a,x})^*\theta_U(u_{a',x'}), \ \ \ 
x,x'\in X, a,a'\in A;$$
we note that every unital *-homomorphism of $\cl C_{X,A}$ arises in this way \cite[Lemma 5.1]{tt-QNS}. 

Let 
$$\cl T_{X,A} = {\rm span}\{e_{x,x',a,a'} : x,x'\in X, a,a'\in A\},$$
viewed as an operator subsystem of $\cl C_{X,A}$. 
Since $\cl T_{X,A}$ is finite dimensional, its dual $\cl T_{X,A}^{\rm d}$ 
is an operator system \cite[Corollary 4.5]{CE2}; 
in fact, it was shown in \cite[Proposition 5.5]{tt-QNS} that, up to a complete order isomorphism, 
\begin{equation}\label{eq_dualofTXA}
\tightmath
\cl T_{X,A}^{\rm d} =
\left\{(\lambda_{x,x',a,a'})
\in 
M_{XA} : \exists \ c\in \bb{C} \mbox{ s.t.} 
\sum_{a\in A} \lambda_{x,x',a,a} = \delta_{x,x'} c, \ x,x'\in X\right\},
\end{equation}
considered as an operator subsystem of $M_{XA}$. 

Let 
$$\cl O_{X,A} = {\rm span}\{u_{a,x} : x\in X, a\in A\},$$
considered as an operator subspace of $\cl V_{X,A}$. 
We use the abbreviation
$\cl S_1^{A,X} = \cl S_1(\bb{C}^A,\bb{C}^X)$.
write $\langle \cdot,\cdot\rangle$
for the pairing between 
$\cl S_1^{A,X}$ and $M_{A,X}$, 
given by 
\begin{equation}\label{eq_duSTnot}
\langle S,T\rangle = {\rm Tr}(ST), \ \ \ S\in \cl S_1^{A,X}, \ T\in M_{A,X},
\end{equation}
and equip $\cl S_1^{A,X}$ with the operator space structure that turns the 
identification $\left(\cl S_1^{A,X}\right)^* = M_{A,X}$
into a complete isometry.

\begin{proposition}\label{p_OXAS1}
The linear map $\alpha : \cl S_1^{A,X}\to \cl O_{X,A}$, given by 
$\alpha(\epsilon_{x,a}) = u_{a,x}$, $a\in A$, $x\in X$, is a surjective complete isometry.
\end{proposition}

\begin{proof}
The surjectivity is trivial. To show that $\alpha$ is a complete isometry, 
let $S=(S^{i,j})_{i,j} \in M_n(\cl S_1^{A,X})$,
and 
$F_S\in {\rm CB}(M_{A,X}, M_n)$ given by $F_S(T)=(\langle S^{i,j}, T\rangle)_{i,j}$ and let  
$m\geq n$ vary. 
By the definition of the operator space structure on $\cl S_1^{A,X}$ and \cite[Proposition 2.2.2]{er}, 
\begin{align*}
\|S\| 
  & =    \|F_S\|_{\rm cb}=
\|F_S^{(n)}\|\\
  & =   \sup\{\|(\langle S^{i,j},T^{k,l}\rangle)_{i,j,k,l}\| : T=(T^{k,l})_{k,l}\in M_n(M_{A,X}), \|T\|\leq 1\}\\
  & \leq   \sup\{\|(\langle S^{i,j},T^{k,l}\rangle)_{i,j,k,l}\| : T = (T^{k,l})_{k,l}\in M_{n,m}(M_{A,X}), \|T\|\leq 1\}.
\end{align*}
Write $S = \sum_{a\in A}\sum_{x\in X} S_{x,a}\otimes \epsilon_{x,a}$, 
and $T = \sum_{a\in A}\sum_{x\in X} T_{a,x}\otimes \epsilon_{a,x}$, 
where $S_{x,a}\in M_n$ and $T_{a,x}\in M_{n,m}$, 
$x\in X$, $a\in A$.
Then 
\begin{equation}\label{eq_ijkl}
\left(\langle S^{i,j},T^{k,l}\rangle\right)_{i,j,k,l}
= \sum_{a\in A}\sum_{x\in X}  S_{x,a} \otimes T_{a,x};
\end{equation}
indeed, by linearity, 
it suffices to check the formula in the case where 
there exist $a',x',i',j'$, such that 
$S_{x',a'} = \epsilon_{i',j'}$ and $S_{x,a} = 0$ whenever $(x,a)\neq (x',a')$, and 
there exist $x'',a'',k,l$, such that 
$T_{a'',x''} = \epsilon_{k,l}$ and $T_{a,x} = 0$ whenever $(a,x)\neq (a'',x'')$.
In the latter case, 
$S^{i,j} = \delta_{(i',j'),(i,j)} \epsilon_{x',a'}$ 
and 
$T^{k,l} = \delta_{(k',l'),(k,l)} \epsilon_{a'',x''}$, and all the entries, 
apart from the $(i',j',k',l')$-entry, 
of the matrix 
$\left(\langle S^{i,j},T^{k,l}\rangle\right)_{i,j,k,l}$ are zero, while the $(i',j',k',l')$-entry
equals 
${\rm Tr}(\epsilon_{x',a'} \epsilon_{a'',x''}) 
=  \delta_{a',a''} \delta_{x',x''}$. 
Direct comparison with the right hand side of 
(\ref{eq_ijkl}) shows that the latter identity holds true.

Let $m\in \bb{N}$ with $m > \frac{n|X|}{|A|}$, so that 
${\rm dim}({\bb C}^X\otimes{\mathbb C}^n) < 
{\rm dim}({\bb C}^A\otimes{\mathbb C}^m)$. 
We note that the convex hull of 
the isometries in 
$\cl B({\bb C}^X\otimes{\mathbb C}^n, {\bb C}^A\otimes{\mathbb C}^m)$ 
coincides with its unit ball.
(Indeed, if $p < q$ and 
$T:\mathbb C^p\to\mathbb C^q$ is a contraction, write $T = (T_1 \; 0)^{\rm t}$, where 
$T_1\in M_p$ is a contraction. 
By \cite[\S 3.1, Problem 27]{horn-johnson}, $T_1 = \sum_{i=1}^l \lambda_i U_i$ as a convex combination, 
where $U_i \in M_p$ is a unitary, $i\in [l]$; thus, 
$T = \sum_{i=1}^l \lambda_i (U_i \; 0)^{\rm t}$ as a convex combination, 
where $(U_i \; 0)^{\rm t} : \bb{C}^p\to \bb{C}^q$ is an isometry, $i\in [l]$.)
Denoting by $H$ and $K$
arbitrary Hilbert spaces, we therefore have 
\begin{align*}
\|S\| &
= 
\sup\{\|\textstyle \sum_{a,x} S_{x,a}\otimes U_{a,x}\| : U\in {\cl B}(\bb{C}^X\otimes {\mathbb C}^n,\bb{C}^A\otimes{\mathbb C}^m) \mbox{ isometry}\}\\
  & \leq    
\sup_{H,K}\{\|\textstyle \sum_{a,x} S_{x,a}\otimes U_{a,x}\| : U\in \cl B({\mathbb C}^X\otimes H,{\mathbb C}^A\otimes K) \mbox{ isometry}\}\\
  & =    
\sup_{H,K} \{\|\theta_U^{(n)}(\alpha^{(n)}(S))\| : 
U \in \cl B(\bb{C}^X\otimes H,\bb{C}^A\otimes K) \mbox{ isometry}\}\\
  & =   
\|\alpha^{(n)}(S)\|.
\end{align*}
In summary, $\|S\|\leq \|\alpha^{(n)}(S)\|$.
To see the reverse inequality, we use (\ref{eq_ijkl}) to note that
\begin{align*}
\|\alpha^{(n)}(S)\|
  & =
\sup_{H,K}\{\|\textstyle \sum_{a,x}S_{x,a}\otimes U_{a,x}\| : U\in \cl B({\mathbb C}^X\otimes H,{\mathbb C}^A\otimes K) \mbox{ isometry}\}\\
  & \leq 
    \sup_{H,K}\{\|\textstyle \sum_{a,x}S_{x,a}\otimes T_{a,x}\| : T\in \cl B({\mathbb C}^X\otimes H,{\mathbb C}^A\otimes K), \|T\|\leq 1\}\\&=\|S\|.
  \qedhere
\end{align*}
\end{proof}


\subsection{Value expressions}\label{ss_hyper}

Let $U=(U_{a,x})_{a,x}\in \cl B(H^X,K^A)$ be a block operator isometry over $(X,A)$, where $H$ and $K$ are Hilbert spaces.
We let 
$\phi_U : \cl O_{X,A}\to \cl B(H,K)$ be the restriction of $\theta_{U}$; we have that  
$\phi_U$ is completely contractive.

Given a state $\sigma\in\cl S_1(H)$, the map $\Gamma_{U,\sigma}:M_X\to M_A$ given by
\begin{equation}\label{eq:Gamma-U-sigma}
\Gamma_{U,\sigma}(\epsilon_{x,x'})=\sum_{a,a'\in A}\sigma(U_{a,x}^*U_{a',x'})\epsilon_{a,a'}\end{equation}
is a quantum channel, where as usual, $\sigma(T)=\Tr(\sigma T)$ for $T\in \cl B(H)$.
Moreover, the assignment
\begin{equation}\label{eq_sUsi}
s_{U,\sigma}(e_{x,x',a,a'})=\sigma(U_{a,x}^*U_{a',x'}),\quad x,x'\in X,\;a,a'\in A,
\end{equation}
extends to a well-defined state on $\cl T_{X,A}$ (which
can be easily verified using the identification (\ref{eq_dualofTXA}) of the dual $\cl T_{X,A}^{\rm d}$). 
A standard calculation shows that 
\begin{equation}\label{e:Gammap} \Gamma_{U,\sigma}^{\#}(\rho):=\Gamma_{U,\sigma}\left(\rho^{\rm t}\right)^{\rm t} = 
(\id\ten\tr_K)(U(\rho\ten\sigma)U^*), \ \ \ \rho\in M_X.
\end{equation}

\begin{remark}\label{r_tenso}
\rm
Let $X$, $A$, $Y$ and $B$ be finite sets and 
$U_{a,x}\in \cl B(H,K)$ and $U'_{b,y}\in \cl B(H',K')$, $x\in X$, $a\in A$, $y\in Y$, $b\in B$ be such 
that the block operator matrices $U = (U_{a,x})_{a,x}$ and $U' = (U_{b,y})_{b,y}$ are isometries.
Let, further, $\sigma\in \cl S_1(H)$ and $\sigma'\in \cl S_1(H')$ be states. 
It is straightforward to verify that 
$\Gamma_{U,\sigma} \otimes \Gamma_{U',\sigma'} = \Gamma_{U\otimes U',\sigma\otimes \sigma'}$.
\end{remark}

Given a family $\cl R$ of block operator isometries $U : H^X_U\to K^A_U$ (where the Hilbert spaces 
$H_U$ and $K_U$ are allowed to vary), let 
$$\mathsf{QC}(\cl R) = \{\Gamma_{U,\sigma} : U\in \cl R, \; \sigma\in \cl S_1(H_U) \mbox{ is a normal state}\}.$$
The set $\cl R$ will be called 
a \emph{resource over the pair} $(X,A)$ if it is closed under taking finite direct sums, and $\cl R$ is said to be \emph{separating}
if
the family of states $\{s_{U,\sigma} : U\in \cl R, \; \sigma \mbox{ a state}\}$ on $\cl T_{X,A}$ is separating
for the subset $\{u^*u : u\in \cl O_{X,A}\}$ of $\cl T_{X,A}$.
For a
family $\cl R_0$ of block operator isometries 
we let $\langle\cl R_0\rangle$ be the set of all direct sums of elements of $\cl R_0$. Note that, if $\cl R_0$ is separating, then $\langle \cl R_0\rangle$ is the smallest resource containing $\cl R_0$.

\begin{remark}\label{r_conve}
\rm
For a resource $\cl R$ over $(X,A)$, the set $\mathsf{QC}(\cl R)$ is convex. 
Indeed, suppose that $U,U'\in \cl R$ and let $H$ (resp. $H'$) be the domain of the 
entries of $U$ (resp. $U'$). Let $\sigma$ (resp. $\sigma'$) 
be a normal state on $\cl B(H)$ (resp. $\cl B(H')$). 
Fix $\lambda\in (0,1)$ and let 
$\lambda\sigma \dot{+} (1-\lambda)\sigma'$ be the state on $\cl B(H\oplus H')$, 
given by 
$$(\lambda\sigma \dot{+} (1-\lambda)\sigma')(T\oplus T') = \lambda\sigma(T) + (1-\lambda)\sigma'(T').$$
Since $(U\oplus U')_{a,x} = U_{a,x}\oplus U'_{a,x}$, 
a direct verification shows that 
$$\lambda \Gamma_{U,\sigma} + (1-\lambda) \Gamma_{U',\sigma'} = 
\Gamma_{U\oplus U',\lambda\sigma \dot{+} (1-\lambda)\sigma'}.$$
\end{remark}

\begin{remark}\label{r_resource}
The above notion of resource is related to, but in general distinct from, that which appears in quantum resource theories \cite[Definition 1]{chitambar}. Therein, a family of quantum channels is assigned to suitable pairs of Hilbert spaces and stipulated to satisfy identity and compositional conditions. By combining resources $\cl R$ over the pair $(X,A)$ for all finite sets $X$ and $A$ which admit the appropriate identity and compositional structure, one obtains a quantum resource theory in the sense of \cite[Definition 1]{chitambar} (e.g., the resource theory of local operations and shared randomness (LOSR), see Section \ref{s_qv}). Such resources theories are convex by Remark \ref{r_conve}, and, as we shall now see, induce norms on finite-dimensional trace class operators $\cl S_1(\bb{C}^A,\bb{C}^X)$ which quantify state conversion by resource channels (see, e.g., Corollary \ref{c_convert}).
\end{remark}

Let $\cl R$ be a separating family of block operator isometries over $(X,A)$. 
For $u\in M_n(\cl O_{X,A})$, let 
$$\|u\|^{(n)}_{\cl R} = \sup_{U\in \cl R} \left\|\phi_U^{(n)}(u)\right\|, \ \ \ u\in \cl O_{X,A}.$$
If no confusion arises, we will write $\|u\|_{\cl R}$ in the place of $\|u\|^{(n)}_{\cl R}$.

\begin{remark}\label{r_conve2}
\rm
If $\cl R$ is a separating family of block operator isometries over $(X,A)$ and $u\in M_n(\cl O_{X,A})$, 
then $\|u\|^{(n)}_{\cl R} = \|u\|^{(n)}_{\langle\cl R\rangle}$.
This is immediate from that fact that, if $V,W\in \cl R$, then 
\begin{equation}\label{eq_conve}
\phi^{(n)}_{V\oplus W}(u) = \phi^{(n)}_{V}(u) \oplus  \phi^{(n)}_{W}(u), \ \ \ u\in M_n(\cl O_{X,A}).
\end{equation}
\end{remark}

\begin{proposition}\label{p_opspst}
If $\cl R$ is a resource over $(X,A)$, then 
the family of norms
$(\|\cdot\|_{\cl R}^{(n)})_{n\in \bb{N}}$
is an operator space structure on the vector space $\cl O_{X,A}$. 
\end{proposition}

\begin{proof}
Suppose that $u\in \cl O_{X,A}$ is such that $\|u\|_{\cl R}^{(1)} = 0$, and let 
$U = (U_{a,x})_{a,x}\in \cl R\cap M_{A,X}(\cl B(H,K))$.  
Write $u = \sum_{a\in A}\sum_{x\in X} \lambda_{a,x} u_{a,x}$, where $\lambda_{a,x}\in \bb{C}$, $a\in A$, $x\in X$. 
Then $u^*u\in \cl T_{X,A}$ and, if $\sigma = \zeta\zeta^*$ is a vector state on $\cl B(H)$ then, by (\ref{eq_sUsi}),\begin{eqnarray*}
s_{U,\sigma}(u^*u) 
& = & 
\sum_{x,x'\in X} \sum_{a,a'\in A} \bar{\lambda}_{a,x}\lambda_{a',x'}s_{{U,\sigma}}(e_{x,x',a,a'})\\
& = & 
\sum_{x,x'\in X} \sum_{a,a'\in A} \bar{\lambda}_{a,x}\lambda_{a',x'}\, \sigma( U_{a,x}^*U_{a',x'})\\
& = & 
\sum_{x,x'\in X} \sum_{a,a'\in A} \bar{\lambda}_{a,x}\lambda_{a',x'} \langle U_{a',x'}\zeta,U_{a,x}\zeta\rangle\\
& = & 
\left\| \sum_{x\in X} \sum_{a\in A} \lambda_{a,x} U_{a,x} \zeta\right\|^2
= 
\left\|\phi_U(u)\zeta\right\|^2=0.
\end{eqnarray*}
It follows that $u^*u = 0$ in $\cl C_{X,A}$, that is, $u = 0$ in $\cl O_{X,A}$, and thus $\|\cdot\|_{\cl R}$ is a norm on $\cl O_{X,A}$.

For $u\in M_n(\cl O_{X,A})$ and $\alpha,\beta\in M_{m,n}$, we have 
\begin{align*}
\|\alpha \cdot u \cdot \beta^*\|_{\cl R}^{(m)} 
& = 
\sup_{U\in \cl R} \|\phi_U^{(m)}(\alpha \cdot u \cdot \beta^*)\|
= 
\sup_{U\in \cl R} \|\alpha\cdot \phi_U^{(n)}(u)\cdot \beta^*\|^{(n)}_{\cl R}\\
& \leq 
\|\alpha\| \|\beta^*\| \sup_{U\in \cl R} \|\phi_U^{(n)}(u)\|^{(n)}_{\cl R}
= 
\|\alpha\| \|\beta^*\| \|u\|_{\cl R}^{(n)}.
\end{align*}

Finally, for $u\in M_n(\cl O_{X,A})$ and $v\in M_m(\cl O_{X,A})$, we have 
\begin{align*}
\|u \oplus v\|_{\cl R}^{(n + m)} 
& =  
\sup_{U\in \cl R} \|\phi_U^{(n+m)}(u \oplus v)\|
= 
\sup_{U\in \cl R} \|\phi_U^{(n)}(u) \oplus \phi_U^{(m)}(v)\|\\
& = 
\sup_{U\in \cl R} \max\{\|\phi_U^{(n)}(u)\|, \|\phi_U^{(m)}(v)\|\}\\
& \leq 
\max\left\{\sup_{V\in \cl R} \|\phi_V^{(n)}(u)\|, \sup_{W\in \cl R} \|\phi_W^{(m)}(v)\|\right\}\\
& = 
\max\left\{\|u\|_{\cl R}^{(n)},\|v\|_{\cl R}^{(m)}\right\}.
\end{align*}
On the other hand, note that if $V,W\in \cl R$ then (\ref{eq_conve}) implies
\begin{align*}
\|u \oplus v\|_{\cl R}^{(n + m)} 
& \geq 
\sup_{V,W\in \cl R} \left\|\phi_{V\oplus W}^{(n+m)}(u \oplus v)\right\|
= 
\sup_{V,W\in \cl R} \left\|\phi_{V\oplus W}^{(n)}(u) \oplus \phi_{V\oplus W}^{(m)}(v)\right\|\\
& \geq 
\sup_{V,W\in \cl R} \hspace{-0.1cm}\max\left\{\|\phi_{V}^{(n)}(u)\|, \|\phi_{W}^{(m)}(v)\|\right\}
\hspace{-0.1cm} = \hspace{-0.1cm}
\max\left\{\|u\|_{\cl R}^{(n)},\|v\|_{\cl R}^{(m)}\right\}.
\end{align*}
The proof is complete. 
\end{proof}


The operator space, based on the vector space $\cl O_{X,A}$, resulting from Proposition \ref{p_opspst},
will be denoted by $\cl O^{\cl R}_{X,A}$. 
We note that, since the maps $\phi_U$ are completely contractive, 
the identity map $\cl O_{X,A}\to \cl O^{\cl R}_{X,A}$ is completely contractive, when $\cl O_{X,A}$ is equipped with the operator space structure described in Subsection \ref{s_S1}.

Let $\cl R$ be a resource, and $H_R$ be a Hilbert space.  For a unit vector $\xi\in \bb{C}^{X}\ten H_R$ and a positive contraction $P\in\cl B(\bb{C}^{A}\ten H_R)$, we let 
\begin{equation}\label{eq_omegaRP}
\omega_{\cl R}(\xi,P) = \sup_{\Gamma\in \mathsf{QC}(\cl R)} {\rm Tr}((\Gamma\otimes \id_R)(\xi\xi^*)P)
\end{equation}
be the \emph{$\cl R$-value} of the pair $(\xi,P)$. When $P=\rho$ is a density operator, the $\cl R$-value quantifies the ability to convert $\xi\xi^*$ to $\rho$ using local quantum channels from the resource $\cl R$. Note that when $\cl R$ is the family of all block operator isometries (so that $\mathsf{QC}(\cl R)=\mathsf{QC}(M_X,M_A)$), $\omega_{\cl R}(\xi,\rho)=1$ corresponds to perfect local state conversion: $\rho=(\Gamma\ten\id)\xi\xi^*$ for some $\Gamma\in\mathsf{QC}(M_X,M_A)$. Such local state conversion has recently been studied through the lens of quantum majorisation
\cite{ggpp,gouretal}. When $X$ and $A$ are bipartite, then for general $\cl R$ and $P$, the $\cl R$-value will quantify the maximal probability of success when playing a 2-player quantum game specified by $(\xi,P)$ using strategies from the resource $\cl R$ (see Section \ref{s_qv} and, in particular, Theorem \ref{th_qcval}). To that end, the following theorem is one of our main tools. In the statement and its proof, we suppress the notation for the isomorphism $\alpha$ from Proposition \ref{p_OXAS1} and,
for Hilbert spaces $H, K$ and $H_R$, 
use the notation $\Tr_R : \cl S_1(H\otimes H_R,K\otimes H_R)\to \cl S_1(H,K)$
for the partial trace along 
the trace class on the Hilbert space $H_R$. Complex conjugation is with respect to the standard basis. 

\begin{theorem}\label{th_R-val}
Let $\cl R$ be a resource over $(X,A)$, 
$\xi\in \bb{C}^{X}\ten H_R$ be a unit vector, and $P\in\cl B(\bb{C}^{A}\ten H_R)$ be a positive contraction with countable discrete spectrum and spectral decomposition
$P = \sum_{n=1}^\infty\lm_n\gamma_n\gamma_n^*$. 
Let $\rho_n = \sqrt{\lm_n}\,\overline{\Tr_R(\xi\gamma_n^*)}$, viewed as an element of $\cl O_{X,A}$, $n\in \bb{N}$.
Then
$$\omega_{\cl R}(\xi,P) = \norm{[\rho_n]}^2_{M_{\infty,1}(\cl O^{\cl R}_{X,A})}.$$
\end{theorem}

\begin{proof}
  Given $U\in \mathcal B(H^X,K^A)\cap \mathcal{R}$ and a normal
   state $\sigma\in \cl S_1(H)$, let
  $\Gamma=\Gamma_{U,\sigma}\in \mathsf{QC}(\mathcal R)$.
  By~\eqref{eq:Gamma-U-sigma},
  \[ \tr(\Gamma_{U,\sigma}(\epsilon_{x,x'}) \eps_{a',a})= \sigma(U_{a,x}^*U_{a',x'}),\quad x,x'\in X,\;a,a'\in A.\]
Decompose 
$\xi = \sum_{x\in X} e_x\ten \xi_x$ and $\gamma_n = \sum_{a\in A} e_a\ten \gamma_{n,a}$ for some $\xi_x,\gamma_{n,a}\in H_R$. Then
\begin{equation}\label{eq_Mn=}
\rho_n=\textstyle{\sqrt{\lm_n}}\,\overline{\Tr_R({\xi}{\gamma}_n^*)}=\sqrt{\lm_n}\sum_{x\in X}\sum_{a\in A}\la\gamma_{n,a},\xi_x\ra \eps_{x,a}.
\end{equation}
Since the sum in 
the spectral decomposition of $P$ is weak* convergent, taking into account (\ref{eq_Mn=}) we have
\begin{align*}
  &\Tr\big((\Gamma_{U,\sigma}\otimes \id_R)(\xi\xi^*)P\big)=
  \sum_{n=1}^\infty\lm_n \Tr\big((\Gamma_{U,\sigma}\otimes \id_R)(\xi\xi^*)(\gamma_n\gamma_n^*)\big)\\
  &=\sum_{n=1}^\infty
  \sum_{x,x' \in X}
\sum_{a,a'\in A}
  \lm_n \Tr\big((\Gamma_{U,\sigma}\otimes \id_R)(\eps_{x,x'}\ten \xi_x\xi_{x'}^*)(\eps_{a',a}\ten \gamma_{n,a'}\gamma_{n,a}^*)\big)\\
  &=\sum_{n=1}^\infty
  \sum_{x,x' \in X}
\sum_{a,a'\in A}\lm_n \Tr(\Gamma_{U,\sigma}(\eps_{x,x'})\eps_{a',a})\Tr\left((\xi_x\xi_{x'}^*)(\gamma_{n,a'}\gamma_{n,a}^*)\right)\\
    &=\sum_{n=1}^\infty
    \sum_{x,x' \in X}
\sum_{a,a'\in A}
    \lm_n \sigma(U_{a,x}^*U_{a',x'})\la\gamma_{n,a'},\xi_{x'}\ra\la \xi_x,\gamma_{n,a}\ra\\
    &=\sum_{n=1}^\infty
    \sigma\left(\phi_U(\rho_n)^*\phi_U(\rho_n)\right).
\end{align*}

It follows that the series $\sum_{n=1}^\infty \phi_U(\rho_n)^*\phi_U(\rho_n)$ converges weak* in $\cl B(H)$, and that
\begin{align}
\sup_{\sigma}\Tr\big((\Gamma_{U,\sigma}\otimes \id_R)(\xi\xi^*)P\big)&=\bignorm{\sum_{n=1}^\infty \phi_U(\rho_n)^*\phi_U(\rho_n)} \nonumber\\
&=\norm{[\phi_U(\rho_n)]}^2_{M_{\infty,1}(\cl B(H,K))}.\label{eq:U-value}
\end{align}
Taking the supremum over all $U\in \mathcal{R}$, we obtain the desired formula for $\omega_{\mathcal{R}}(\xi,P)$.
\end{proof}

\begin{remark}\label{r:ortho} The orthonormality of the eigenvectors $\gamma_n$ of $P$ was 
not used in the proof above. A similar result therefore applies to any decomposition of $P$ into a 
countable sum of positive rank one operators.
\end{remark}

A resource $\cl R$ over the pair $(X,A)$ will be called \emph{closed} if the set $\mathsf{QC}(\cl R)$ is closed.

\begin{corollary}\label{c_convert}
Let $\cl R$ be a closed resource over $(X,A)$, and 
$\xi\in \bb{C}^X$ and $\gamma\in \bb{C}^A$ be unit vectors. The following are equivalent:
\begin{itemize}
\item[(i)] there exists $\Gamma\in \mathsf{QC}(\cl R)$ such that $\Gamma(\xi\xi^*) = \gamma\gamma^*$;

\item[(ii)] $\left\|\overline{\xi\gamma^*}\right\|_{\cl R} = 1$.
\end{itemize}
\end{corollary}

\begin{proof}
(i)$\Rightarrow$(ii) Letting $P = \gamma\gamma^*$ in the formula (\ref{eq_omegaRP}), we see that 
$\omega_{\cl R}(\xi,P) = 1$. 
Theorem \ref{th_R-val} now implies that 
$\left\|\overline{\xi\gamma^*}\right\|_{\cl R} = 1$.

(ii)$\Rightarrow$(i) By Theorem \ref{th_R-val}, 
$\omega_{\cl R}(\xi,P) = 1$. Since the set $\mathsf{QC}(\cl R)$ is compact, there exists $\Gamma\in \mathsf{QC}(\cl R)$ such that 
${\rm Tr}(\Gamma(\xi\xi^*)P) = 1$; the positivity of $\Gamma$ now implies that $\Gamma(\xi\xi^*) = \gamma\gamma^*$.
\end{proof}

See Corollary \ref{LOSR_convert} for applications to the resources of local operations and shared randomness (LOSR) and local operations and classical communication (LOCC).

\subsection{Values of quantum hypergraphs}



We view $\bb{P}_X$ and $\cl P_A$ as metric spaces under the operator norm.
A \emph{probabilistic quantum hypergraph}
over $(X,A)$ is a triple 
$\bb{H} = (\bb{P}_X,\nph,\mu)$, where $\mu$ is a regular Borel probability measure on $\bb{P}_X$ and 
$\nph : \bb{P}_X \to \cl P_A$ is a (Borel) measurable function.
We will often abbreviate this as $\bb H=(\nph,\mu)$, where context allows.
Recall that a (classical) hypergraph $E$ with vertex set $X$ is a set of subsets of $X$, referred to as \emph{hyperedges}. A hypergraph $E$ can be identified with a subset (denoted in the same way) $E\subseteq X\times A$, where each element $a\in A$ gives rise to the hyperedge $E(a) = \{x\in X : (x,a)\in E\}$. The hypergraph $E$ is called \emph{probabilistic} if its vertex set $X$ is equipped with a probability distribution $\mu$. 
A probabilistic hypergraph $(E,\mu)$ gives rise to the probabilistic quantum hypergraph 
$\bb{H} = (\bb{P}_X, \nph, \mu)$, where 
$\mathrm{supp}(\mu) \subseteq \{\epsilon_{x,x}\}_{x\in X}$ and
$$\nph(\epsilon_{x,x}) = \sum_{a\in E^{-1}(x)} \epsilon_{a,a}, \ \ \ x\in X.$$

Given a probabilistic quantum hypergraph $\bb{H} = (\bb{P}_X,\nph,\mu)$ over $(X,A)$, and a resource $\cl R$ over $(X,A)$, we let

$$\omega_{\cl R}(\bb{H}) = \sup_{\Gamma\in \mathsf{QC}(\cl R)} 
\int_{\bb{P}_X} \Tr\left(\Gamma(p)\nph(p)\right)d\mu(p)$$
be the \emph{$\cl R$-value} of $\bb{H}$. This quantity has a similar operational interpretation as described above Theorem \ref{th_R-val}. We now establish a concrete expression for the $\cl R$-value in terms of the operator space structure associated with $\cl R$, which may be viewed as a ``semi-classical'' generalisation of a special case of Theorem \ref{th_R-val}, 
where $P$ has continuous spectrum. We require the following instance of measurable selection.

\begin{theorem}\label{t:azoff}\cite[Theorem 1]{azoff} Let $M$ and $N$ be complete separable metric spaces, $E$ be a closed $\sigma$-compact subset of $M\times N$ and $\pi_1 : M\times N \to M$ be the projection onto the first coordinate.
Then $\pi_1(E)$ is a Borel set in $M$ and there exists a Borel function $f:\pi_1(E)\to N$ whose graph is contained in $E$. 
\end{theorem}

Similarly to \cite[Corollary 2]{azoff}, Theorem \ref{t:azoff} allows us to measurably select a spectral decomposition.
We let $\bb{C}^A_{1}$ be the unit sphere of $\bb{C}^A$.

\begin{lemma}\label{l:selection} 
Let $A$ be a finite set. There exist Borel functions $\lambda_a:\cl P_A\to\{0,1\}$ and $u_a:\cl P_A\to \bb{C}^A_{1}$, $a\in A$, such that $p = \sum_{a\in A}\lambda_a(p)u_a(p)u_a(p)^*$ is a spectral decomposition, $p\in\cl P_A$.
\end{lemma}

\begin{proof} 
Let $\cl U(M_A)$ denote the unitary group of $M_A$, and
$$E = \{(p,u,d)\in \cl P_A\times \cl U(M_A)\times \cl D_A :  p=udu^*\}.$$
Since $\cl P_A$, $\cl U(M_A)$, and $\cl D_A$ are closed (hence complete) subsets of $M_A$, and the operations of multiplication and adjoint are continuous, $E$ is a closed ($\sigma$-compact) subset of the complete (separable) space $\cl P_A\times (\cl U(M_A)\times \cl D_A)$. By the spectral theorem, $\pi_1(E) = \cl P_A$. Hence, Theorem \ref{t:azoff} yields a Borel function $f:\cl P_A\to \cl U(M_A)\times \cl D_A$ whose graph is contained in $E$, that is, $f(p)=(u(p),d(p))$,
with $p=u(p)d(p)u(p)^*$. Since coordinate projections are Borel, the maps $u:=\pi_1\circ f:\cl P_A\rightarrow\cl U(M_A)$ and $d:=\pi_2\circ f:\cl P_A\to\cl D_A$ are Borel and satisfy $p=u(p)d(p)u(p)^*$, $p\in\cl P_A$. Let $u_1(p),\dots,u_{|A|}(p)$ denote the column vectors of $u(p)$, and let $\lm_1(p),\dots,\lm_{|A|}(p)$ denote the diagonal entries of $d(p)$. Then $u_a(p) = u(p)e_a$ 
and $\lm_a(p)=\la d(p)e_a,e_a\ra$. 
By the continuity of evaluation at vectors and of the map 
$\la (\cdot)e_a,e_a\ra$, we have that the functions $u_a:\cl P_A\to\bb{C}^A$ and $\lm_a:\cl P_A\to\bb{C}$ are Borel with ranges in 
$\bb{C}^A_{1}$ and $\{0,1\}$, respectively. Finally,
\[p=u(p)d(p)u(p)^* = \sum_{a\in A}\lm_a(p)u(p)\epsilon_{a,a}u(p)^* = \sum_{a\in A}\lambda_a(p)u_a(p)u_a(p)^*.\qedhere\]
\end{proof}

Let $\bb{H} = (\bb{P}_X,\nph,\mu)$ be a probabilistic quantum hypergraph over $(X,A)$. 
Apply Lemma \ref{l:selection} to obtain Borel functions $\lambda_x:\cl P_X\to\{0,1\}$, 
$u_x:\cl P_X\to\bb{C}^X_{1}$ and 
$\eta_a:\cl P_A\to\bb{C}^A$ such that
$$p = \sum_{x\in X}\lm_x(p)u_x(p)u_x(p)^*
\ \mbox{ and } \ q = \sum_{a\in A}\eta_a(q)\eta_a(q)^*, \quad p\in\cl P_X, \ q\in\cl P_A,$$
where we have incorporated the eigenvalue terms corresponding to $q$ from the statement of Lemma 
\ref{l:selection} into $\eta_a(q)$ 
(resulting, by the continuity of scalar multiplication, in the Borel measurability of 
$\eta_a$, $a\in A$).
Let $\xi : \bb{P}_X\rightarrow \bb{C}^X_{1}$ be the Borel function, given by
$$\xi(p) = \sum_{x\in X}\lambda_x(p)u_x(p),\quad p\in\bb{P}_X;$$
we consider $\xi$ as an element of the $\bb{C}^X$-valued space $L^2(\bb{P}_X,\mu; \bb{C}^X)$. 

We equip $A$ with the discrete topology and counting measure $|\hspace{-0.05cm}\cdot\hspace{-0.05cm}|$, and write 
$\tilde{\mu} = \mu\times |\hspace{-0.05cm}\cdot\hspace{-0.05cm}|$; thus, $\tilde{\mu}$ 
is a Borel measure on $\bb{P}_X\times A$.
Let $\bar{\xi} (\bar{\eta}\circ\vphi)^* 
: \bb{P}_X\times A \to \cl{S}_1^{A,X}$ be the (Borel) function, given by

\begin{equation}\label{ximueta}
\left(\bar{\xi}(\bar{\eta}\circ\vphi)^*\right) (p,a) = \overline{\xi(p)} \ \overline{\eta_a(\vphi(p))}^*, \ \ \ 
p\in \bb{P}_X, a\in A.
\end{equation}
Identifying $\cl S_1^{A,X}$ with $\cl O_{X,A}$ via Proposition \ref{p_OXAS1}, and using its 
finite dimensionality, in the next theorem 
we view $\bar{\xi} (\bar{\eta}\circ\vphi)^*$ as an element of 
the Haagerup tensor product 
$L^2(\bb{P}_X\hspace{-0.08cm}\times \hspace{-0.08cm}A,\tilde{\mu})\otimes_{\rm h}\cl O^{\cl R}_{X,A}$, 
where 
$L^2(\bb{P}_X\hspace{-0.08cm}\times \hspace{-0.08cm}A,\tilde{\mu})$ is equipped with its 
column operator space structure.

\begin{theorem}\label{th_hypvge}
Let $\bb{H} = (\bb{P}_X,\nph,\mu)$ be a probabilistic quantum hypergraph, and $\cl R$ be a resource, over $(X,A)$. Then 
$$\omega_{\cl R}(\bb{H}) = \norm{\bar{\xi}(\bar{\eta}\circ\vphi)^*}_{L^2(\bb{P}_X\times A,\tilde{\mu})\otimes_{\rm h}\cl O^{\cl R}_{X,A}}^2.$$
\end{theorem}

\begin{proof}
For every $\zeta\in\bb{C}^X_{1}$ we have that 
$$\xi(\zeta\zeta^*)=\sum_{x\in X} \lambda_x(\zeta\zeta^*)u_x(\zeta\zeta^*)\in\mathbb{T}\zeta,$$
so that 
$$\norm{\xi}^2_{L^2(\bb{P}_X,\mu; \bb{C}^X)} = \int_{\bb{P}_X}\norm{\xi(\zeta\zeta^*)}^2 d\mu(\zeta\zeta^*) = 1.$$
As $L^2(\bb{P}_X,\mu; \bb{C}^X)\cong \bb{C}^X\ten L^2(\bb{P}_X,\mu)$, we may write 
$\xi = \sum_{x\in X} e_x\ten\xi_{x}$, where $\xi_{x}(p) = \la\xi(p),e_x\ra$, $p\in\bb{P}_X$; we note that the function
$\xi_{x} : \bb{P}_X\to \bb{C}$ is Borel, $x\in X$. 

Since $\vphi : \bb{P}_X\to\cl P_A$ is a bounded Borel function, 
it canonically defines a projection 
$P_\vphi \in M_A\ten L^\infty(\bb{P}_X,\mu) \equiv L^\infty(\bb{P}_X,\mu; M_A)$ via 
$$(P_\vphi\eta)(p) = \vphi(p)\eta(p), \ \ \ \eta\in L^2(\bb{P}_X,\mu; \bb{C}^A).$$
Writing $\vphi(p)=\sum_{a\in A}\eta_a(\vphi(p))\eta_a(\vphi(p))^*$, 
we have 
$$P_\vphi=\sum_{a,a'\in A} \epsilon_{a,a'}\ten 
\sum_{a''\in A}\la\eta_{a''}(\vphi(\cdot)),e_a\ra\la e_{a'},\eta_{a''}(\vphi(\cdot))\ra.$$

Calculations, similar to the ones in the proof of Theorem \ref{th_R-val}, show that
\begin{eqnarray*}
\omega_{\cl R}(\bb{H})
& = & 
\sup_{U,\sigma}\sum_{x,x' \in X}
\sum_{a,a'\in A}
\sigma(U_{a,x}^*U_{a',x'})\int_{\bb{P}_X}
\langle\vphi(p)e_a,e_{a'}\rangle\xi_{x}(p)\overline{\xi_{x'}(p)} d\mu(p)\\
& = & 
\sup_{U,\sigma}\int_{\bb{P}_X}\sum_{a\in A}\sigma\big(\phi_U(\overline{\xi(p)} \ \overline{\eta_a(\vphi(p))}^*)^*\phi_U(\overline{\xi(p)} \ \overline{\eta_a(\vphi(p))}^*)\big)d\mu(p),
\end{eqnarray*}
where the suprema are taken over all $U\in \cl R$ and all corresponding states $\sigma$. 
Letting $H$ and $K$ be the Hilbert spaces corresponding to the isometry $U$, 
we have that the function 
$F_U : \bb{P}_X\times A\to\cl B(H)$, given by 
$$F_U(p,a) = \phi_U(\overline{\xi(p)} \ \overline{\eta_a(\vphi(p))}^*)^*\phi_U(\overline{\xi(p)} \ \overline{\eta_a(\vphi(p))}^*),
$$
is a bounded $\tilde{\mu}$-measurable function taking values in a 
finite-dimensional subspace of $\cl B(H)$. It follows that


\begin{align*}
\omega_{\cl R}(\bb{H})
&
= 
\sup_{U,\sigma}\int_{\bb{P}_X}\sum_{a\in A}\sigma(F_U(p,a))d\mu(p)
=  
\sup_{U,\sigma}\left\langle \int_{\bb{P}_X}\sum_{a\in A} F_U(p,a) d\mu(p), \sigma\right\rangle\\
&=
\sup_{U\in \cl R}\left\|\int_{\mathbb P_X}\sum_{a\in A}F_U(p,a)d\mu(p)\right\|^2 \\
&=
\sup_{U\in \cl R}
\norm{(\id\ten \phi_U)(\bar{\xi}(\bar{\eta}\circ\vphi)^*)}_{L^2(\bb{P}_X\times A,\tilde{\mu})\otimes_{\rm h}\cl B(H,K)}^2\\
&
=\norm{\bar{\xi}(\bar{\eta}\circ\vphi)^*}_{L^2(\bb{P}_X\times A,\tilde{\mu})\otimes_{\rm h}\cl O^{\cl R}_{X,A}}^2
\end{align*}
(we refer to the proof of \cite[Lemma 4.4]{bc} for justification of the third equality).
The proof is complete. 
\end{proof}

\begin{remark}\label{r_link}
\rm
We comment on a special case, which highlights the link between the values of a pair $(\xi,P)$ as in Theorem \ref{th_R-val} and the value of a probabilistic quantum hypergraph. 
Suppose that $\bb{H} = (\bb{P}_X,\nph,\mu)$ is a probabilistic quantum hypergraph, for which the measure $\mu$ has finite support 
$\cl P = \{p_i\}_{i=1}^k$. Let
$\xi_i \in \bb{C}^X$ be a unit vector such that $p_i = \xi_i\xi_i^*$, $i = 1,\dots,k$, and set 
$\xi_{\mu} = \sum_{i=1}^k \sqrt{\mu(p_i)} \xi_i  \otimes e_i \in \bb{C}^{X} \otimes \bb{C}^k$ and 
$P_{\nph} = \sum_{i=1}^k \nph(p_i) \otimes \epsilon_{i,i}$. 
Then
\begin{align*}
\omega_{\cl R}\left(\xi_{\mu},P_{\nph}\right)& = 
\sup_{\Gamma\in \mathsf{QC}(\cl R)} 
{\rm Tr}\sum_{i,j=1}^k \sqrt{\mu(p_i)}\sqrt{\mu(\smash{p_j})}
\left(\left(\Gamma(\xi_i\xi_j^*)\otimes \epsilon_{i,j}\right) 
\left(\nph(p_i)\otimes \epsilon_{i,i}\right)\right)\\
& = 
\sup_{\Gamma\in \mathsf{QC}(\cl R)} \sum_{i = 1}^k \mu(p_i) \Tr\left(\Gamma(p_i)\nph(p_i)\right) 
= 
\omega_{\cl R}(\bb{H}).
\end{align*}
Writing $\nph(p_i) = \sum_{j=1}^{k_i} \eta_{i,j}\eta_{i,j}^*$, where $\{\eta_{i,j}\}_{j=1}^{k_i}$ is an orthonormal set of vectors in $\bb{C}^A$, 
by Theorem \ref{th_R-val} (or Theorem \ref{th_hypvge}) we thus have that 
$$\omega_{\cl R}(\bb{H}) = 
\left\|\left[\textstyle{\sqrt{\mu(p_i)}}\, \overline{\xi_i\eta_{i,j}^*} \right]_{i,j}\right\|_{\cl R}^2,
$$
where $\left[\textstyle{\sqrt{\mu(p_i)}}\,\overline{\xi_i \eta_{i,j}^*} \right]_{i,j}$ is considered as a column operator. \end{remark}


\section{Resource types}\label{s_qv}

In this section, we apply the general setup of Section \ref{s_hyper} to
five types of no-signalling resources, obtaining as a consequence expressions 
for the different types of values of non-local games.


\subsection{Quantum no-signalling correlations}\label{ss_qnlgd}

Let $X$, $Y$, $A$ and $B$ be non-empty finite sets.
A \emph{quantum correlation over $(X,Y,A,B)$} 
(or simply a \emph{quantum correlation} if the sets are understood from the context) 
is a quantum channel $\Gamma : M_{XY}\to M_{AB}$.
Such a $\Gamma$ is called a
\emph{quantum no-signalling (QNS) correlation} \cite{dw} if
\begin{equation}\label{eq_qns1}
\Tr_A\Gamma(\rho_X\otimes \rho_Y) = 0 \ \mbox{ whenever } 
\rho_X\in M_X, \rho_Y\in M_Y \mbox{ and } 
\Tr(\rho_X) = 0,
\end{equation}
and
\begin{equation}\label{eq_qns2}
\Tr_B\Gamma(\rho_X\otimes \rho_Y) = 0 \ \mbox{ whenever } 
\rho_X\in M_X, \rho_Y\in M_Y \mbox{ and } \Tr(\rho_Y) = 0.
\end{equation}
We let $\cl Q_{\rm ns}$ be the set of all QNS correlations.

A \emph{stochastic operator matrix} over $(X,A)$,
acting on a Hilbert space $H$, is a positive block operator matrix $\tilde{E} = (E_{x,x',a,a'})_{x,x',a,a'}\in M_{XA}(\cl B(H))$ such that $\Tr_A \tilde{E} = I_X\otimes I_H$. 
A QNS correlation $\Gamma : M_{XY}\to M_{AB}$ is \emph{quantum commuting} \cite{tt-QNS} if 
there exist a Hilbert space $H$, a unit vector $\xi\in H$ and
stochastic operator matrices $\tilde{E} = (E_{x,x',a,a'})_{x,x',a,a'}$ over $(X,A)$ and $\tilde{F} = (F_{y,y',b,b'})_{y,y',b,b'}$ over $(Y,B)$, acting 
on $H$, such that 
$$E_{x,x',a,a'} F_{y,y',b,b'} = F_{y,y',b,b'}E_{x,x',a,a'}$$
for all $x,x'\in X$, $y,y'\in Y$, $a,a'\in A$ and $b,b'\in B$,
and 
\begin{equation}\label{eq_EFp}
\Gamma(\epsilon_{x,x'} \otimes \epsilon_{y,y'}) = \sum_{a,a'\in A} \sum_{b,b'\in B}
\left\langle E_{x,x',a,a'}F_{y,y',b,b'}\xi,\xi \right\rangle \epsilon_{a,a'} \otimes \epsilon_{b,b'}, 
\end{equation}
for all $x,x' \in X$ and all $y,y' \in Y$.
\emph{Quantum} QNS correlations are defined as in (\ref{eq_EFp}), but 
requiring that $H$ has the form $H_A\otimes H_B$, for some finite dimensional Hilbert spaces $H_A$ and $H_B$, and 
$E_{x,x',a,a'} = \tilde{E}_{x,x',a,a'} \otimes I_B$ and 
$F_{y,y',b,b'} = I_A \otimes \tilde{F}_{y,y',b,b'}$, for some stochastic operator matrices 
$(\tilde{E}_{x,x',a,a'})$ and 
$(\tilde{F}_{y,y',b,b'})$, acting on $H_A$ and $H_B$, respectively. 
Finally, the \emph{local} QNS correlations are 
the convex combinations of the form 
$\Gamma = \sum_{i=1}^k \lambda_i \Phi_i \otimes \Psi_i$, 
where $\Phi_i : M_X\to M_A$ and $\Psi_i : M_Y\to M_B$ are quantum channels, $i = 1,\dots,k$. Note that the local QNS correlations coincide
with the class of local operations and shared randomness (LOSR) \cite{watrous}.

We write $\cl Q_{\rm qc}$ (resp. $\cl Q_{\rm q}$, $\cl Q_{\rm loc}$) for the (convex) set of all quantum commuting (resp. quantum, local) QNS correlations, and note the
strict \cite{tt-QNS}
inclusions 
\begin{equation}\label{eq_Qchain}
\cl Q_{\rm loc}\subseteq \cl Q_{\rm q} \subseteq \cl Q_{\rm qc}\subseteq \cl Q_{\rm ns}.
\end{equation}


\subsection{The local resource}\label{ss_loc}

Let 
$$\cl R_{\rm loc} \hspace{-0.1cm}=\hspace{-0.1cm} \left\langle\hspace{-0.05cm}\{U\otimes V : U \in \cl B(\bb{C}^X\hspace{-0.1cm},\hspace{-0.05cm}\bb{C}^{AS}), 
V \in \cl B(\bb{C}^Y\hspace{-0.1cm},\hspace{-0.05cm}\bb{C}^{BT})
\mbox{ isom.}, \ S,T\mbox{ finite sets}\}\hspace{-0.05cm}\right\rangle\hspace{-0.1cm}.$$

\begin{proposition}\label{p-Rvalq}
We have that $\mathsf{QC}(\cl R_{\rm loc}) = \cl Q_{\rm loc}$.
\end{proposition}

\begin{proof}
Let $\Phi : M_X\to M_A$ be a quantum channel, and write
$$\Phi(\rho)=\sum_{s\in S}V_s\rho V_s^*, \ \ \ \rho\in M_X,$$
in its Kraus decomposition; thus,
$S$ is a finite set and 
$V_s : \bb{C}^X\to \bb{C}^A$ are operators such that $\sum_{s\in S} V_s^*V_s = I_X$. 
Let $V : \bb{C}^X\to \bb{C}^A\otimes \bb{C}^S$ be the column operator, given by 
$V\xi = (V_s\xi)_{s\in S}$; we have that $V$ is an isometry.
Writing $V = (V_{a,x})_{a,x}$, where $V_{a,x} : \bb{C}\to \bb{C}^S$, 
we have 
$V_{a,x} = (\langle V_s e_x,e_a\rangle)_{s\in S}$ as a vector in $\bb{C}^S$.
Thus, viewing $V$ as a block operator isometry in $M_{A,X}\otimes \cl B(\bb C^X,(\bb C^S)^A)$ and writing $1$ for the trivial state on $\bb C$, by (\ref{eq:Gamma-U-sigma}) we have
\begin{align*}
\langle\Gamma_{V,1}(\epsilon_{x,x'})e_{a'},e_a\rangle
& = 
V_{a,x}^*V_{a',x'}
= \sum_{s\in S} 
\overline{\langle V_s e_x,e_a\rangle} \langle V_s e_{x'},e_{a'}\rangle\\
& =  
\sum_{s\in S} 
\langle e_a,V_s e_x\rangle \langle V_s e_{x'},e_{a'}\rangle
= 
\sum_{s\in S} 
\langle (V_s e_{x'})(V_s e_x)^* e_a,e_{a'}\rangle\\
& =  
\sum_{s\in S} \langle V_s \epsilon_{x',x} V_s^* e_a,e_{a'}\rangle
= \langle \Phi(\epsilon_{x',x})e_a,e_{a'}\rangle;
\end{align*}
hence $\Gamma_{V,1} = \Phi^{\#}$, where $\Phi^{\#}(\rho) = \Phi(\rho^{\rm t})^{\rm t}$, as in equation (\ref{e:Gammap}). Let $\overline{V}:\bb{C}^X\to\bb{C}^A\ten\bb{C}^S$ denote the isometry, given by 
$$\overline{V}e_x = \sum_{a\in A} e_a\ten \overline{V_{a,x}}, \quad x\in X.$$
If $x,x'\in X$, then
\begin{eqnarray*}
\Gamma_{V,1}^{\#}(\epsilon_{x,x'})
= 
\Gamma_{V,1}(\epsilon_{x',x})^{\rm t}
& = & 
\sum_{a,a'\in A} V_{a',x'}^*V_{a,x} \epsilon_{a,a'}\\
& = & 
\sum_{a,a'\in A} \overline{V_{a,x}}^*\overline{V_{a',x'}} \epsilon_{a,a'}
=
\Gamma_{\overline{V},1}(\epsilon_{x,x'}).
\end{eqnarray*}
Thus, $\Phi=\Gamma_{V,1}^{\#}=\Gamma_{\overline{V},1}$. 

By Remark \ref{r_tenso}, 
if $\Psi : M_Y\to M_B$ is a(nother) quantum channel, then 
$\Phi\otimes\Psi\in \mathsf{QC}(\cl R_{\rm loc})$.
By Remark \ref{r_conve}, $\mathsf{QC}(\cl R_{\rm loc})$ contains the
convex hull of product channels of the form 
$\Phi\otimes\Psi$; hence $\cl Q_{\rm loc}\subseteq \mathsf{QC}(\cl R_{\rm loc})$.
The reverse inclusion follows by reversing the previous arguments.
\end{proof}

\begin{corollary}\label{c_locre}
The family $\cl R_{\rm loc}$ is a resource over $(X Y,A B)$. 
\end{corollary}

\begin{proof}
The set $\cl R_{\rm loc}$ is closed under direct sums by definition.
For separation, assume that 
$u=\sum_{x,y,a,b}\lm_{x,a,y,b} u_{a,x}\ten v_{b,y}\in \cl{O}_{XY,AB}= \cl{O}_{X,A}\ten \cl{O}_{Y,B}$ and that 
$s_{U\ten V,1}(u^*u)=0$ for all $U\ten V\in \cl R_{\rm loc}$. 
Define
$$U:\bb{C}^X\ni e_x\mapsto\sum_{a\in A} e_a\ten \underbrace{\frac{1}{|A|^{1/2}}e_a\ten e_x}_{=:U_{a,x}\in \bb C^A\otimes \bb C^X}\in\bb{C}^A\ten(\bb{C}^A\ten\bb{C}^X)$$
and, similarly,
$$V:\bb{C}^Y\ni e_y\mapsto\sum_{b\in B} e_b\ten \underbrace{\frac{1}{|B|^{1/2}}e_b\ten e_y}_{=:V_{b,y}\in \bb{C}^B\ten\bb{C}^Y}\in\bb{C}^B\ten(\bb{C}^B\ten\bb{C}^Y).$$
Then $U$ and $V$ are isometries and $\{U_{a,x}\ten V_{b,y} : a\in A, \ b\in B, \ x\in X, y\in Y\}$ is an orthogonal set. Hence,
$$0 = s_{U\ten V,1}(u^*u) = \bignorm{\sum_{x,y,a,b}
\lm_{x,a,y,b} U_{a,x}\ten V_{b,y}}^2=|AB|^{-1}\sum_{x,y,a,b}|\lambda_{x,a,y,b}|^2,$$
implying $u = 0$. 
\end{proof}

\begin{remark}\label{r_infi}
\rm 
An examination of the proof of Proposition \ref{p-Rvalq} shows that
its statement remains true if the resource $\cl R_{\rm loc}$ is replaced by the 
resource $\cl R_{\rm loc}^{\infty}$, defined similarly to $\cl R_{\rm loc}$ but allowing 
infinite sets $S$ and $T$. 
\end{remark}

Let $\cl X$ and $\cl Y$ be operator spaces. 
In the sequel, it will be convenient to write ${\rm CC}(\cl X,\cl Y)$ for the collection of 
all completely contractive maps in ${\rm CB}(\cl X,\cl Y)$. 
Let $R = (\ell_2)_r$ and $C = (\ell_2)_c$
denote the row and column operator space structures over $\ell_2=\ell_2(\mathbb{N})$, respectively, and write $\cl{S}_2(\ell_2)$ for the
Hilbert-Schmidt ideal in $\cl B(\ell_2)$. 
For a linear map
$\vphi:\cl X\to \cl Y$, 
its \textit{${\rm cb}$-weak Hilbert-Schmidt norm} \cite{jkppg} is
defined by letting
\begin{equation}\label{eq_wcb}
\norm{\vphi}_{w,{\rm cb}} :=
\sup\{\norm{ \beta\circ\vphi\circ \alpha}_{\cl{S}_2(\ell_2)} : 
\alpha\in {\rm CC}(R,\cl X), \beta\in {\rm CC}(\cl Y,C)\};
\end{equation}
thus, $\norm{\vphi}_{w,{\rm cb}}$ 
is the supremum of the Hilbert-Schmidt norms of the compositions
\begin{equation*}
\begin{tikzcd}
\ell_2\arrow[r]\arrow[d, " \alpha"] & \ell_2\\
\cl X\arrow[r, "\vphi"] & \cl Y\arrow[u, " \beta"],
\end{tikzcd}
\end{equation*}
where $\|\alpha : R\to \cl X\|_{\rm cb}\leq 1$ and $\|\beta : \cl Y \to C\|_{\rm cb}\leq 1$. 
We note that the supremum in (\ref{eq_wcb}) can be taken over all finite-dimensional Hilbert 
operator spaces $R_k=(\bb C^k)_r$ and $C_k=(\bb C^k)_c$, $k\in\mathbb N$. 
We denote by 
$\cl{S}_2^{w,{\rm cb}}(\cl X,\cl Y)$ the space of 
\emph{weak-cb Hilbert-Schmidt operators}, that is, 
the linear maps 
$\vphi: \cl X\rightarrow \cl Y$ with $\norm{\vphi}_{w,{\rm cb}} < \infty$.
We note that ${\rm CB}(\cl X,\cl Y)\subseteq \cl{S}_2^{w,{\rm cb}}(\cl X,\cl Y)$; 
indeed, the isometric identification $\cl S_2(\ell_2) = {\rm CB}(R,C)$ \cite[Corollary 4.5]{er91} 
implies that
\begin{eqnarray*}
\norm{\vphi}_{w,{\rm cb}}
& = & 
\sup\left\{\norm{\beta\circ\vphi\circ \alpha : R\to C}_{\rm cb} : 
\alpha\in {\rm CC}(R,\cl X), \beta\in {\rm CC}(\cl Y,C)\right\}\\
& \leq & 
\norm{\vphi}_{\rm cb}.
\end{eqnarray*}

Given $[\vphi_{i,j}]_{i,j}\in M_n(\cl{S}_2^{w,{\rm cb}}(\cl X,\cl Y))$, letting 
$\vphi : \cl X\rightarrow M_n(\cl Y)$ be the associated map, we set
$$\norm{[\vphi_{i,j}]}_{w,{\rm cb}}^{(n)}
= \sup\{\norm{ \beta^{(n)}\circ \vphi\circ \alpha}_{\rm cb}
 : 
\alpha\in {\rm CC}(R,\cl X), \beta\in {\rm CC}(\cl Y,C)\}.$$
Appealing to Ruan's Theorem~\cite[Proposition 2.3.6]{er},
one easily sees that the sequence $(\|\cdot\|_{w,{\rm cb}}^{(n)})_{n\in \bb{N}}$ of matricial norms
defines an operator space structure on the space
$\cl{S}_2^{w,{\rm cb}}(\cl X,\cl Y)$, and has the property that
\begin{equation}\label{eq_compo}
\norm{G^{(n)}\circ \vphi\circ F}_{w,{\rm cb}}^{(n)} \leq\norm{G}_{\rm cb}\norm{\vphi}_{w,{\rm cb}}^{(n)}\norm{F}_{\rm cb}
\end{equation}
for all completely bounded maps $F : \cl X'\to \cl X$ and $G : \cl Y\to \cl Y'$ and all 
$\nph \in M_n(\cl{S}_2^{w,{\rm cb}}(\cl X,\cl Y))$.

For finite-dimensional spaces $\cl X$ and $\cl Y$, the space of linear maps 
$\varphi:\cl X\to \cl Y$ can be identified with the tensor product 
$\cl X^*\otimes \cl Y$, associating to any element 
$\sum_{i=1}^m v_i^*\otimes w_i\in \cl X^*\otimes \cl Y$ the linear map 
$\varphi: \cl X\to \cl Y$, $v\mapsto \sum v_i^*(v)w_i$, and, conversely, 
associating to  $\varphi: \cl X\to \cl Y$ the tensor $\sum_{i=1}^m v_i^*\otimes \varphi(v_i)$, 
where $(v_i)_{i=1}^m$ and $(v_i^*)_{i=1}^m$ are mutually 
dual bases of $\cl X$ and $\cl X^*$, respectively;
we thus identify linear maps $\varphi:\cl X\to\cl Y$ with the corresponding elements in $\cl X^*\otimes\cl Y$.
Making such an identification, we write 
$\cl X^*\ten_{w,{\rm cb}}\cl Y$ in the place of  
$\cl{S}_2^{w,{\rm cb}}(\cl X,\cl Y)$, 
and consider $\|\cdot\|_{w,{\rm cb}}^{(n)}$ as a norm on $M_n(\cl X^*\ten_{w,{\rm cb}}\cl Y)$.

\begin{theorem}\label{th_wcb}
We have that 
$\cl O_{XY,AB}^{\cl R_{\rm loc}} = \cl{S}_1^{A,X}\ten_{w,{\rm cb}}\cl{S}_1^{B,Y},$
up to a complete isometry. 
\end{theorem}

\begin{proof} 
Let $\omega\in M_n(\cl S_1^{A,X}\otimes \cl S_1^{B,Y})$; as needed, we will interpret 
$\omega$ as a linear map from $M_{A,X}$ into $M_n(\cl S_1^{B,Y})$ without 
specific clarification.
Fix finite sets $S$ and $T$, and isometries 
$U : \bb{C}^X\to \bb{C}^{AS}$ and $V : \bb{C}^Y\to \bb{C}^{BT}$.
Let $R_S=(\bb C^S)_r$, $R_T=(\bb C^T)_r$ and $C_T=(\bb C^T)_c$.
Write $U = (U_s)_{s\in S}$ and $V = (V_t)_{t\in T}$, where 
$U_s \in M_{A,X}$ and $V_t \in M_{B,Y}$, $s\in S$, $t\in T$. 
Let 
$\alpha \in {\rm CB}(R_S,M_{A,X})$ and 
$\beta_* \in {\rm CB}(R_T,M_{B,Y})$ be the maps, given by
$$\alpha(\xi) = \sum_{s\in S} \la\xi,e_s\ra U_s \ \mbox{ and } \  
\beta_*(\eta) = \sum_{t\in T} \la\eta,e_t\ra V_t, 
\ \ \xi\in\ell_2^S, \eta\in\ell_2^T;$$
note that $\alpha$ (resp. $\beta$) is completely contractive due to the 
completely isometric identification 
${\rm CB}(R_S,M_{A,X}) = C_S\otimes_{\min} M_{A,X}$ 
(resp. ${\rm CB}(R_T,M_{B,Y}) = C_T\otimes_{\min} M_{B,Y}$). 
Set $\beta = (\beta_*)^*$; thus, 
$\beta \in {\rm CC}(\cl{S}_1^{B,Y},C_T)$ is given by
$$\beta(\rho) = \sum_{t\in T} \la\rho,V_t\ra e_t, \ \ \ \rho\in\cl{S}_1^{B,Y}.$$
Assuming for a moment that $n = 1$ (so that $\omega\in \cl S_1^{A,X}\otimes\cl S_1^{B,Y}$),
for $\xi\in \ell_2^S$, we have 
\begin{eqnarray*}
(\beta\circ\omega\circ\alpha)(\xi) 
\hspace{-0.2cm} & = & 
\hspace{-0.3cm} \sum_{s\in S} \langle \xi,e_s\rangle (\beta\circ\omega)(U_s)
= \hspace{-0.1cm}
\sum_{s\in S} \sum_{t\in T} \langle \xi,e_s\rangle \langle \omega (U_s),V_t\rangle e_t\\
\hspace{-0.2cm} & = & 
\hspace{-0.3cm} \sum_{s\in S} \sum_{t\in T} \langle \omega (U_s),V_t\rangle e_te_s^*(\xi)
= \hspace{-0.1cm}
\left(\sum_{s\in S} \sum_{t\in T} \langle U_s\otimes V_t, \omega\rangle e_te_s^*\right)
\hspace{-0.1cm}(\xi).
\end{eqnarray*}
On the other hand, assuming that $\omega = (\xi_1\xi_2^*)\otimes (\eta_1\eta_2^*)$, where 
$\xi_1\in \bb{C}^X$, $\xi_2\in \bb{C}^A$, $\eta_1\in \bb{C}^Y$ and $\eta_2\in \bb{C}^B$, 
we have 
\begin{eqnarray*}
(\phi_U\otimes\phi_V)(\omega) 
\hspace{-0.2cm} & = & \hspace{-0.2cm}
(\langle U_s\xi_1,\xi_2\rangle)_{s\in S} \otimes \hspace{-0.05cm}(\langle V_t\eta_1,\eta_2\rangle)_{t\in T}\\
\hspace{-0.2cm} & = & \hspace{-0.2cm}
(\langle U_s\otimes V_t, \omega\rangle)_{(s,t)\in S\times T} 
= 
\sum_{s\in S} \sum_{t\in T} \langle U_s\otimes V_t, \omega\rangle e_s\otimes e_t.
\end{eqnarray*}
Relaxing the assumption that $n = 1$, write $\omega = (\omega_{i,j})_{i,j=1}^n$. 
Let $Z_{s,t} = (\langle U_s\otimes V_t, \omega_{i,j}\rangle)_{i,j}$; thus, $Z_{s,t}\in M_n$, 
$s\in S$, $t\in T$.
Hence
$$(\phi_U\otimes\phi_V)^{(n)}(\omega) = \sum_{s\in S} \sum_{t\in T} 
e_s\otimes e_t \otimes Z_{s,t},$$
while 
$$\beta^{(n)}\circ\omega\circ\alpha = 
\sum_{s\in S} \sum_{t\in T} e_te_s^* \otimes Z_{s,t}.$$
It follows that 
$$\|\beta^{(n)}\circ\omega\circ\alpha : R\to M_n(C)\|_{\rm cb} 
= \|(\phi_U\otimes\phi_V)^{(n)}(\omega)\|_{M_n(C_{ST})};$$
indeed, the completely contractively complemented inclusions $R_S\subseteq R$ and $C_T\subseteq C$ imply that
$$\|\beta^{(n)}\circ\omega\circ\alpha : R\to M_n(C)\|_{\rm cb}
= \|\beta^{(n)}\circ\omega\circ\alpha : R_S\to M_n(C_T)\|_{\rm cb},$$
while $C_{ST} = (\bb{C}^S\ten\bb{C}^T)_c = C_S\ten_{\min} C_T$ \cite[Proposition 9.3.5]{er}, and 
\begin{equation}\label{e:tensor}{\rm CB}(R_S,M_n(C_T)) = M_n({\rm CB}(R_S,C_T)) = M_n(C_S\ten_{\min}C_T),\end{equation}
the last equality using \cite[Proposition 8.1.2]{er}. Thus, by Remark \ref{r_conve2}, 
$\|\omega\|_{\cl R_{\rm loc}}^{(n)}$ $\leq$ $\|\omega\|_{w,{\rm cb}}^{(n)}$.

For the reverse inequality, 
let $k\in\mathbb N$ and $\alpha : R_k\to M_{A,X}$ and 
$\beta : \cl{S}_1^{B,Y} \to C_k$ be complete contractions. 
By \cite[Proposition 1.2.28]{blm}, \cite[Proposition 9.3.1]{er}
and the 
finite-dimensionality of $M_{A,X}$, we have that 
${\rm CB}(R_k,M_{A,X})=C_k\ten_{\rm h}M_{A,X}$,
up to a canonical complete isometry. 
Thus, $\alpha$ is given by a column operator $[U_i]_{i=0}^{k}$ with 
$Q := \sum_{i=0}^k U_i^*U_i\leq I_X$, where $U_i\in M_{A,X}$ for $i=0,\dots,k$. 
Let $I_X - Q = \sum_{j=1}^{|X|}\mu_j\eta_j\eta_j^*$ be the spectral decomposition of the positive operator $I_X - Q\in M_X$. 
Let $\zeta\in \bb{C}^A$ be a unit vector and set 
$U_{-j} = \sqrt{\mu_j}\zeta \eta_j^*$; thus, $U_{-j} \in M_{A,X}$. 
We have that $I_X - Q = \sum_{j=-|X|}^{-1} U_j^*U_j$; setting 
$S = \{-|X|,\dots,-1,0,\dots,k\}$, we have that 
$\tilde{U} = [U_i]_{i \in S}$ is a column isometry. 

Similarly, the adjoint $\beta^* : R_k\to M_{B,Y}$ of $\beta$
gives rise to a contraction $[V_j]_{j=0}^{k}\in C_k\otimes_{\rm h} M_{B,Y}$, 
which, letting $T = \{-|Y|,\dots,-1,\dots, k\}$,
can be completed to a column isometry $\tilde{V} = [V_j]_{j\in T}$.
Let $\tilde{\alpha} : R_S\to M_{A,X}$ (resp. 
$\tilde{\beta} : \cl{S}_1^{B,Y}\to C_T$) be the map, canonically associated with 
$\tilde{U}$ (resp. $\tilde{V}$). 
If $F : R_k\to R_S$ (resp. $G : S_T\to C_k$) is the inclusion (resp. projection) map, 
inequality (\ref{eq_compo}) together with
the first part of the proof imply that 
\begin{multline*}
\norm{\beta^{(n)}\circ\omega\circ\alpha}_{\rm cb}
= \norm{G^{(n)}\circ\tilde\beta^{(n)}\circ\omega\circ\tilde\alpha\circ F}_{\rm cb}\leq \norm{\tilde\beta^{(n)}\circ\omega\circ\tilde\alpha}_{\rm cb}\\=\norm{(\phi_{\tilde U}\otimes\phi_{\tilde V})^{(n)}(\omega)}_{M_n(C_S\otimes_{\rm min}C_T)}\leq\|\omega\|_{\cl R_{\rm loc}}.
\end{multline*}
Taking supremum over all such $k$, $\alpha$ and $\beta$ we obtain
$\|\omega\|_{w,{\rm cb}}\leq \|\omega\|_{\cl R_{\rm loc}}$, and
the proof is complete. 
\end{proof}


\subsection{The quantum resource}\label{ss_eprq}

Let 
$$\cl R_{\rm q} = \left\langle\{U\otimes V : U = (U_{a,x})_{a,x}, V = (V_{b,y})_{b,y}
\mbox{ finite rank isometries}\}\right\rangle$$
(we note that the domains of $U$ and $V$ are allowed to vary).
It is clear that $\cl R_{\rm q}$ consists of isometries and is closed under 
direct sums.
Since $\cl R_{\rm loc} \subseteq \cl R_{\rm q}$, 
Corollary \ref{c_locre} implies that $\cl R_{\rm q}$ is a resource
over $(X Y,A B)$.

\begin{proposition}\label{p-Rvalqq}
We have that $\mathsf{QC}(\cl R_{\rm q}) = \cl Q_{\rm q}$.
\end{proposition}

\begin{proof}
  The elements of $\mathsf{QC}(\cl R_{\rm q})$ are, by their definition and Remark~\ref{r_conve},
  precisely the convex hull of
  the quantum channels $\Gamma : M_{XY}\to M_{AB}$ with Choi matrices 
equal to
$$\left(\sigma\left(U_{a,x}^*U_{a',x'}\otimes V_{b,y}^*V_{b',y'}\right)\right)_{x,x',y,y'}^{a,a',b,b'}$$
for finite rank block operator isometries $U$ and $V$.
In view of the factorisation result \cite[Theorem 3.1]{tt-QNS} for 
stochastic operator matrices, the latter class consists precisely of 
the quantum channels with Choi matrices 
$$\left(\sigma\left(E_{x,x',a,a'}\otimes F_{y,y',b,b'}\right)\right)_{x,x',y,y'}^{a,a',b,b'},$$
where $E = (E_{x,x',a,a'})$ and $F = (F_{y,y',b,b'})$ are finite dimensionally acting 
stochastic operator matrices. 
Writing the (possibly mixed) state $\sigma$ as a convex combination of pure states, 
we see that the latter matrices coincide with the Choi matrices of the channels from $\cl Q_{\rm q}$.
\end{proof}

\begin{remark}\label{r_opspsqa}
\rm 
In view of Proposition \ref{p_OXAS1}, by the injectivity of the minimal tensor product
we have that 
$$\cl S_1^{A,X}\otimes_{\min}\cl S_1^{B,Y}\subseteq \cl V_{X,A}\otimes_{\min}\cl V_{Y,B},$$
and that 
$$\|u\|_{\min} = \sup\{\|(\theta_V\otimes\theta_W)(u)\| : V, W \mbox{ finite rank isometries}\}.$$
It follows from \cite[Section 5]{tt-QNS} that 
$\cl S_1^{A,X}\otimes_{\min}\cl S_1^{B,Y} = \cl O_{XY,AB}^{\cl R_{\rm q}}.$
\end{remark}


\subsection{The quantum commuting resource}\label{ss_qc}

If $H$, $K$ and $L$ are Hilbert spaces, two families 
$\cl E\subseteq \cl B(K,L)$ and $\cl F\subseteq \cl B(H,K)$ will be called 
\textit{semi-commuting} if the families $\cl E^*\cl E$ and $\cl F\cl F^*$ of operators on $K$
commute.
We will say that the block operator matrices $U = (U_{a,x})_{a,x}$ and $V = (V_{b,y})_{b,y}$
\textit{semi-commute} if the families of their entries semi-commute. 

In the sequel, to simplify notation, we will abbreviate $(x,y)$ (resp. $(a,b)$) to $xy$ (resp. $ab$). Let 
$$\cl R_{\rm qc} = \{\left(U_{a,x}V_{b,y}\right)_{ab,xy} : U = (U_{a,x})_{a,x}, V = (V_{b,y})_{b,y}
\mbox{ semi-comm. isom.}\}.$$
We note that $\cl R_{\rm qc}$ is a resource over $(X Y,A B)$. 
Indeed, using standard leg notation, we have that 
$U\cdot V := \left(U_{a,x}V_{b,y}\right)_{ab,xy} = U_{1,3}V_{2,3}$
and, since the block operator matrices
$U_{1,3}$ and $V_{2,3}$ are isometries, so is the 
block operator matrix $\left(U_{a,x}V_{b,y}\right)_{ab,xy}$.  Closure under 
direct sums can be easily verified,
while the separation property follows from the inclusion 
$\cl R_{\rm loc}\subseteq \cl R_{\rm qc}$ and Corollary \ref{c_locre}. 
(We note the stronger inclusion $\cl R_{\rm q}\subseteq \cl R_{\rm qc}$.)

In the sequel, for clarity, we will denote by $v_{b,y}$ (resp. $f_{y,y',b,b'}$) the canonical 
generators of the TRO $\cl V_{Y,B}$ (resp. the C*-algebra $\cl C_{Y,B}$). 

\begin{theorem}\label{th_qcrep}
We have that $\mathsf{QC}(\cl R_{\rm qc}) = \cl Q_{\rm qc}$.
\end{theorem}

\begin{proof}
Let $\Gamma : M_{XY}\to M_{AB}$ be a quantum commuting QNS correlation, 
$H$ be a Hilbert space, $\xi\in H$ be a unit vector, and $(E_{x,x',a,a'})_{x,x',a,a'}$ and $(F_{y,y',b,b'})_{y,y',b,b'}$ be stochastic operator matrices with commuting entries, 
for which (\ref{eq_EFp}) holds.
By \cite[Theorem 5.2]{tt-QNS} there exist commuting representations $\pi^R_{X,A}:\cl C_{X,A}\to\cl B(H)$ and $\pi^R_{Y,B} : \cl C_{Y,B}\to\cl B(H)$ for which
$$\pi^R_{X,A}(e_{x,x',a,a'}) = E_{x,x',a,a'} \ \mbox{ and } \  \pi^R_{Y,B}(f_{y,y',b,b'}) = F_{y,y',b,b'}$$
for all $x,x'\in X$, $y,y'\in Y$, $a,a'\in A$ and $b,b'\in B$. 
By Corollary \ref{c_maxtose}, there exist Hilbert spaces $K$ and $L$, and semi-commuting ternary morphisms $\phi : \cl V_{X,A}\to \cl B(K,L)$ and $\psi : \cl V_{Y,B}\to \cl B(H,K)$ such that 
$$\la\phi(u_{a',x'})\psi(v_{b',y'})\xi,\phi(u_{a,x})\psi(v_{b,y})\xi\ra
= \la E_{x,x',a,a'}F_{y,y',b,b'}\xi,\xi\ra$$
for all $x,x'\in X$, $y,y'\in Y$, $a,a'\in A$ and $b,b'\in B$.
Using arguments as in the proof of Lemma \ref{c:maxnorm}, we can assume that 
$\overline{{\rm span}(\psi(\cl V_{Y,B})^*K)} = H$.
Let
$U = (\phi(u_{a,x}))_{a,x}$, $V = (\psi(v_{b,y}))_{b,y}$, and $\sigma$ be the vector 
state corresponding to $\xi$. Then $\Gamma = \Gamma_{U\cdot V,\sigma}$. 
It follows that $\cl Q_{\rm qc}\subseteq \mathsf{QC}(\cl R_{\rm qc})$.

For the reverse inclusion, let $\Gamma\in \mathsf{QC}(\cl R_{\rm qc})$.
By Lemma \ref{lifting}, there exists a $*$-homomorphism $\rho:\cl R_{\phi(\cl V_{X,A})}\rightarrow \cl R_{\psi(\cl V_{Y,B})}'\subseteq\cl B(H)$ satisfying $\rho(b)\psi(v)^*=\psi(v)^*b$ for all $b\in \cl R_{\phi(\cl V_{X,A})}$
and all $v\in \cl V_{Y,B}$. 
Then $\rho\circ\pi^R_{\phi}:\cl R_{X,A}\to\cl B(H)$ and $\pi^R_{\psi}:\cl R_{Y,B}\to\cl B(H)$ are $*$-homomorphisms with commuting ranges satisfying
$$(\rho\circ\pi^R_{\phi})(e_{x,x',a,a'})\pi^R_{\psi}(f_{y,y',b,b'})
= \psi(v_{b,y})^*\phi(u_{a,x})^*\phi(u_{a',x'})\psi(v_{b',y'}).$$
Relation (\ref{eq_EFp}) is then fulfilled with 
$$E_{x,x',a,a'} = \rho(\phi(u_{a,x})^*\phi(u_{a',x'})) \ \mbox{ and } \ 
F_{y,y',b,b'} = \psi(v_{b,y})^*\psi(v_{b',y'}),$$ 
showing that $\Gamma\in \cl Q_{\rm qc}$. The proof is complete.
\end{proof}

\begin{remark}\label{r_opspsqc}
\rm 
Let $\cl S_1^{A,X}\otimes_{\max}\cl S_1^{B,Y}$ be the operator space structure on 
the tensor product $\cl S_1^{A,X}\otimes \cl S_1^{B,Y}$, arising from the 
inclusion
$$\cl S_1^{A,X}\otimes_{\max}\cl S_1^{B,Y}\subseteq \cl V_{X,A}\otimes_{\rm tmax}\cl V_{Y,B}.$$
Since, by \cite[Theorem 3.1]{tt-QNS}, the pairs $(U,V)$ of semi-commuting block operator isometries
correspond precisely to semi-commuting ternary morphisms of $\cl V_{X,A}$ and $\cl V_{Y,B}$, 
Theorem \ref{t:maxnorm} and Proposition \ref{p:representations} imply that 
$$\cl S_1^{A,X}\otimes_{\max}\cl S_1^{B,Y} = \cl O_{XY,AB}^{\cl R_{\rm qc}}.$$
\end{remark}


\subsection{The no-signalling resource}\label{ss_ns}
Let $K$ be a Hilbert space and suppose
$\cl E = \{\eta_{x,y,a,b}:(x,y,a,b)\in XYAB\} \subseteq K$. 
Write 
$$\zeta_{x,y} = \sum_{a\in A} \sum_{b\in B} e_a\otimes e_b\otimes\eta_{x,y,a,b}, \ \ x\in X, y\in Y.$$
Recall that a PVM (projection-valued measure) is a family of projections summing to the identity operator on a Hilbert space. We call the family $\cl E$ \emph{no-signalling} if 
there exist PVM's 
$\{P_x:x\in X\}$ on $\mathbb{C}^{A}\otimes K$, and $\{Q_y:y\in Y\}$ on $\mathbb{C}^B\otimes K$, 
such that (employing a canonical reshuffling, we have)
\begin{equation}\label{eq_PxQyPVM}
\zeta_{x,y}\in\ran(P_x\ten I_B)\cap \ran (Q_y\ten I_A)\subseteq \mathbb{C}^{AB}\otimes K, 
\ \ \ x\in X, y\in Y.
\end{equation}
For a no-signalling family $\cl E$, let 
$U_{\cl E}:\mathbb C^{XY}\to \mathbb C^{AB}\otimes K$ be the operator, 
given by $U_{\cl E}(e_x\otimes e_y) = \zeta_{x,y}$, $x\in X$, $y\in Y$.
Since 
$$\langle U_{\cl E}(e_x\otimes e_y), U_{\cl E}(e_{x'}\otimes e_{y'})\rangle
= 
\langle \zeta_{x,y},\zeta_{x',y'}\rangle=\delta_{x,x'}\delta_{y,y'}, x,x'\in X, y,y'\in Y,$$ 
the operator $U_{\cl E}$ is an isometry. 
Set 
$$\cl R_{\rm ns} = \left\langle\{U_{\cl E} : \cl E \mbox{ no-signalling family}\}\right\rangle.$$
Corollary \ref{c_locre} implies that $\cl R_{\rm ns}$ is a resource over $(X Y,A B)$.  

\begin{proposition}\label{p_nsre=qns}
We have that $\mathsf{QC}(\cl R_{\rm ns}) = \cl Q_{\rm ns}$.
\end{proposition}

\begin{proof}
Let $\Gamma:M_{XY}\to M_{AB}$ be a QNS correlation. 
Since $\Gamma$ is completely positive, its Choi matrix $S\in M_{XYAB}$ is positive.
Let $K$ be a Hilbert space (of dimension $\rank{S}$) and 
$\{\eta_{x,y,a,b}\}_{x,y,a,b}\subseteq K$ be such that 
$S = (\langle\eta_{x',y',a',b'},\eta_{x,y,a,b}\rangle)_{x,x',a,a'}^{y,y',b,b'}$. 
We have that 
\[ \Gamma(\epsilon_{x,x'}\otimes \epsilon_{y,y'}) 
= \sum_{a,a'\in A} \sum_{b,b'\in B} \langle\eta_{x',y',a',b'},\eta_{x,y,a,b}\rangle 
\epsilon_{a,a'}\otimes \epsilon_{b,b'}. \]

Observe that the vectors 
$\zeta_{x,y} = \sum_{a\in A} \sum_{b\in B} e_a\otimes e_b\otimes \eta_{x,y,a,b}\in \mathbb C^{AB}\otimes K$ are orthonormal, since the trace-preservation of $\Gamma$ gives
\begin{align*}
    \delta_{x,x'}\delta_{y,y'}
    = \tr\Gamma(\epsilon_{x,x'}\otimes \epsilon_{y,y'})
    = \sum_{a\in A} \sum_{b\in B} \langle\eta_{x',y',a,b},\eta_{x,y,a,b}\rangle = \ip{\zeta_{x',y'}}{\zeta_{x,y}}.
\end{align*}
By (\ref{eq_qns1}), if $x\ne x'$ then
\[ 0=\tr_A\Gamma(\epsilon_{x,x'}\otimes \epsilon_{y,y'})
= \sum_{a\in A} \sum_{b,b'\in B} \langle\eta_{x',y',a,b'},\eta_{x,y,a,b}\rangle\eps_{b,b'},\]
so, for all $y,y'\in Y$ and all $b,b'\in B$, we have
\begin{equation}\label{eq_ebstar}
\ip{(I_A\otimes e_{b'}^*\otimes I_K)\zeta_{x',y'}}{(I_A\otimes e_{b}^*\otimes I_K)\zeta_{x,y}} = 0
\end{equation}
(we denote by $e_b^*$ the linear functional of taking an inner product against $e_b$). 
Hence, if $P_x$ is the projection onto the span of 
\[\{(I_A\otimes e_{b}^*\otimes I_K)\zeta_{x,y}:b\in B,\,y\in Y\}\subseteq \mathbb{C}^A\otimes K,\] then $P_x$ and $P_{x'}$ have orthogonal ranges for $x\ne x'$, and $\zeta_{x,y}\in \ran(P_x\otimes I_B)$, for all $x\in X$, $y\in Y$. Extending the range of one of these projections if necessary, we obtain a PVM on $\mathbb{C}^A\otimes K$. 
By symmetry, we can define a PVM $\{Q_y:y\in Y\}$ on $\mathbb{C}^B\otimes K$ with $\zeta_{x,y}\in \ran(Q_y\otimes I_A)$ for all $y\in Y$. This shows that ${\cl E}=\{\eta_{x,y,a,b}\}$ is non-signalling. 
In addition, $\Gamma = \Gamma_{U_{\cl E,1}}$; thus, 
$\cl Q_{\rm ns}\subseteq \mathsf{QC}(\cl R_{\rm ns})$.

For the converse direction, suppose that $K$ is a Hilbert space and 
$\cl E = \{\eta_{x,y,a,b}\}_{x,y,a,b}\subseteq K$ is a no-signalling family, and let 
$\{P_x\}_{x\in X}$ and $\{Q_y\}_{y\in Y}$ be PVM's, for which (\ref{eq_PxQyPVM}) holds. 
Let $\Gamma : M_{XY}\to M_{AB}$ be the completely positive map with Choi matrix
$(\langle\eta_{x',y',a',b'},\eta_{x,y,a,b}\rangle)_{x,x',a,a'}^{y,y',b,b'}$. 
Fix $x\neq x'$. Then
(\ref{eq_ebstar}) holds, implying that the vectors 
$(\eta_{x,y,a,b'})_{a\in A}$ and $(\eta_{x',y',a,b'})_{a\in A}$ are orthogonal, 
that is, $\sum_{a\in A} \langle\eta_{x',y',a,b'},\eta_{x,y,a,b}\rangle = 0$, 
for all $y,y'\in Y$ and all $b,b'\in B$.
It follows that (\ref{eq_qns1}) is satisfied; by symmetry, so is (\ref{eq_qns2}), that is, $\Gamma$ 
is no-signalling.
\end{proof}


\subsection{One-way local operations and classical communication}\label{ss_oneway}

In this subsection we consider the \emph{signalling} resource of 
local operations and one-way (left-to-right, or $(X,A)$-to-$(Y,B)$)
classical communication and obtain analogues of Proposition \ref{p-Rvalq} and Theorem \ref{th_wcb} in this context.
Recall \cite{clmow} that a bipartite quantum channel $\Gamma : M_{XY}\to M_{AB}$ 
belongs to the latter class, which will be denoted in the sequel by $\cl Q_{\rm lowc}$,
if there exist completely positive maps $\Psi_1,\dots,\Psi_n:M_X\to M_A$ such that 
$\sum_{i=1}^n \Psi_i$ is trace preserving, and quantum channels $\Phi_1,\dots,\Phi_n : M_Y\to M_B$ such that 
$\Gamma = \sum_{i=1}^n\Psi_i\ten\Phi_i$. 
Let 
$$\cl R_{\rm lowc}:=\langle  \{ [U_i\ten V_i]_{i=1}^n \}\rangle,$$
where $V_i \in \cl B(\bb{C}^Y,\bb{C}^{BT_i})$ is any isometry, $i\in [n]$,
the operators $U_i\in\cl B(\bb{C}^X,\bb{C}^{AS_i})$ are such that $\sum_{i=1}^n U_i^*U_i = I_X$, 
$S_i$ and $T_i$ are finite sets, $i = 1,\dots,n$, and $n\in \bb{N}$.

\begin{proposition}\label{p_lowc}
We have that $\cl R_{\rm lowc}$ is a resource over $(X Y,A B)$ and 
$\mathsf{QC}(\cl R_{\rm lowc}) = \cl Q_{\rm lowc}$.
\end{proposition}

\begin{proof} 
The family 
$\cl R_{\rm lowc}$ is by definition closed under direct sums and, since it contains 
$\cl R_{\rm loc}$, it is separating by Corollary \ref{c_locre}. 

Let $U\in\cl R_{\rm lowc}$, that is, 
$U = [U_i\ten V_i]_{i=1}^n$ as a column operator, where $U_i\in\cl B(\bb{C}^X,\bb{C}^{AS_i})$,
$\sum_{i=1}^n U_i^*U_i = I_X$, each $V_i \in \cl B(\bb{C}^Y,\bb{C}^{BT_i})$ is an isometry, 
and $S_i$ and $T_i$ are finite sets, $i = 1,\dots,n$.  
Set $K=\oplus_{i=1}^n\bb{C}^{S_i}\ten\bb{C}^{T_i}$; by (\ref{e:Gammap}),
$$\Gamma_{U,1}^{\#}(\rho) = (\id\ten\tr_K)(U\rho U^*), \ \ \ \rho\in M_{XY}.$$
Hence, 
letting $\Psi_i(\rho) = (\id\ten\tr_{S_i})(U_i\rho U_i^*)$ and $\Phi_i(\sigma) = (\id\ten\tr_{T_i})(V_i\rho V_i^*)$, 
for $\rho\in M_X$ and $\sigma\in M_Y$, we have 
$$\tightmath\Gamma_{U,1}^{\#}(\rho\ten \sigma)
 = \sum_{i=1}^n
(\id\ten\tr_{S_iT_i})
(U_i\ten V_i)
(\rho\ten\sigma) 
(U_i^*\ten V_i^*)
 =  
\sum_{i=1}^n\Psi_i(\rho)\ten\Phi_i(\sigma).$$
It follows that $\Gamma^{\#}_{U,1}$ and hence $\Gamma_{U,1}$ is one-way LOCC.

Conversely, given a one-way LOCC channel $\Gamma:M_{XY}\to M_{AB}$, 
there exist completely positive maps $\Psi_1,\dots,\Psi_n:M_X\to M_A$ such that $\sum_{i=1}^n \Psi_i$ is trace preserving, and quantum channels $\Phi_1,\dots,\Phi_n : M_Y\to M_B$ such that 
$\Gamma=\sum_{i=1}^n\Psi_i\ten\Phi_i$. 
Employing Stinespring dilations, we may write
$$\Psi_i(\rho)=(\id\ten\tr_{S_i})(U_i\rho U_i^*), \ \ \ \rho\in M_X,$$
where $U_i : \bb{C}^X\to\bb{C}^{AS_i}$ is a contraction and $S_i$ is a finite set. Since $\sum_{i=1}^n\Psi_i$ is trace preserving, 
it follows 
that $\sum_{i=1}^n U_i^*U_i = I_X$. As each $\Phi_i$ is a quantum channel, 
its Stinespring dilation is of the form
$$\Phi_i(\sigma)=(\id\ten\tr_{T_i})(V_i\sigma V_i^*), \ \ \ \sigma\in M_Y,$$
where $V_i:\bb{C}^Y\to\bb{C}^{BT_i}$ is an isometry and $T_i$ is a finite set.
Defining 
$$\overline{U}_i:\bb{C}^X\ni e_x\mapsto \sum_{a\in A} e_a\ten \overline{U}_{i,a,x}\in \bb{C}^{AS_i},$$
and similarly for $\overline{V}_i$
(with complex conjugations relative to standard bases), it follows as in the proof of Proposition \ref{p-Rvalq} that $\sum_{i=1}^n \overline{U}_i^*\overline{U}_i = I_X$, each 
$\overline{V}_i$ is an isometry and $\Gamma_{\overline{U},1}=\sum_{i=1}^n\Psi_i\ten\Phi_i$, where $\overline{U}= [ \overline{U}_i\ten \overline{V}_i]_{i=1}^n$. 
\end{proof}

We next identify the operator space structure
induced on $\cl O_{X,A}\ten\cl O_{Y,B}$ by the resource $\cl R_{\rm lowc}$. This structure is similar in nature to that induced by $\cl R_{\rm loc}$, but with additional $(2,C)$-summing \cite{junge,jp,jkpv} and operator $\ell_\infty$-module structures which arise from one-way classical communication. 

We first recall Junge's version \cite{junge} of absolutely $p$-summing maps (see also \cite{jp,jkpv}): 
given an operator space $\cl X$ and a Banach space $\cl Y$, a linear map 
$\phi : \cl X\to\cl Y$ is called \emph{$(p,{\rm cb})$-summing} if
$$\pi_{(p,{\rm cb})}(\phi) := \norm{\id_{\ell_p}\ten\phi:\ell_p\ten_{\min}\cl X\to\ell_p(\cl Y)}<\infty,$$
where $\ell_p$ is equipped with Pisier's interpolated operator space structure \cite{pisier_intr}.
It was shown in \cite[Proposition 4.1]{jkpv} that the $(1,{\rm cb})$-summing norm expresses 
the one-way classical communication bias for quantum XOR games. This result, in part, motivated
Theorem \ref{th_lowc1} below.
We will be interested in the case $p = 2$, and will (as in Subsection \ref{ss_loc})
denote by $C$ the column operator space structure on $\ell_2$. For a linear map 
$\phi : \cl X\to \cl Y$, set 
\begin{equation}\label{eq_2Cnormde}
\pi_{(2,C)}(\phi) :=\norm{\id_{C}\ten\phi:C\ten_{\min}\cl X\to\ell_2(\cl Y)}
\end{equation}
and write 
$$\Pi_{(2,C)}(\cl X,\cl Y) = \{\phi : \cl X\to \cl Y \ : \ \mbox{linear and } \pi_{(2,C)}(\phi) < \infty\}.$$
Note that we only consider the Banach space structure of $\ell_2(\cl Y)$ in the definition (\ref{eq_2Cnormde}). 
The appropriate matricial amplifications of $\pi_{(2,C)}$ in this context relies on a particular operator $\ell_\infty$-module quantisation of $\ell_2(\cl Y)$. Note that $\ell_2(\cl Y)$ admits a canonical (algebraic) $\ell_\infty$-module structure via pointwise operations. 

\begin{proposition}\label{p:opmodule} Let $\cl X$ be an operator space. Define
\begin{equation}\label{e:opmodule}\norm{[(x^k_{i,j})]}_{M_n(\ell_2(\cl X))}:=\sup\{\norm{[(\phi_k(x^k_{i,j}))]}_{M_n((\ell_2\ten H)_c)}\},\end{equation}
where the supremum is over all Hilbert spaces $H$ and all complete contractions $\phi_k:\cl X\to H_c$, $k\in\N$. Then $(\ref{e:opmodule})$ equips $\ell_2(\cl X)$ with an operator $\ell_\infty$-module structure (under the canonical action).
\end{proposition}

\begin{proof} By linearity of each $\phi_k$ in the definition $(\ref{e:opmodule})$, for each $[(x^k_{i,j})]\in M_m(\ell_2(\cl X))$, $f=[f_{i,j}]\in M_{n,m}(\ell_\infty)$ and $\beta\in M_{m,n}(\bb{C})$, it follows that

\begin{align*}\norm{f\cdot[(x^k_{i,j})]\cdot\beta}_{M_n(\ell_2(\cl X))}&=\bignorm{\bigg[\bigg(\sum_{r,s}f_{i,r}(k)x^k_{r,s}\beta_{s,j}\bigg)\bigg]_{i,j}}_{M_n(\ell_2(\cl X))}\\
&=\sup_{(\phi_k)}\bigg\{\bignorm{\bigg[\bigg(\sum_{r,s}f_{i,r}(k)\phi_k(x^k_{r,s})\beta_{s,j}\bigg)\bigg]_{i,j}}_{M_n((\ell_2\ten H)_c)}\bigg\}\\
&=\sup_{(\phi_k)}\{\norm{f\cdot[(\phi_k(x^k_{i,j}))]\cdot \beta}_{M_n((\ell_2\ten H)_c)}\}\\
&\leq\sup_{(\phi_k)}\{\norm{f}\norm{[(\phi_k(x^k_{i,j}))]}_{M_m((\ell_2\ten H)_c)}\norm{\beta}\}\\
&=\norm{f}\norm{[(x^k_{i,j})]}_{M_m(\ell_2(\cl X))}\norm{\beta}.
\end{align*}
Next, given $x=[(x^k_{i,j})]\in M_m(\ell_2(\cl X))$ and $y=[(y^k_{i,j})]\in M_n(\ell_2(\cl X))$, we clearly have
\begin{align*}\norm{x\oplus y}_{M_{m+n}(\ell_2(\cl X))}&=\sup_{(\phi_k)}\{\norm{[\phi_k(x^k_{i,j})]\oplus[\phi_k(y^k_{i,j})]}_{M_{m+n}((\ell_2\ten H)_c)}\}\\
&\leq\norm{x}_{M_m(\ell_2(\cl X))}+\norm{y}_{M_{n}(\ell_2(\cl X))}.
\end{align*}
It follows from \cite[Proposition 2.3.6]{er} that $(\ref{e:opmodule})$ defines an operator  $\ell_\infty$-module structure on $\ell_2(\cl X)$. 
\end{proof}

\begin{remark}\label{r:injection} For any operator space $\cl X$ and $k\in\N$, it follows that the embedding $\cl X\ni x\mapsto e_k\ten x\in\ell_2(\cl X)$ is completely contractive: 
\begin{align*}\norm{[e_k\ten x_{i,j}]}_{M_n(\ell_2(\cl X))}&=\sup\{\norm{[\phi(x_{i,j})]}_{M_n(H_c)} : \phi\in{\rm CC}({\cl X},H_c)\}\\
&\leq\sup\{\norm{[\phi(x_{i,j})]}_{M_n(\cl B(H,K))} : \phi\in{\rm CC}({\cl X},\cl B(H,K))\}\\
&=\norm{[x_{i,j}]}_{M_n(\cl X)}
\end{align*}
\end{remark}

Given Hilbert spaces $H$ and $K$, we will furnish the space of Hilbert-Schmidt operators from $H$ to $K$ with the operator space structure arising from the isometric identification $\cl S_2(H,K) = {\rm CB}(H_r,K_c)$ \cite[Corollary 4.5]{er91}. Note that this coincides with the column Hilbert space structure on $\cl S_2(H,K)=H^*\ten K$ as
$$(H^*\ten K)_c=(H^*)_c\ten_{\min}K\subseteq {\rm CB}(H_r,K_c).$$
Recall that the row $R=(\ell_2)_r$ and column $C=(\ell_2)_c$ Hilbert spaces admit a canonical $\ell_\infty$-module structure.

Let $\cl X$ and $\cl Y$ be operator spaces, and let $\phi\in \Pi_{(2,C)}(\cl X,\cl Y)$.  We define the \textit{$cb$-right Hilbert-Schmidt norm} of the map $\phi$ by letting
\begin{equation}\label{eq_rightcb}
\norm{\phi}_{r,{\rm cb}} :=
\sup\{\norm{ \beta\circ(\id\otimes\phi)\circ \alpha}_{\cl{S}_2(\ell_2,\ell_2\ten\ell_2)},\}
\end{equation}
where the supremum is over all completely contractive $\ell_\infty$-module maps $\alpha:R\to C\ten_{\min}\cl X$ and $\beta:\ell^2(\cl Y)\to(\ell_2\ten\ell_2)_c$, where $\ell_2(\cl Y)$ is equipped with the operator $\ell_\infty$-module structure from Proposition \ref{p:opmodule}, and the module actions on the codomains are on the first tensor leg. In other words, $\norm{\phi}_{r,{\rm cb}}$ is the supremum of the Hilbert-Schmidt norms of the compositions
\begin{equation*}
\begin{tikzcd}
\ell_2\arrow[r]\arrow[d, " \alpha"] & \ell_2\ten\ell_2\\
C\ten_{\min} \cl X\arrow[r, "\id\ten\phi"] & \ell^2(\cl Y)\arrow[u, " \beta"],
\end{tikzcd}
\end{equation*}
for completely contractive $\ell_\infty$-module maps $\alpha$ and $\beta$. We note that, since the supremum in (\ref{eq_rightcb}) is that of Hilbert-Schmidt norms, 
it can be taken over all maps $\alpha$ and $\beta$ with the properties that $\alpha(R_l^\perp)=0$, $\alpha|_{R_k}:R_k\to C_k\ten_{\min} \cl X$ and $\beta:\ell_2(\cl Y)\to C_k\ten_{\min} C_l$ for some $k,l\in \bb{N}$.
We let 
$$\cl{S}_2^{r,{\rm cb}}(\cl X,\cl Y) = \{\phi\in \Pi_{(2,C)}(\cl X,\cl Y) : \norm{\phi}_{r,{\rm cb}} < \infty\}.$$
Given a matrix $[\phi_{i,j}]_{i,j}\in M_n(\cl{S}_2^{r,{\rm cb}}(\cl X,\cl Y))$, and letting
$\phi : \cl X\rightarrow M_n(\cl Y)$ be the associated map, we set
$$\norm{[\phi_{i,j}]}_{r,{\rm cb}}^{(n)}
= \sup\{\|\beta^{(n)}\circ (\id\otimes\phi)\circ \alpha : R \to M_n((\ell_2\ten \ell_2)_c)\|_{\rm cb}\},$$
where, as before, $\alpha:R\to C\ten_{\min}\cl X$ and $\beta:\ell^2(\cl Y)\to(\ell_2\ten\ell_2)_c$ are completely contractive $\ell_\infty$-module maps. For finite-dimensional spaces $\cl X$ and $\cl Y$, 
we write 
$\cl X^*\ten_{r,{\rm cb}}\cl Y := \cl{S}_2^{r,{\rm cb}}(\cl X,\cl Y)$, 
and consider $\|\cdot\|_{r,{\rm cb}}^{(n)}$ as a norm on $M_n(\cl X^*\ten_{r,{\rm cb}}\cl Y)$.

\begin{theorem}\label{th_lowc1}
Let $X,Y,A,B$ be finite sets. We have that 
$\cl O_{XY,AB}^{\cl R_{\rm lowc}} = \cl{S}_1^{A,X}\ten_{r, {\rm cb}}\cl{S}_1^{B,Y},$
up to a complete isometry. 
\end{theorem}

\begin{proof}
Let $m\in \bb{N}$, $S_i$ and $T_i$ be finite disjoint sets, and
$U_i\in\cl B(\bb{C}^X,\bb{C}^{AS_i})$ and $V_i \in \cl B(\bb{C}^Y,\bb{C}^{BT_i})$ be such that 
$\sum_{i=1}^m U_i^*U_i = I_X$ and $V_i$ is an isometry, $i\in [m]$. 
Write $W = [U_i\ten V_i]_{i=1}^m$, considered as a column isometry. 
Write, further, $U_i = [U_{i,s}]_{s\in S_i}$, where $U_{i,s} \in M_{A,X}$, and 
$V_i = [V_{i,t}]_{t\in T_i}$, where $V_{i,t} \in M_{B,Y}$. Set $S = \sqcup_{i=1}^m S_i$, and make the identification $[U_i]_{i=1}^m = U := [U_s]_{s\in S}$, 
where, for $s\in S$, we have let $U_s = U_{i(s),s}$, where $i(s)\in[m]$ is the unique element for which $s\in S_i$; 
we note that $U : \bb{C}^X\to \bb{C}^{AS}$ is an isometry.
Observe that
\begin{equation*}\label{eq_splitsW}
\|\phi_W^{(n)}(\omega)\| =  \left\|\left[(\phi_{U_i}\otimes\phi_{V_i})^{(n)}(\omega)\right]_{i=1}^m\right\|_{M_n((\oplus_{i=1}^m\bb{C}^{S_i}\ten\bb{C}^{T_i})_c)},\end{equation*}
for every $\omega\in M_n(\cl S_1^{A,X}\otimes \cl S_1^{B,Y})$. 

Define
$$\alpha:R_S\ni e_{s}\mapsto e_{s}\ten U_{s}\in C_S\ten_{\min}M_{X,A}.$$
Then $\alpha$ is an $\ell_\infty$-module map, and, as $\sum_{s\in S}U_{s}^*U_{s}=1$, it is completely contractive by (\ref{e:tensor}). Letting $T=\sqcup_{i=1}^m T_i$ and embedding $\bb{C}^T$ into the first $|T|$-summands of $C$, for each $s\in S$, define $\beta_{s} : \cl S_1^{B,Y}\to C$ by $\beta_{s}(\rho) = \sum_{t\in T_i} \langle \rho, V_{i,t}\rangle e_{t}$. Let $\beta : \ell_2(\cl S_1^{B,Y}) \to (\ell_2\ten\ell_2)_c=C\ten_{\min} C$ be the $\ell_\infty$-module map, supported on the first $|S|$ entries of the domain,
given by $\beta((\rho_s)_{s\in S}) = (\beta_{s}(\rho_s))_{s\in S}$. The fact that $V_i$ is an isometry for each $i\in [m]$ implies that $\beta_{s} : \cl S_1^{B,Y}\to C$ is a complete contraction
for every $s\in S$; thus, $\beta$ is a complete contraction by (\ref{e:opmodule}). Using (\ref{e:column}) and similar arguments to the proof of Theorem \ref{th_wcb}, for every $[\omega_{k,l}]\in M_n(\cl S_1^{A,X}\otimes \cl S_1^{B,Y})$ we have
\begin{align*}\|\phi_W^{(n)}([\omega_{k,l}])\| &=  \left\|\left[(\phi_{U_i}\otimes\phi_{V_i})^{(n)}([\omega_{k,l}])\right]_{i=1}^m\right\|_{M_n((\oplus_{i=1}^m\bb{C}^{S_i}\ten\bb{C}^{T_i})_c)}\\
&=\bignorm{\bigg[\sum_{j=1}^n\sum_{s\in S}\sum_{t\in T_{i(s)}}\la \om_{j,l}(U_{s}),V_{i(s),t}\ra\overline{\la\om_{j,k}(U_{s}),V_{i(s),t}\ra}\bigg]_{k,l}}\\
&=\bignorm{\bigg[\sum_{j=1}^n\sum_{s\in S}\sum_{t\in T_{i(s)}}\la \om_{j,l}(\alpha(e_{s})),\beta_s^*(e_t)\ra\overline{\la\om_{j,k}(\alpha(e_{s})),\beta_s^*(e_{t})\ra}\bigg]_{k,l}}\\
&=\bignorm{\bigg[\sum_{j=1}^n\sum_{s\in S}\sum_{t\in T_{i(s)}}\la \beta_s(\om_{j,l}(\alpha(e_{s}))),e_t\ra\la e_t,\beta_s(\om_{j,k}(\alpha(e_{s})))\ra\bigg]_{k,l}}\\
&=\bignorm{[\sum_{j=1}^n\la \beta\circ(\id\ten\om_{j,l})\circ\alpha,\beta\circ(\id\ten\om_{j,k})\circ\alpha\ra_{\cl S_2(\ell_2,\ell_2\ten\ell_2)}]_{k,l}}^{1/2}\\
&=\norm{\beta^{(n)}\circ(\id\ten[\om_{k,l}])\circ\alpha}_{M_n(\cl S_2(\ell_2,\ell_2\ten\ell_2)_c)}\\
&=\norm{\beta^{(n)}\circ(\id\ten[\om_{k,l}])\circ\alpha}_{\rm cb},
\end{align*}
where the final equality is due to the identifications
$$M_n(\cl S_2(\ell_2,\ell_2\ten\ell_2)_c)=M_n({\rm CB}(R,(\ell_2\ten\ell_2)_c))={\rm CB}(R,M_n((\ell_2\ten\ell_2)_c)).$$
We therefore have $\|\omega\|_{\cl R_{\rm lowc}}^{(n)} \leq \|\omega\|_{r, {\rm cb}}^{(n)}$. 

For the reverse inequality, let $k,l\in\mathbb N$, $\alpha : R_k\to C_k\ten_{\min}M_{A,X}$ and 
$\beta :\ell_2(\cl{S}_1^{B,Y}) \to C_k\ten_{\min} C_l$ be completely contractive $\ell_{\infty}$-module maps. Then, by (\ref{e:tensor}) and the module property, $\alpha$ is necessarily of the form $e_i\mapsto e_i\ten U_i$ for a column contraction $[U_i]_{i=1}^k$. 
As in the proof of Theorem \ref{th_wcb}, complete the column contraction $[U_i]_{i=1}^k$ to a column isometry $[U_i]_{i=1}^{k+k_0}$, and let $\tilde{\alpha}:R_{k+k_0}\to C_{k+k_0}\ten_{\min}M_{A,X}$ be the associated completely contractive $\ell_\infty$-module map.

Next, $\beta$ is necessarily supported on the first $k$ summands of $\ell_2(\cl{S}_1^{B,Y})$, and by Remark \ref{r:injection}, the associated maps 
$$\beta_i:=(e_i^*\ten\id)\circ\beta|_{\bb{C}e_i\ten\cl{S}_1^{B,Y}}:\cl S_1^{B,Y}\to C_l$$
are completely contractive, $i\in [k]$. Letting $V_i = [V_{i,j}]_{j=1}^l$ be the operator, canonically corresponding to $\beta_i$, we have that $V_i$ is a column contraction for every $i\in [k]$. Complete $V_i$ to a column isometry $\tilde{V_i}$ by adding the operator terms $\tilde{V}_{i,1},\dots, \tilde{V}_{i,l_0}\in M_{X,A}$, $i\in [k]$ 
Let $\tilde{\beta}_i : \cl S_1^{B,Y}\to C_{l+l_0}$ be the map, corresponding to $\tilde{V}_i$, $i\in [k]$. For $i\in\{k+1,...,k+k_0\}$, choose a complete contraction $\tilde{\beta_i}:\cl S_1^{B,Y}\to C_{l+l_0}$ governed by column isometries. Then $\tilde{\beta}:\ell_2(\cl S_1^{B,Y})\to C_{k+k_0}\ten_{\min} C_{l+l_0}$ is a completely contractive $\ell_\infty$-module map, and it follows that
\begin{eqnarray*}
\norm{\beta^{(n)}\circ(\id\otimes\omega)\circ\alpha}_{\rm cb}
& \leq & 
\norm{\tilde\beta^{(n)}\circ(\id\otimes\omega)\circ\tilde\alpha}_{\rm cb}\\
& = & 
\left\|\left[(\phi_{U_i}\otimes\phi_{{\tilde V}_i})^{(n)}(\omega)\right]_{i=1}^{k+k_0}\right\|
\leq
\|\omega\|_{\cl R_{\rm lowc}}.
\end{eqnarray*}
Hence, $\|\omega\|_{r,{\rm cb}}\leq \|\omega\|_{\cl R_{\rm lowc}}$, and
the proof is complete. 
\end{proof}

In the remainder of this subsection we 
provide an upper bound on the matricial norms induced by the resource $\cl R_{\rm lowc}$ using the theory of operator sequence spaces \cite{lam};
we review the necessary prerequisites, and 
refer the reader to \cite[\S2--\S3]{lnr} for a detailed overview.

An \textit{operator sequence space} is a Banach space $\cl X$ equipped with a family $(\norm{\cdot}_{\hat{n}})$ of norms on direct sums $\cl X^n$ satisfying
\begin{enumerate}
\item $\norm{[x \  0]^{\rm t}}_{m\hat{+}n} = \norm{x}_{\hat{m}}$, $m\in\N$, $x\in\cl X^m$;
\item $\norm{[x \ y]^{\rm t}}_{\hat{m+n}}\leq\norm{x}_{\hat{m}}+\norm{y}_{\hat{n}}$, $m,n\in\N$, $x\in\cl X^m$, $y\in\cl X^n$, and
\item $\norm{\alpha x}_{\hat{m}}\leq\norm{\alpha}\norm{x}_{\hat{n}}$, $m,n\in\N$, $\alpha\in M_{m,n}(\bb{C})$, $x\in\cl X^n$.
\end{enumerate}
We let $\cl X^{\hat{n}}$ denote the normed space $(\cl X^n,\norm{\cdot}_{\hat{n}})$, $n\in\N$, and 
we write $(\cl X^{\hat{n}})_{n\in \bb{N}}$ for the corresponding operator sequence space. 
If $\cl X$ is a Banach space, the family $(\ell_2^n(\cl X))_{n\in \bb{N}}$, 
where each term is equipped with its canonical norm, defines an operator sequence space
$(\ell_2^{\hat{n}}(\cl X))_{n\in \bb{N}}$.
Further, if $\cl X$ is an operator space, the family 
$C(\cl X) := (C(\cl X)^{\hat{n}})_{n\in \bb{N}}$ of column spaces $M_{n,1}(\cl X)$ is an operator sequence space. 

A linear map $\phi:\cl X\to\cl Y$ between operator sequence spaces is \textit{sequentially bounded} if 
$$\norm{\phi}_{\rm sb}:=\sup_{n\in\N}\norm{\phi^{\hat{n}}:\cl X^{\hat{n}}\to\cl Y^{\hat{n}}} < \infty,$$
where $\phi^{\hat{n}}([x_j]) = [\phi(x_j)]$, $[x_j]\in\cl X^{\hat{n}}$. 
We let ${\rm SB}(\cl X,\cl Y)$ denote the space of sequentially bounded linear maps,
and note that ${\rm SB}(\cl X,\cl Y)$ is a Banach space under the norm
$\norm{\cdot}_{\rm sb}$. Note that if $\cl X$ is an operator space and $\cl Y$ is a Banach space, then, 
for a linear map $\phi:\cl X\to\cl Y$, we have
\begin{equation}\label{eq_seqs}
\pi_{(2,C)}(\phi) 
= \norm{\phi : C(\cl X) \to (\ell_2^{\hat{n}}(\cl X))_{n\in \bb{N}}}_{\rm sb}.
\end{equation}

The \emph{minimal quantisation} $\min(\cl X)$ of an operator sequence space 
$\cl X$ is the operator space defined via $M_n(\min(\cl X)):= \cl B(\ell_2^n,\cl X^{\hat{n}})$, $n\in\N$ \cite{lam} (see also \cite[Definition 3.1]{lnr}). It was shown in \cite[Satz 4.1.6]{lam} that 
if $\cl X$ is an operator space and $\cl Y$ is an operator sequence space then 
${\rm CB}(\cl X,\min(\cl Y)) = {\rm SB}(C(\cl X),\cl Y)$ isometrically. Thus, in view of (\ref{eq_seqs}), 
$$\pi_{(2,C)}(\phi)=\norm{\phi:\cl X\to\min(\cl Y)}_{\rm cb},$$
where we have set $\min(\cl Y) = \min((\ell_2^{\hat{n}}(\cl Y))_{n\in \bb{N}})$ for brevity. 
We therefore obtain a natural operator space structure on the space $\Pi_{(2,C)}(\cl X,\cl Y)$ of $(2,C)$-summing maps via 
$$\pi_{(2,C)}([\phi_{i,j}])=\norm{[\phi_{i,j}]:\cl X\to M_n(\min(\cl Y))}_{\rm cb}, \ \ \ [\phi_{i,j}]\in M_n(\Pi_{(2,C)}(\cl X,\cl Y)).$$

\begin{theorem}\label{th_lowcboundsq}
Let $X,Y,A,B$ be finite sets. Then the formal identity map
$\cl{S}_1^{A,X}\ten_{(2,C)}\cl{S}_1^{B,Y} \to \cl O_{XY,AB}^{\cl R_{\rm lowc}}$ is an isometric 
complete contraction. 
\end{theorem}

\begin{proof} 
Fix $[\om_{k,l}]\in M_n(\cl{S}_1^{A,X}\ten\cl{S}_1^{B,Y})$. 
Let $U\in\cl R_{\rm lowc}$ have the form 
$W = [U_i\ten V_i]_{i=1}^{k_U}$, where $U_i\in\cl B(\bb{C}^X,\bb{C}^{AS_i})$ are such that
$\sum_{i=1}^{k_U} U_i^*U_i = I_X$,  and each $V_i \in \cl B(\bb{C}^Y,\bb{C}^{BT_i})$ 
is an isometry; here $S_i$ and $T_i$ are finite sets, $i\in [k_U]$.
Write $U_i$ as the column $[U_{i,s}]$, $s\in S_i$, with $U_{i,s}\in M_{A,X}$. 
Similarly, view $V_i$ as the column $[V_{i,t}]$, $t\in T_i$, with $V_{i,t}\in M_{B,Y}$. 
It follows as in the proof of Theorem \ref{th_wcb} that
$$\phi_W(\om_{k,l}) = \bigoplus_{i=1}^{k_U}\sum_{s\in S_i}\sum_{t\in T_i}\la U_{i,s}\ten V_{i,t},\om_{k,l}\ra e^i_s\ten e^i_t\in\bigoplus_{i=1}^{k_U}\bb{C}^{S_i}\ten\bb{C}^{T_i} =: H.$$
Let $\beta_i\in {\rm CC}(R_{T_i},M_{B,Y})$ be the linear map, given by 
$\beta_i(e^i_t) = V_{i,t}$
(see the proofs of Theorem \ref{th_wcb}); here, we let $(e^i_t)_{t\in T_i}$ be the canonical basis of $\ell_2^{T_i}$,
$i\in [k_U]$. 
Then
\begin{equation}\label{eq_UisVit}
\la U_{i,s}\ten V_{i,t},\om_{k,l}\ra = \la\om_{k,l}(U_{i,s}),V_{i,t}\ra = \la\beta_i^*(\om_{k,l}(U_{i,s})),e^i_t\ra.
\end{equation}
Since the column operator space structure on $H = \cl B(\bb{C},H)$ is given via the identifications
$M_n(H) = \cl B(\bb{C}^n,H^n)$, using (\ref{eq_UisVit}) and recalling that $\bb{C}^n_{1}$ 
denotes the unit ball of the Hilbert space $\bb{C}^n$, we have
\begin{align*}
& \norm{[\phi_U(\om_{k,l})]}^2
=
\sup\left\{\sum_{k=1}^n\bignorm{\sum_{l=1}^n\phi_U(\om_{k,l})\xi_l}_H^2 : \xi\in\bb{C}^n_{1}\right\}\\
&=
\sup\left\{\sum_{k=1}^n\bignorm{\sum_{l=1}^n\sum_{i,s,t}\xi_l \la\beta_i^*(\om_{k,l}(U_{i,s})),e^i_t\ra e^i_s\ten e^i_t}_H^2 : \xi\in\bb{C}^n_{1}\right\}\\
&=
\sup\left\{\sum_{k=1}^n\bignorm{\sum_{l=1}^n\sum_{i,s}\xi_l  e^i_s\ten \beta_i^*(\om_{k,l}(U_{i,s}))}_{H}^2 : \xi\in\bb{C}^n_{1}\right\}\\
&=
\sup \left\{\sum_{k=1}^n\sum_{i,s}\bignorm{\sum_{l=1}^n\xi_l\beta_i^*(\om_{k,l}(U_{i,s}))}_{\bb{C}^{T_i}}^2 : \xi\in\bb{C}^n_{1}\right\}\\
&\leq
\sup \left\{\sum_{k=1}^n\sum_{i,s}\bignorm{\sum_{l=1}^n\xi_l\om_{k,l}(U_{i,s})}_{\cl{S}_1^{B,Y}}^2 : \xi\in\bb{C}^n_{1}\right\}.
\end{align*}
Viewing the columns $[U_{i,s}]$ and $[\xi_l]$ as square and rectangular matrices, respectively (by adding zeros), it follows that
$$\norm{[\phi_U(\om_{k,l})]}^2\leq\sup\bigg\{\sum_{i=1}^m\sum_{k=1}^n\bignorm{\sum_{j=1}^m\sum_{l=1}^n\xi_{j,l}\om_{k,l}(T_{i,j})}_{\cl{S}_1^{B,Y}}^2\bigg\},$$
where the supremum runs over all $\xi\in\bb{C}^{mn}_{1}$, 
all operators $T$ in the unit ball $M_m(\cl B(\bb{C}^X,\bb{C}^A))_{1}$ and all $m\in\N$. 
By the definition of the minimal quantisation of the operator sequence space 
$(\ell_2^{\hat{n}}(\cl S_1^{B,Y}))_{n\in \bb{N}}$, the latter supremum coincides with 
\begin{align*}&\sup\{\norm{[\om_{k,l}(T_{i,j})]}^2_{M_m(M_n(\min(\cl S_1^{B,Y})))} : T\in M_m(\cl B(\bb{C}^X,\bb{C}^A))_{1},  m\in\N\}\\
&= \pi_{(2,C)}([\om_{k,l}])^2.
\end{align*}
Since $U$ was an arbitrary element of $\cl R_{\rm lowc}$, 
we have that $\norm{[\om_{k,l}]}_{\cl R_{\rm lowc}}\leq \pi_{(2,C)}([\om_{k,l}])$.

When $n=1$, by above we have

$$\pi_{(2,C)}(\om)^2=\sup\bigg\{\sum_{i=1}^m\bignorm{\sum_{j=1}^m\xi_j\om(T_{i,j})}_{\cl S_1^{B,Y}}^2\bigg\}.$$

Since $\xi$ and $T$ are in the respective unit balls, it follows that $U_i:=\sum_{j=1}^m \xi_j T_{i,j}$ satisfies $\sum_{i=1}^m U_i^*U_i\leq 1$. For each $i$, pick $V_i\in\cl B(\bb{C}^B,\bb{C}^Y)$ of norm 1 for which $|\la V_i,\om(U_i)\ra|=\norm{\om(U_i)}_{\cl S_1^{B,Y}}$. Let $\alpha\in{\rm CC}(R_m,C_m\ten_{\min}B(\bb{C}^X,\bb{C}^A))$ be the $\ell_\infty$-module map corresponding to the tuple $[U_i]_{i=1}^m$ and let $\beta_i\in\mathrm{CC}(\cl S_1^{B,Y},C_1)$ correspond $V_i$. It follows that
$$\sum_{i=1}^m\bignorm{\sum_{j=1}^m\xi_j\om(T_{i,j})}_{\cl S_1^{B,Y}}^2=\sum_{i=1}^m|\la V_i,\om(U_i)\ra|^2=\norm{\beta\circ(\id\ten\om)\circ\alpha}^2_{\cl S_2(\ell_2,\ell_2\ten\ell_2)}$$
where $\beta:\ell_2(S_1^{B,Y})\to C_m \ten_{\min}C_1$ is the completely contractive $\ell_\infty$-module map corresponding to the family $(\beta_i)$. It follows from Theorem \ref{th_lowc1} that $\pi_{(2,C)}(\om)^2\leq\|\omega\|_{\cl R_{\rm lowc}}^2$.
\end{proof}


\subsection{Values of non-local games}\label{ss_nonlogam}

We will be concerned with two classes of quantum non-local games which we now introduce.
A \emph{projection quantum game} over the quadruple $(X,Y,A,B)$ is a pair $(\xi,P)$, where 
$\xi\in \bb{C}^{XY}\otimes \ell_2^R$
and $P$ is a projection in $M_{AB}\otimes \cl B(\ell_2^R)$. Here, 
$R$ is an additional (perhaps countably infinite)
set. 
In the case
${\rm rank}(P) = 1$, this definition reduces to that of 
\emph{rank one quantum games} from \cite{junge}; a projection quantum game $(\xi,P)$ with 
${\rm rank}(P) < \infty$ will be called a \emph{finite rank quantum game}.
A \emph{hypergraph quantum game} over $(X,Y,A,B)$ is 
a probabilistic quantum hypergraph $(\nph,\pi)$, where $\nph : \bb P_{XY}\to \cl P_{AB}$.



We note that hypergraph quantum games constitute a quantisation of classical non-local games. 
Indeed, recall that a \emph{classical non-local game} is a tuple 
$(X,Y,A,B,\lambda,\pi)$, where 
$\lambda : X Y \times A
 B\to \{0,1\}$ is the \emph{rule function} of the game and 
$\pi$ is a probability measure on $X Y$. 
Given a rule function $\lambda : X Y\times A B\to \{0,1\}$ and an element $(x,y)\in X Y$, 
let 
$$E_{(x,y)} = \{(a,b)\in A B : \lambda(x,y,a,b) = 1
\}
;$$
the classical non-local game with rule function $\lambda$ give rise to the 
map $\nph_{\lambda} : \bb P_{XY}^{\rm cl} \to \cl P_{AB}^{\rm cl}$, given by 
\begin{equation}\label{eq_nphla}
\nph_{\lambda}(\epsilon_{x,x}\otimes \epsilon_{y,y}) = P_{E_{(x,y)}},  \ \ \ (x,y)\in X Y.
\end{equation}

If $(\varphi,\pi)$ is a hypergraph quantum game, and 
$(\xi,P)$ is a finite rank quantum game, where 
$\xi\in \bb{C}^{XY}\otimes \ell_2^R$ and $P$ is a projection in $M_{AB}\otimes \cl B(\ell_2^R)$, 
we let 
$$\omega_{\rm t}(\xi,P) = \omega_{\cl R_{\rm t}}(\xi,P) 
\ \mbox{ and } \ \omega_{\rm t}(\varphi,\pi) = \omega_{\cl R_{\rm t}}(\varphi,\pi),$$ 
for ${\rm t}\in \{{\rm loc}, {\rm lowc}, {\rm q}, {\rm qc}, {\rm ns}\}$.
By (\ref{eq_Qchain}), the following inequalities holds true:
$$\omega_{\rm loc}(\xi,P)\leq \omega_{\rm q}(\xi,P)\leq \omega_{\rm qc}(\xi,P)
\leq \omega_{\rm ns}(\xi,P)$$
(resp.,
$$\omega_{\rm loc}(\varphi,\pi)\leq \omega_{\rm q}(\varphi,\pi)\leq \omega_{\rm qc}(\varphi,\pi)
\leq \omega_{\rm ns}(\varphi,\pi)).$$

If $\cl E = \{\eta_{x,y,a,b}\}_{x,y,a,b}$ is a no-signalling family
of vectors in a Hilbert space $K$, we let 
$\tilde{\cl E} : \cl S_1^{A,X}\otimes \cl S_1^{B,Y}\to K$ be the map, given by 
$\tilde{\cl E}(\epsilon_{x,a} \otimes \epsilon_{y,b}) = \eta_{x,y,a,b}$.

Item (ii) in the next theorem is an extension of \cite[Theorem 3.2 (1)]{junge}.

\begin{theorem}\label{th_qcval}
Let $X$, $Y$, $A$ and $B$ be finite sets, and $R$ a countable set. 
Let $\xi\in \bb{C}^{XY}\otimes\ell_2^R$ be a unit vector, and 
$P = \sum_{k = 1}^{\infty} \lambda_k \gamma_k\gamma_k^*$, where 
$(\gamma_k)_{k=1}^n\subseteq \bb{C}^{AB}\otimes\ell_2^R$ is an orthonormal family and $0\leq \lambda_k\leq 1$, 
$k\in \bb{N}$. Set 
$\rho_k = \sqrt{\lm_k}\,\overline{ \Tr_R(\xi\gamma_k^*)}$, $k \in \bb{N}$, and 
write $\rho = [\rho_k]_{k\in\bb{N}}\in M_{\infty,1}(\cl S_1^{A,X}\otimes \cl S_1^{B,Y})$.
Then
\begin{itemize}
\item[(i)]
$\om_{\rm loc}(\xi,P) = \norm{\rho}_{w,{\rm cb}}^2$;

\item[(ii)]
$\omega_{\rm q}(\xi,P) = \left\|\rho\right\|_{\rm min}^2$;

\item[(iii)]
$\omega_{\rm qc}(\xi,P)=\norm{\rho}_{\max}^2$;

\item[(iv)] 
$\omega_{\rm ns}(\xi,P) = 
\sup \|\tilde{\cl E}^{(\infty,1)}(\rho)\|^2_{K^{\infty}}$, where the supremum is taken over all non-signalling
$K$-valued families $\cl E$ and all Hilbert spaces~$K$;
\item[(v)]
$\omega_{\rm lowc}(\xi,P) = \norm{\rho}_{r,{\rm cb}}^2 \leq \pi_{(2,C)}(\rho)^2$.
\end{itemize}
\end{theorem}

\begin{proof}
(i) follows from Theorem \ref{th_R-val}, Proposition \ref{p-Rvalq} and Theorem \ref{th_wcb}.

(ii) follows from Theorem \ref{th_R-val}, Proposition \ref{p-Rvalqq} and Remark \ref{r_opspsqa}.

(iii) follows from Theorems \ref{th_R-val}, \ref{th_qcrep} and Remark \ref{r_opspsqc}.

(iv)
We follow arguments in the proof of Theorem \ref{th_R-val}. 
Write 
$\xi = \sum_{r=1}^{\infty}$ $\sum_{x,y} \xi_{x,y,r} e_{x,y,r}$ for some $\xi_{x,y,r}\in \bb{C}$, $x\in X$, 
$y\in Y$, $r\in \bb{N}$. 

By the proof of Theorem \ref{th_R-val}, 
for a unit vector
$\gamma = \sum_{r=1}^{\infty} \sum_{a,b}\gamma_{a,b,r}e_{a,b,r}\in \mathbb{C}^{AB}\otimes\ell_2^R$ and 
$\rho_\gamma = \overline{\tr_R(\xi\gamma^*)}$, we have
\[
\phi_{U_{\cl E}}(\alpha(\rho_\gamma)) = 
\sum_{r=1}^{\infty} 
\sum_{(x,y)\in X Y} \sum_{(a,b)\in A B}
\overline{\xi_{x,y,r}}\gamma_{a,b,r}\eta_{x,y,a,b}=\tilde{\cl E}(\rho_\gamma),
\]
where $\alpha: \cl S_1^{AB,XY}\to \cl O_{XY,AB}$ is the complete isometry from Proposition \ref{p_OXAS1}. 
In view of Proposition \ref{p_nsre=qns}, 
the result follows by taking the supremum of~\eqref{eq:U-value} over all non-signalling 
families $\mathcal{E}$.

(v) follows from Proposition \ref{p_lowc}, Theorem \ref{th_lowc1} and Theorem \ref{th_lowcboundsq}. 
\end{proof}

\begin{remark} For rank one games ($P=\gamma\gamma^*$) with trivial referee space $\ell_2^R=\bb{C}1$, the local value $\om_{\rm loc}(\xi,\gamma\gamma^*)$ coincides with the LOSR fidelity $F_{\rm LOSR}(\xi,\gamma)$ recently studied in \cite[\S IV.C]{gc}. Given that mixed state conversion can be naturally expressed in terms of quantum games, it would be interesting to further analyse the corresponding game values through our metric characterisations, in particular, for the classes LOSR and one-way LOCC.
\end{remark}

With the notation in the statement of Theorem \ref{th_qcval}, let 
$\omega_{\rm h}(\xi,P) := \left\|\rho\right\|_{\rm h}^2$.
In the case the projection $P$ has rank one, the parameter $\omega_{\rm h}(\xi,P)$ was 
interpreted in \cite{junge} as the value of the rank one quantum game $(\xi,P)$ 
in case the players use one-way quantum communication. 
The next corollary strengthens the inequality, pointed out after \cite[Theorem 3.2 (1)]{junge}. 

\begin{corollary}\label{c_boundsqc}
If $(\xi,P)$ is a projection quantum game, then
$\omega_{\rm qc}(\xi,P)\leq \omega_{\rm h}(\xi,P)$.
\end{corollary}

\begin{proof}
  The inequality follows from Theorem~\ref{t:maxnorm},
  Corollary \ref{c:h} and Theorem \ref{th_qcval}. 
\end{proof}

\begin{corollary}\label{c_adv}
If $(\xi,P)$ is a  projection quantum game and $P$ has rank $n$,  then
$\frac{\omega_{\rm qc}(\xi,P)}{\omega_{\rm q}(\xi,P)}\leq 4n$.
\end{corollary}

\begin{proof}
Write $P=\sum_{k=1}^n\gamma_k\gamma_k^*$, where $\gamma_1,\dots,\gamma_n$ are orthonormal. Set $\rho_k=\overline{\Tr(\xi\gamma_k^*)}$ and  $\rho=[\rho_k]_{k=1}^n\in M_{n,1}(\cl S_1^{A,X}\otimes\cl S_1^{B,Y})$. Let $\omega_{\mu}(\xi,P) := \left\|\rho\right\|_{\mu}^2$, where $\mu$ is the symmetrized Haagerup norm. 
Note that the symmetrised Haagerup norm on $\cl S_1^{A,X}\otimes\cl S_1^{B,Y}$ dominates the max norm.
Indeed, for an element $M\in \cl S_1^{A,X}\otimes\cl S_1^{B,Y}$, 
we have, by Corollary \ref{c:h} and Theorem \ref{t:maxnorm}, that 
$\|M\|_{\max} \leq \|M\|_{\rm h}$; thus, by Theorem \ref{t:maxnorm}, 
$\|M\|_{\max}\leq \|M\|_{{\rm h}^{\rm t}}$ and hence, if $u,v\in  \cl S_1^{A,X}\otimes\cl S_1^{B,Y}$ and 
$M = u+v$, then 
$\|M\|_{\max} \leq \|u\|_{\max} + \|v\|_{\max} \leq \|u\|_{\rm h} + \|v\|_{{\rm h}^{\rm t}}$ implying, by (\ref{eq_hht+}), 
that $\|M\|_{\max} \leq \|M\|_{\mu}$. 

By Remark~\ref{r_opspsqc}, we have $\omega_{\rm qc}(\xi,P)\le \omega_\mu(\xi,P)$.
By the non-commutative Grothendieck inequality \cite{hm,ps} (see also \cite[Theorem 2.23]{junge})
we then have $\|\rho_k\|_\mu\leq 2\|\rho_k\|_{\rm min}$ for all $k$. 
Therefore, 
$$\|\rho\|_\mu\leq\left(\sum_{k=1}^n\|\rho_k\|_\mu^2\right)^{1/2}\leq 2\left(\sum_{k=1}^n\|\rho_k\|_{\rm min}^2\right)^{1/2}\leq2n^{1/2}\|\rho\|_{\rm min},$$ where we use $\|\rho_k\|_{\rm min}\leq\|\rho\|_{\rm min}$, and hence 
$$\frac{\omega_{\rm qc}(\xi,P)}{\omega_{\rm q}(\xi,P)}
\leq \frac{\omega_{\mu}(\xi,P)}{\omega_{\rm q}(\xi,P)} =\left(\frac{\|\rho\|_{\mu}}{\|\rho\|_{\min}}\right)^2\leq 4n.
$$
Note that we have implicitly used the completely isometric inclusions $\cl S_1^{A,X}\subseteq \cl S_1^{\max\{A,X\}}$ and $\cl S_1^{B,Y}\subseteq \cl S_1^{\max\{B,Y\}}$ together with injectivity of the minimal tensor product. 

\end{proof}

We next include metric characterisations of state convertibility via LOSR and via 
local operations and classical communication (LOCC) (we refer the reader to 
\cite{clmow} for a detailed mathematical treatment of the latter class of quantum channels).

\begin{corollary}\label{LOSR_convert}
Let $X,Y$ be finite sets, and $\xi\in \bb{C}^{XY}$ and $\gamma\in\bb{C}^{XY}$ be unit vectors. Consider the statements:
\begin{itemize}
\item[(i)] there exists a quantum channel $\Gamma : M_{XY}\to M_{XY}$ in 
the class {\rm LOSR} such that $\Gamma(\xi\xi^*) = \gamma\gamma^*$;
\item[(ii)] $\norm{\xi\gamma^*}
_{w, {\rm cb}}
= 1$;\smallskip
\item[(i')] there exists a quantum channel $\Gamma : M_{XY}\to M_{XY}$ in 
the class {\rm LOCC} such that $\Gamma(\xi\xi^*) = \gamma\gamma^*$;
\item[(ii')] $\norm{\xi\gamma^*}
_{r,{\rm cb}}
= 1$.
\end{itemize}
Then (i)$\Leftrightarrow$(ii) and (i')$\Leftrightarrow$(ii').
\end{corollary}

\begin{proof} 
(i)$\Leftrightarrow$(ii) Since $\mathcal{Q}_{\rm loc}$ is compact in $\mathsf{QC}(M_{XY},M_{XY})$, this follows
from the proof of Corollary \ref{c_convert},
Proposition \ref{p-Rvalq} and Theorem \ref{th_wcb}.

(i')$\Rightarrow$(ii') By \cite{lo-popescu} (see also \cite[Corollary 4.12]{cklt-CMP}), 
if $\Gamma$ in LOCC has the property $\Gamma(\xi\xi^*) = \gamma\gamma^*$ then $\Gamma$ can be chosen 
to belong to $\mathsf{QC}(\cl R_{\rm lowc})$. The statement now follows from 
Proposition \ref{p_lowc} and Theorem \ref{th_lowc1}.

(ii')$\Rightarrow$(i') By Proposition \ref{p_lowc} and Theorem \ref{th_lowc1}, for every $\ep>0$, there exists a one-way LOCC map $\Gamma_\ep:M_{XY}\to M_{XY}$ such that $\Tr(\Gamma_\ep(\xi\xi^*)\gamma\gamma^*)>1-\frac{\ep^2}{2|X||Y|}$. Then
$$\norm{\Gamma_\ep(\xi\xi^*)-\gamma\gamma^*}_2^2=\norm{\Gamma_\ep(\xi\xi^*)}_2^2-2\Tr(\Gamma_\ep(\xi\xi^*)\gamma\gamma^*)+1<\frac{\ep^2}{|X||Y|},$$
so that
$$\norm{\Gamma_\ep(\xi\xi^*)-\gamma\gamma^*}_1<\sqrt{|X||Y|}\norm{\Gamma_\ep(\xi\xi^*)-\gamma\gamma^*}_2<\ep.$$
By \cite[Theorem 1]{owari} (or \cite[Theorem 5.3]{cklt-CMP}), the reduced density matrices $\rho_\xi=\Tr_Y\xi\xi^*$ and $\rho_\gamma=\Tr_Y\gamma\gamma^*$ satisfy $\rho_\xi\prec\rho_\gamma$, where $\prec$ denotes majorisation of density matrices (see, e.g., \cite[\S 4.3.2]{watrous}). But then, by \cite{nielsen}, there exists an LOCC map $\Gamma:M_{XY}\to M_{XY}$ such that $\Gamma(\xi\xi^*)=\gamma\gamma^*$. 
\end{proof}

Similarly to Theorem \ref{th_qcval}, but using Theorem \ref{th_hypvge} instead of 
Theorem \ref{th_R-val}, we obtain the following characterisations of the values of 
hypergraph quantum games. 

For brevity, we will use the notation $\|\rho\|_{\rm t}$ for the norm in 
$L^2(\mathbb P_{XY}\times AB, \mu\times |\cdot |)\otimes_{\rm h}(\cl S_1^{X,A}\otimes_{\rm t}\cl S_1^{Y,B})$, 
where ${\rm t}$ can be any of \lq\lq$\max$'', \lq\lq$\min$'', \lq\lq$w,{\rm cb}$'' or ``$r,{\rm cb}$''.

\begin{theorem}\label{th_hypqhval}
Let $\mathbb H=(\mathbb P_{XY},\varphi,\mu)$ be a probabilistic quantum hypergraph over $(XY,AB)$. Write $\rho = \bar\xi(\bar\eta\circ\varphi)^*\in L^2(\mathbb P_{XY}\times AB, \tilde{\mu})\otimes_{\rm h}(\cl S_1^{X,A}\otimes\cl S_1^{Y,B})$ as defined in (\ref{ximueta}). 
Then

\begin{itemize}
\item[(i)]
$\om_{\rm loc}(\mathbb H) = \norm{\rho}_{w,{\rm cb}}^2$;

\item[(ii)]
$\omega_{\rm q}(\mathbb H) = \left\|\rho\right\|_{\mathstrut\rm min}^2$;

\item[(iii)]
$\omega_{\rm qc}(\mathbb H) = \norm{\rho}_{\mathstrut\max}^2$;

\item[(iv)] 
$\omega_{\rm ns}(\mathbb H) = 
\sup \|({\rm id}\otimes\tilde{\cl E})(\rho)\|^2_{L^2(\mathbb P_{XY}\times AB, \tilde{\mu})\otimes K}$,
where the supremum is taken over all non-signalling
$K$-valued families $\cl E$ and all Hilbert spaces~$K$;

\item[(v)]
$\omega_{\rm lowc}(\mathbb H) = \norm{\rho}_{r,{\rm cb}}^2 \leq  \pi_{(2,C)}(\rho)^2$.
\end{itemize}
\end{theorem}


\section{Alternative game value expressions}\label{s_ctoq}

The main focus of this section are classical-to-quantum games. 
We obtain expressions for their values of various types;
we start with an alternative expression for the quantum value of a 
hypergraph quantum game, of which they will be a consequence. 


\subsection{The quantum and the local value revisited}\label{ss_alter-qc}

Let $X$, $Y$, $A$ and $B$ be non-empty finite sets.
Fix a probabilistic quantum hypergraph
$\mathbb H=(\mathbb P_{XY},\varphi, \mu)$ over $(X Y, A B)$; 
here, $\varphi : \mathbb P_{XY}\to\cl P_{AB}$ is a Borel measurable function, and $\mu$ is a regular Borel probability measure on $\mathbb P_{XY}$.  
Let $\widehat{\mathbb H} = \int_{\mathbb P_{XY}} \overline{ p\otimes\varphi(p)}\,d\mu(p)$;
after shuffling, we view $\widehat{\mathbb H}$ as an element of 
$(\cl S^X_1\otimes M_A)\otimes (\cl S^Y_1\otimes M_B)$. 
Using Lemma \ref{l:selection}, write
$$p = \sum_{x\in X}\sum_{y\in Y}\xi_{x,y}(p)\xi_{x,y}(p)^* \text{ and }
\varphi(p) = \sum_{a\in A}\sum_{b\in B}\eta_{a,b}(\varphi(p))\eta_{a,b}(\varphi(p))^*,$$ 
where $\xi_{x,y} : \mathbb P_{XY}\to\mathbb C^{XY}$ and
$\eta_{a,b}:\cl P_{AB}\to\mathbb C^{AB}$ are  Borel functions.

Set $\xi_{xy}^{x'y'}(p) = \langle\xi_{x,y}(p),e_{x'}\otimes e_{y'}\rangle$ and 
$\eta_{ab}^{a'b'}(p) = \langle\eta_{a,b}(p),e_{a'}\otimes e_{b'}\rangle$; 
it is clear that the functions $\xi_{xy}^{x'y'} : \bb{P}_{XY}\to \bb{C}$ and 
$\eta_{ab}^{a'b'} : \cl P_{AB}\to \bb{C}$ are Borel.
In the sequel, if the limits of the summations are not specified, they are 
along all indices which appear.
We have that 
\begin{align}\label{eq_integsp}\tightmath
&\widehat{\mathbb H}=
\int_{\mathbb P_{XY}} \overline{p\otimes\varphi(p)}\,d\mu(p)\\
&\, = 
\sum
\int_{\mathbb P_{XY}\mspace{-18mu}}
\overline{\xi_{xy}^{x'y'}\!(p)}\,
\xi_{xy}^{x''y''}\!(p)\,
\overline{\eta_{ab}^{a'b'}\!(p)}\,
\eta_{ab}^{a''b''}\!(p)\,d\mu(p)\,
\epsilon_{x',x''\!}\otimes\epsilon_{y',y''\!}
\otimes\epsilon_{a',a''\!}\otimes\epsilon_{b',b''}.
\nonumber
\end{align}

Let $\cl S^X_1(M_A)$ (resp. $\cl S^Y_1(M_B)$) be the dual operator space of
$M_X\otimes_{\rm min} \cl S^A_1$ (resp. $M_Y\otimes_{\rm min} \cl S^B_1$).
For an element $w\in (\cl S^X_1\otimes M_A)\otimes (\cl S^Y_1\otimes M_B)$, 
we write $\|w\|_{\min}$ and $\|w\|_{\varepsilon}$ for its norms in  
$\cl S^X_1(M_A)\otimes_{\rm min} \cl S^Y_1(M_B)$ 
and  $\cl S^X_1(M_A)\otimes_{\varepsilon} \cl S^Y_1(M_B)$, respectively (see (\ref{eq_inno3})). 

\begin{theorem}\label{th_alterq} 
Let $\mathbb H = (\mathbb P_{XY}, \varphi,\mu)$ be a 
probabilistic quantum hypergraph over $(X Y,A B)$. Then
$\omega_{\rm q}(\nph,\mu) = \|\widehat{\mathbb H}\|_{\rm min}$ and 
$\omega_{\rm loc}(\nph,\mu) = \|\widehat{\mathbb H}\|_{\varepsilon}.$
\end{theorem}

\begin{proof}

Up to canonical complete isometries, we have 
\begin{equation}\label{eq_CBide}\tightmath
{\rm CB}(M_A,M_X(\cl B(H)))
\simeq \cl S^A_1\otimes_{\rm min}
M_X\otimes_{\rm min}\cl B(H)
\simeq 
{\rm CB}(\cl S^X_1(M_A),\cl B(H))
\end{equation}
(see e.g. \cite[(1.32)]{blm}). 
Thus, given a block operator isometry $U = (U_{a,x})_{a,x}$, whose entries have domain $H$, 
the unital 
completely positive map, defined by letting 
$\Phi_U(\epsilon_{a,a'}) = \sum_{x,x'\in X} \epsilon_{x,x'} \otimes U_{a,x}^*U_{a',x'}$, 
gives rise to a map $T_U : \cl S^X_1(M_A) \to \cl B(H)$, 
defined by 
\begin{equation}\label{eq_TUx'x}
T_U(\epsilon_{x',x}\otimes \epsilon_{a,a'}) = U_{a,x}^*U_{a',x'}, \ \ \ x,x'\in X, a,a'\in A.
\end{equation}
Define $S_V:\cl S^Y_1(M_B) \to \cl B(K)$  similarly for a block isometry $V=(V_{b,y})_{b,y}$.

Considering the resource $\cl R_{\rm q}$ (see Subsection \ref{ss_eprq}), 
the proof of Theorem \ref{th_hypvge} implies that, if $U = (U_{a,x})_{a,x}$ and $V = (V_{b,y})_{b,y}$
vary over all finite-rank block operator isometries in the supremum below, taking into account (\ref{eq_integsp}) 
and (\ref{eq_TUx'x}), we have
\begin{eqnarray}\label{eq_sxx'aa'lices}
\omega_{\rm q}(\mathbb H)
& = & 
 \sup_{(U,V)}\|\textstyle\sum
 (\int_{\mathbb P_{XY}}\xi_{xy}^{x'y'}\!(p)\,\overline{\xi_{xy}^{x''y''}\!(p)\,\eta_{ab}^{a'b'}\!(p)}\,\eta_{ab}^{a''b''}\!(p)\,d\mu(p))
 \nonumber\\[-1.5ex] 
& & \hspace{3.5cm}  
\times \ U_{a',x'}^*U_{a'',x''}\otimes V_{b',y'}^*V_{b'',y''}\|\nonumber\\
& = & 
\sup_{(U,V)}\|\textstyle\sum
 (\int_{\mathbb P_{XY}}\xi_{xy}^{x'y'}\!(p)\,\overline{\xi_{xy}^{x''y''}\!(p)\,\eta_{ab}^{a'b'}\!(p)}\,\eta_{ab}^{a''b''\!}(p)d\mu(p))\\[-1.5ex]
& & \hspace{3.5cm}
\times \ T_U(\epsilon_{x'',x'}\otimes \epsilon_{a',a''\!})\otimes S_V(\epsilon_{y'',y'}\otimes \epsilon_{b',b''\!})\|\nonumber\\
& = & 
 \sup_{(U,V)}\|(T_U\otimes S_V)(\textstyle\int_{\mathbb P_{XY}}
\overline{p\otimes \varphi(p)}\,d\mu(p))\|\nonumber\\
& \leq & 
\|\widehat{\mathbb H}\|_{\cl S^X_1(M_A)\otimes_{\rm min} \cl S^Y_1(M_B)}. \nonumber
\end{eqnarray}

To prove the converse inequality, fix $\varepsilon > 0$ and let 
$T\in {\rm CB}(\cl S^X_1(M_A),M_n)$ and $S\in {\rm CB}(\cl S^Y_1(M_B),M_n)$ with
$\|T\|_{\rm cb}\leq 1$, $\|S\|_{\rm cb}\leq 1$, and $\zeta$ and $\eta$ be unit vectors in $\bb C^{n^2}$,
such that 
$$\left|\left\langle (T\otimes S)\left(\int_{\mathbb P_{XY}}\overline{p\otimes\varphi(p)}\,d\mu(p)\right)\zeta, \eta
\right\rangle\right|
\geq \|\widehat{\mathbb H}\|_{\cl S^X_1(M_A)\otimes_{\rm min} \cl S^Y_1(M_B)}-\varepsilon.$$
Let $\phi : M_A\to M_X\otimes M_n$ (resp. $\psi : M_B\to M_Y\otimes M_n$) be the
complete contraction
corresponding to $T$ (resp. $S$) via (\ref{eq_CBide}). 
By \cite[Theorem 4.2]{palazuelos-vidick} there are unital completely positive maps 
$\tilde\phi_i : M_A\to M_X\otimes M_n$ and $\tilde\psi_i : M_B\to M_Y\otimes M_n$, $i=1,2$,  
such that the maps $\Phi: M_A\to M_2(M_n\otimes M_X)=M_X\otimes M_2(M_n)$ and 
$\Psi: M_B\to M_2(M_n\otimes M_Y) = M_Y\otimes M_2(M_n)$, given by 
$$\Phi(\rho) = \left(\begin{array}{cc}\tilde\phi_1(\rho) & \phi(\rho)\\ \phi(\rho^*)^* & \tilde \phi_2(\rho) \end{array}\right) \text{ and }\Psi(\sigma) = \left(\begin{array}{cc}\tilde \psi_1(\sigma) & \psi(\sigma)\\ \psi(\sigma^*)^* & \tilde\psi_2(\sigma)
\end{array} \right), 
$$
are completely positive.

For $i = 1,2$, there exist \cite[Theorem 3.1]{tt-QNS} block operator 
isometries $U_i = (U_{a,x}^i)$ and $V_i = (V_{b,y}^i)$, such that
$\tilde\phi_i(\epsilon_{a,a'}) = [(U_{a,x}^i)^*U_{a',x'}^i]_{x,x'}$ and  
$\tilde\psi_i(\epsilon_{b,b'}) = [(V_{b,y}^i)^*V_{b',y'}^i]_{y,y'}$.
Let $\tilde\zeta = (\zeta,0,0,0)\in\mathbb C^{4n^2}$,
$\tilde\eta = (0,0,0,\eta)\in \mathbb C^{4n^2}$ and 
$$G_{\Phi,\Psi} = \displaystyle\int_{\mathbb P_{XY}}\sum_{x\in X} \sum_{y\in Y} 
L_{\overline{\xi_{x,y}(p)}\overline{\xi_{x,y}(p)^*}}((\Phi\otimes\Psi)(\overline{\varphi(p)}))\,d\mu(p),$$
where $L_{\omega}$ is the slice map corresponding to the element $\omega$.
Since the map $\Phi\otimes\Psi$ is positive, the operator $G_{\Phi,\Psi}$ is positive.

Recalling the duality 
(\ref{eq_duSTnot}), we have
$$U_{a',x'}^*U_{a'',x''} = L_{\epsilon_{x'',x'}}(\Phi_U({\epsilon_{a',a''}})) \ \mbox{ and } \ 
V_{b',y'}^*V_{b'',y''} = L_{\epsilon_{y'',y'}}(\Phi_V({\epsilon_{b',b''}})).$$
Thus, in view of (\ref{eq_sxx'aa'lices}), we have 
\begin{align*}
  &\hspace{-2ex}\left|\left\langle (T\otimes S) \left(\textstyle\int_{\mathbb P_{XY}}\overline{p\otimes\varphi(p)}\,d\mu(p)\right)
                                 \zeta,\eta
\right\rangle\right|\\
& = 
\left|\left\langle \sum
\int_{\mathbb P_{XY}}\xi_{xy}^{x'y'}\!(p)\,\overline{\xi_{xy}^{x''y''}\!(p)}\,
L_{\epsilon_{x'',x'}\otimes \epsilon_{y'',y'}}((\phi\otimes\psi)(\overline{\varphi(p)}))\,d\mu(p) \zeta,\eta\right\rangle\right|\\
& = 
\left|\left\langle\sum
\int_{\mathbb P_{XY}}\xi_{xy}^{x'y'}\!(p)\,\overline{\xi_{xy}^{x''y''}\!(p)}\,
L_{\epsilon_{x'',x'}\otimes \epsilon_{y'',y'}}
((\Phi\otimes\Psi)(\overline{\varphi(p)}))\,d\mu(p) \tilde\zeta,\tilde\eta\right\rangle\right|\\
& = 
\left|\int_{\mathbb P_{XY}}\sum_{x\in X} \sum_{y\in Y} 
\left\langle L_{\overline{\xi_{x,y}(p)},\overline{\xi_{x,y}(p)}}(\Phi\otimes\Psi)(\overline{\varphi(p)})\,d\mu(p)\tilde\zeta,\tilde\eta\right\rangle\right|\\
& \leq  
\langle G_{\Phi,\Psi}\tilde\zeta,\tilde\zeta\rangle^{1/2} 
\langle G_{\Phi,\Psi}\tilde\eta,\tilde\eta\rangle^{1/2}
=   
\langle G_{\tilde\phi_1,\tilde\psi_1} \zeta,\zeta\rangle^{1/2}
\langle G_{\tilde\phi_2,\tilde\psi_2}\eta,\eta
\rangle^{1/2}\\
& \leq  
\sup\{\langle G_{\phi',\psi'} \zeta,\zeta\rangle : \ \phi', \psi' \mbox{ unital completely positive}\}
=
\omega_{\rm q}(\mathbb H),
\end{align*}
giving the converse inequality.

To see the expression for the local value, recall from Proposition~\ref{p-Rvalq} that the resource ${\cl R}_{\rm loc}$
over $(X Y, A B)$ corresponding to ${\cl Q}_{\rm loc}$
is generated by $U\otimes V$, where $U = (U_{a,x})_{a,x}$ , $V = (V_{b,y})_{b,y}$ are isometries with column-vector entries so that the products $U_{a,x}^*U_{a',x'}$ and $V_{b,y}^*V_{b',y'}$ are scalars. 

Following the arguments from the previous part of the proof, we obtain that
\begin{eqnarray*}
&\omega_{\rm loc}(\varphi,\pi)\leq \sup_{(f,g)}\|(f\otimes g)(\widehat{\mathbb H})\|=\|\widehat{\mathbb H}\|_{\cl S^X_1(M_A)\otimes_{\varepsilon}\cl S^Y_1(M_B)}
\end{eqnarray*}
where the supremum 
is taken over all linear functionals $f$ and $g$ on $\cl S^X_1(M_A)$ and $\cl S^Y_1(M_B)$, respectively.
The converse statement is obtained by letting $n=1$ in the above arguments and taking into account that any contractive linear functional is automatically completely contractive. 
\end{proof}

Let $R$ be a finite set. Write $\phi_R: M_R\otimes M_R\to\mathbb C$ for the map $a\otimes b\mapsto\tr(ab)$. 

\begin{theorem}\label{alter}
Let $X$, $Y$, $A$ and $B$ be finite sets, and $R$ a finite set.   
Let $\xi\in \bb{C}^{XY}\otimes\ell_2^R$ be a unit vector, and 
$P = \sum_{k = 1}^{\infty} \lambda_k \gamma_k\gamma_k^*$, where 
$(\gamma_k)_{k=1}^n\subseteq \bb{C}^{AB}\otimes\ell_2^R$ is an orthonormal family and $0\leq \lambda_k\leq 1$, 
$k\in \bb{N}$. Consider $\xi\xi^*\otimes P$ as an element in $\cl S^X_1(M_A)\otimes\cl S^Y_1(M_B)\otimes M_R\otimes M_R$ and set $\mathbb G=\overline{({\rm id}\otimes\phi_R)(\xi\xi^*\otimes P)}$. Then
\begin{itemize}
\item[(i)]
$\om_{\rm loc}(\xi,P) = \norm{\mathbb G}_{\cl S^X_1(M_A)\otimes_\epsilon\cl S^Y_1(M_B)}$;

\item[(ii)]
$\omega_{\rm q}(\xi,P) = \norm{\mathbb G}_{\cl S^X_1(M_A)\otimes_{\rm min}\cl S^Y_1(M_B)}$.
\end{itemize}
\end{theorem}

\begin{proof}
We follow similar arguments as in the proof of Theorem \ref{th_alterq}.  Let $U=(U_{a,x})_{a,x}$, $V=(V_{b,y})_{b,y}$ be block operator isometries such that $U\otimes V\in\cl R_{\rm q}$, $U_{a,x}\in \cl B(H)$, $V_{b,y}\in \cl B(K)$,  and $\sigma$ be a normal state on $\cl B(H\otimes K)$.

Observe that if $\xi=\sum_{x,y}e_x\otimes e_y\otimes\xi_{x,y}$, $\gamma_n=\sum_{a,b}e_a\otimes e_b\otimes\gamma_{n,a,b}$, then
$$\xi\xi^*=\sum\epsilon_{x,x'}\otimes\epsilon_{y,y'}\otimes\xi_{x,y}\xi_{x',y'}^*=\sum\epsilon_{x,x'}\otimes\epsilon_{y,y'}\otimes\langle\xi_{x,y},e_r\rangle\langle e_{r'},\xi_{x',y'}\rangle\epsilon_{r,r'},$$
and 
$$\gamma_n\gamma_n^*=\sum\epsilon_{a,a'}\otimes\epsilon_{b,b'}\otimes\langle\gamma_{n,a,b},e_r\rangle\langle e_{r'}\gamma_{n,a',b'}\rangle\epsilon_{r,r'},$$
where the summations are over all indices appearing.
With $T_U$ and $S_V$ defined as in (\ref{eq_TUx'x}), the value can be expressed as follows:
\begin{align*}
  \omega_{\rm q}&(\xi,P)=\sup_{(U,V, \sigma)}\Tr((\Gamma_{U\otimes V,\sigma}\otimes\text{id}_R)(\xi\xi^*)P)\\&=\sup_{(U,V,\sigma)}\sigma\left(\sum \lambda_n(U_{a,x}^*U_{a',x'}\otimes V_{b,y}^*V_{b',y'}\langle\gamma_{n,a',b'},\xi_{x',y'}\rangle \langle\xi_{x,y},\gamma_{n,a,b}\rangle\right)\\
&=\sup_{(U,V,\sigma)}\sigma\Big(\sum \lambda_n((T_U\otimes S_V)(\epsilon_{x',x}\otimes\epsilon_{a,a'}\otimes\epsilon_{y',y}\otimes\epsilon_{b,b'})\\[-2ex]&\hspace{6cm}\times\langle\gamma_{n,a',b'},\xi_{x',y'}\rangle \langle\xi_{x,y},\gamma_{n,a,b}\rangle\Big)\\
&=\sup_{(U,V,\sigma)}\sigma\Big(\sum \lambda_n((T_U\otimes S_V)(\epsilon_{x',x}\otimes\epsilon_{a,a'}\otimes\epsilon_{y',y}\otimes\epsilon_{b,b'})\\[-2ex]&\hspace{4cm}\times\langle\gamma_{n,a',b'},e_{r'}\rangle\langle e_{r'},\xi_{x',y'}\rangle \langle\xi_{x,y},e_r\rangle\langle e_r,\gamma_{n,a,b}\rangle\Big)\\
&=\sup_{(U,V,\sigma)}\sigma\Big(\sum \lambda_n((T_U\otimes S_V)(\epsilon_{x,x'}\otimes\epsilon_{a,a'}\otimes\epsilon_{y,y'}\otimes\epsilon_{b,b'})\\[-2ex]&\hspace{2cm}\times\langle\gamma_{n,a',b'},e_{s'}\rangle\langle e_{r},\xi_{x,y}\rangle \langle\xi_{x',y'},e_{r'}\rangle\langle e_{s},\gamma_{n,a,b}\rangle\Tr(\epsilon_{r,r'}\epsilon_{s,s'})\Big)\\
&=\sup_{(U,V,\sigma)}\sigma\left((T_U\otimes S_V)\overline{(\text{id}\otimes\phi_R)(\xi\xi^*\otimes P)}\right).
\end{align*}
Hence 
$$\omega_{\rm q}(\xi,P)\leq  \norm{\mathbb G}_{\cl S^X_1(M_A)\otimes_{\rm min}\cl S^Y_1(M_B)}.$$
The reverse inequality is proved as in Theorem \ref{th_alterq} using the above expression for the value, where the averaging $\int_{\mathbb P_{XY}}p\otimes\varphi(p)\,d\mu(p)$ is replaced by $(\text{id}\otimes\phi_R)(\xi\xi^*\otimes P)$. In fact, retaining the notation from the proof of Theorem \ref{th_alterq} and denoting by $T_U\in {\rm CB}(\cl S_1^X(M_A), M_2(M_n))$  (resp. $S_V\in {\rm CB}(\cl S_1^Y, M_2(M_n))$) the maps corresponding to the unital completely positive maps $\Phi$  and related block operator isometry $U$ (resp. $\Psi$  and block operator isometry $V$)  we have,  in particular, that $$\tightmath\left\langle(T\otimes S)\left(\overline{(\text{id}\otimes\phi_R)(\xi\xi^*\otimes P)}\right)\zeta,\eta\right\rangle=\left\langle G_{\Phi,\Psi}\tilde\zeta,\tilde \eta\right\rangle$$ 
where $G_{\Phi,\Psi}=(T_U\otimes S_V)\left(\overline{(\text{id}\otimes\phi_R)(\xi\xi^*\otimes P)}\right)$.
Noting that for $\sigma=\omega_{\delta,\delta}$, the vector state corresponding to a unit vector $\delta\in \mathbb C^{4n^2}$, we have $$\left\langle G_{\Phi,\Psi}\delta,\delta\right\rangle=
\Tr((\Gamma_{U\otimes V,\omega_{\delta,\delta}}\otimes{\rm id}_R)(\xi\xi^*)P)\geq 0,$$ 
we obtain
\begin{multline*}
\left\langle(T\otimes S)\left(\overline{(\text{id}\otimes\phi_R)(\xi\xi^*\otimes P)}\right)\zeta,\eta\right\rangle\\\leq  \langle G_{\Phi,\Psi}\tilde\zeta,\tilde \zeta\rangle^{1/2} \langle \strut G_{\Phi,\Psi}\tilde\eta,\tilde \eta\rangle^{1/2}=\langle G_{\tilde\phi_1,\tilde\psi_1}\zeta,\zeta\rangle ^{1/2} \langle G_{\tilde\phi_2,\tilde\psi_2}\eta, \eta\rangle^{1/2}
\end{multline*}
where we use the Cauchy-Schwarz inequality for the last inequality. This allows us
to complete the proof as in Theorem \ref{th_alterq}. 

(i) is proved in a similar way. 
\end{proof}


\subsection{Classical-to-quantum game values}\label{ss_class-to-quantum}

We first recall the classical-to-quantum no-signalling correlation types introduced in \cite{tt-QNS}. 
A \emph{classical-to-quantum no-signalling (CQNS)} correlation is a channel $\cl E : \cl D_{XY}\to M_{AB}$ such that (\ref{eq_qns1}) and (\ref{eq_qns2}) hold true for $\rho_X\in \cl D_X$ and $\rho_Y\in \cl D_Y$. 
A \emph{semi-classical} stochastic operator matrix acting on a Hilbert space $H$ is a positive 
block operator matrix $E = (E_{x,a,a'})_{x,a,a'}\in \cl D_X\otimes M_A(\cl B(H))$ with $\Tr_A E = I_X\otimes I_H$. 
A CQNS correlation $\Gamma$ is \emph{quantum commuting} if its
Choi matrix is given as in (\ref{eq_EFp}) but employing semi-classical stochastic operator matrices. This is equivalent to the requirement that its
canonical extension to a QNS correlation (denoted in the same way) 
$\Gamma : M_{XY}\to M_{AB}$ is quantum commuting; 
we denote by $\CQ_{\rm qc}$ the class of all quantum commuting 
CQNS correlations. 
The class $\CQ_{\rm q}$ of classical-to-quantum no-signalling correlations of 
\emph{quantum} type is defined by making the analogous replacements 
with semi-classical stochastic operator matrices in the definition of 
QNS correlations of quantum type, while the class 
$\CQ_{\rm loc}$ of local CQNS correlations consists of 
convex combinations of the CQNS correlations of the form 
$\cl E\otimes \cl F$, where $\cl E : \cl D_X\to M_A$ and $\cl F\to M_B$ are 
trace preserving (completely) positive maps. 
We note the (strict \cite{tt-QNS}) inclusions
$$\CQ_{\rm loc}\subseteq \CQ_{\rm q} \subseteq \CQ_{\rm qc}\subseteq \CQ_{\rm ns}.$$
Let $\cl B_{X,A}$
be the C*-algebraic free product $M_{A}\ast_1\cdots \ast_1 M_{A}$ (with $|X|$ copies of $M_A$), amalgamated over the units.
We denote by $e_{x,a,a'}$, $a,a'\in A$, the matrix units of the $x$-th copy of $M_A$ in $\cl B_{X,A}$, 
and let 
$$\cl R_{X,A} = {\rm span}\{e_{x,a,a'} : x\in X, a,a'\in A\},$$
viewed as an operator subsystem of $\cl B_{X,A}$.

Further, writing $\cl S_1^A:=\cl S_1^{A,A}$, let 
$\ell^X_\infty(\cl S_1^A) = \oplus_{x\in X} \cl S_1^A$,
equipped with the operator space 
structure arising from the identification
$\ell^X_\infty(\cl S_1^A) \cong \ell^X_\infty\otimes_{\rm min} \cl S^A_1$. 
Note that, if $R_{x,a,a'}\in M_n$, $x\in X$, $a,a' \in A$, then
$$
\left\|\sum_{x\in X}\sum_{a,a'\in A}\!\! R_{x,a,a'}\otimes e_x\otimes \epsilon_{a,a'}\right\|_{M_n(\ell^\infty_X(\cl S^A_1))}
\!\!\! = 
\sup_{x\in X}\left\|\sum_{a,a'\in A}\!\!R_{x,a,a'}\otimes\epsilon_{a,a'}\right\|_{M_n(\cl S^A_1)}\!\!.$$ 
We write
$\ell^X_1(M_A) = \ell^X_\infty(\cl S_1^A)^*$ for the corresponding dual operator space.
Note that 
\begin{equation}\label{eq_semicliden}
{\rm CB}(\ell^X_1(M_A), \cl B(H))\simeq \ell^X_\infty(\cl S_1^A)\otimes_{\rm min}\cl B(H)
\simeq \ell^X_\infty \otimes_{\min} \cl S_1^A\otimes_{\rm min}\cl B(H),
\end{equation}
where an element $\hat{T} = \sum_{x\in X} \sum_{a,a'\in A} e_x\otimes \epsilon_{a,a'}\otimes T_x^{a,a'}$ 
of the last of the three spaces in (\ref{eq_semicliden}) corresponds to the 
map $T : \ell^X_1(M_A) \to \cl B(H)$, given by $T(e_x\otimes \epsilon_{a,a'}) = T_x^{a',a}$. 
Letting further
$T_x : M_A\to \cl B(H)$ be the linear map, given by $T_x(\epsilon_{a,a'})=T_x^{a',a}$, $a,a'\in A$,
we have
$$\|T\|_{\rm cb}
= \left\|\sum_{x\in X}\sum_{a,a'\in A}e_x\otimes \epsilon_{a,a'}\otimes T_x^{a,a'}\right\|_{\ell^X_\infty(\cl S_1^A)\otimes_{\rm min}\cl B(H)}
= \sup_{x\in X} \|T_x\|_{\rm cb}.$$

\begin{remark}\label{r_ciembl1}
\rm 
The natural embedding $\iota : \ell^X_1(M_A)\to \cl S^X_1(M_A)$ is completely isometric. 
\end{remark}

\begin{proof}
Fix $n\in \bb{N}$ and $u=(u_{i,j})_{i,j}\in M_n(\ell^X_1(M_A))$, and let 
$T_u$ be the corresponding map in ${\rm CB}(\ell^X_\infty(\cl S^A_1), M_n)$; then
$\|u\| = \|T_u\|_{\rm cb}$. Let $\tilde u = \iota^{(n)}(u)$. 
and $\Delta : M_X\to \ell_\infty^X$ be the conditional expectation. Then
$T_{\tilde u} = T_u\circ (\Delta\otimes{\rm id})\in {\rm CB}(M_X(\cl S^A_1), M_n)$.
Thus, $\|T_{\tilde u}\|_{\rm cb}\leq\|T_u\|_{\rm cb}$; on the other hand, as
$\ell^X_\infty(\cl S^A_1)\subseteq M_X(\cl S^A_1)$ completely isometrically, 
$\|T_u\|_{\rm cb}\leq\|T_{\tilde u}\|_{\rm cb}$, as $T_u$ is the restriction of $T_{\tilde u}$.
\end{proof}
 
For every $S\in M_A$, let $f_S : M_A\to \bb{C}$ be the functional, given by 
$f_S(T) = {\rm Tr}(S^{\rm t}T)$. 
We equip $\cl S^A_1$ with the family of matricial cones, dual to the 
canonical family of cones of $M_A$, with respect to the duality $S\to f_S$ (see e.g. \cite[Theorem 6.2]{ptt}).
Further, equip $\ell_\infty^X(\cl S^A_1)$ with the canonical direct sum operator system structure, and 
let $\tilde{\ell}^X_1(M_A)$ be the operator system dual of $\ell_\infty^X(\cl S^A_1)$.
We note that, by \cite[Lemma 5.7]{kptt},
$$M_n(\tilde{\ell}^X_1(M_A))^+ \equiv
(\ell_\infty^X(\cl S^A_1)\otimes_{\min} M_n)^+ \equiv
(\ell_\infty^X\otimes_{\min} (\cl S^A_1 \otimes_{\min} M_n))^+.$$ 
Thus, 
$F = [f_{i,j}]_{i,j} \in M_n(\tilde{\ell}^X_1(M_A))^+$ if and only if the associated map
$F_x : M_A\to M_n$ is completely positive for every $x\in X$.
Further, let $\tilde{\ell}^X_1(M_A)\otimes_{\rm c} \tilde{\ell}^Y_1(M_B)$ be the commuting tensor product
\cite[Section 6]{kptt}; thus, $u = [u_{i,j}]_{i,j}$ belongs to 
$M_n(\tilde{\ell}^X_1(M_A)\otimes_{\rm c} \tilde{\ell}^Y_1(M_B))^+$ 
if and only if, for any pair of unital completely positive maps 
$\phi : \tilde{\ell}^X_1(M_A)\to \cl B(H)$ and 
$\psi : \tilde{\ell}^Y_1(M_B)\to \cl B(H)$ with commuting ranges, we have 
$(\phi\cdot\psi)^{(n)}(u)\geq 0$.
By \cite[Theorem 6.4]{kptt},
$\tilde{\ell}^X_1(M_A)\otimes_{\rm c} \tilde{\ell}^Y_1(M_B)\subseteq C^*_u(\tilde{\ell}^X_1(M_A))\otimes_{\rm max}C^*_u(\tilde{\ell}^Y_1(M_B))$
completely order isomorphically, where $C_u^*(\cl S)$ is the universal cover of an operator system $\cl S$
(see e.g. \cite[p. 289]{kptt}). We require the following lemma.

\begin{lemma}\label{l_jointdi}
Let $\cl A$ and $\cl B$ be unital C*-algebras, $H$ be a Hilbert space, and 
$\phi : \cl A\to \cl B(H)$ and $\psi : \cl B\to \cl B(H)$ be unital completely positive maps 
with commuting ranges. Then there exists a Hilbert space $K$, an isometry $V : H\to K$, and 
unital *-homomorphisms $\pi : \cl A\to \cl B(K)$ and $\rho : \cl B\to \cl B(K)$
with commuting ranges, such that 
\begin{equation}\label{eq_dila}
(\phi\cdot\psi)(w) = V^*((\pi\cdot\rho)(w))V, \ \ \ w\in \cl A\otimes\cl B.
\end{equation}
\end{lemma}

\begin{proof}
By Stinespring's Dilation Theorem, there exists a Hilbert space $K_0$, an isometry 
$U : H\to K_0$ and a unital *-homomorphism 
$\pi_{0} : \cl A\to \cl B(K_0)$, such that 
$$\phi(u) = U^*\pi_{0}(u) U, \ \ \ u\in \cl A;$$
choosing a minimal Stinespring dilation, we assume that 
$\overline{{\rm span}(\pi_0(\cl A)UH)} = K_0$. 
We have that the range of $\psi$ commutes with the operators of the form 
$U^*\pi_{0}(u)U$, $u\in \cl A$. 
By Arveson's Commutant Lifting Theorem \cite[Theorem 12.7]{Pa}, there 
exists a unital *-homomorphism 
$\theta : \psi(\cl B)\to \pi_{0}(\cl A)'$ such that 
$$U^*\pi_{0}(u)U \psi(v) = U^*\pi_{0}(u)\theta(\psi(v))U, \ \ \ u\in \cl A, v\in \cl B.$$

Write $\psi_{0} = \theta\circ \psi$; then $\psi_{0} : \cl B\to \cl B(K_0)$ is a unital 
completely positive map whose range commutes with the range of $\pi_0$.
By Stinespring's Dilation Theorem, there exists a Hilbert space $K$,
an isometry $W : K_0\to K$, and a
unital *-homomorphism $\rho : \cl B\to \cl B(K)$, such that 
$$\psi_0(v) = W^*\rho(w) W, \ \ \ v\in \cl B;$$
again, assume, without loss of generality, that 
$\overline{{\rm span}(\rho(\cl B)WK_0)} = K$. 
Applying Arveson's Commutant Lifting Theorem again, 
we obtain a unital *-homomorphism 
$\pi_1 : \pi_0(\cl A)'\to \rho(\cl B)'$, such that 
$$\pi_0(u) W^*\rho(v)W = W^*\pi_1(\pi_0(u))\rho(v)W, \ \ \ u\in \cl A, v\in \cl B.$$
Letting $\pi = \pi_1\circ \pi_0$ and $V = WU$, we see that $V$ is an isometry, and that 
(\ref{eq_dila}) holds true.
\end{proof}

A \emph{classical-to-quantum hypergraph game} over $(X,Y,A,B)$ 
is a probabilistic classical-to-quantum hypergraph $(\varphi, \pi)$, that is, a hypergraph quantum game $(\varphi,\pi)$ over $(X,Y,A,B)$, where 
$\pi$ is supported on the classical state space $\bb{P}_{XY}^{\rm cl}$; thus, we will assume that
$\nph$ is defined on $\bb{P}_{XY}^{\rm cl}$ (and takes values in $\cl P_{AB}$).
We will identify the probability distribution $\pi$ on $\bb{P}_{XY}^{\rm cl}$ 
with a probability distribution (denoted in the same way) $\pi : X Y \to [0,1]$.
For ${\rm t}\in \{{\rm loc},{\rm q}, {\rm qc},{\rm ns}\}$, the ${\rm t}$-value
$\omega_{\rm t}(\varphi,\pi)$ of $(\varphi, \pi)$ is given by
$$\omega_{\rm t}(\varphi,\pi) = \sup_{\Gamma\in \CQ_{\rm t}}
\sum_{x\in X}\sum_{y\in Y} \pi(x,y) {\rm Tr}(\Gamma(\epsilon_{x,x}\otimes \epsilon_{y,y})
\nph(\epsilon_{x,x}\otimes \epsilon_{y,y})).$$

As in Remark \ref{r_link}, if
$$\xi_\pi = \sum_{x\in X}\sum_{y\in Y} \sqrt{\pi(x,y)} e_x\otimes e_y\otimes e_x\otimes e_y$$ 
and 
$P = \sum_{x\in X}\sum_{y\in Y} 
\varphi(\epsilon_{x,x}\otimes\epsilon_{y,y})\otimes\epsilon_{x,x}\otimes\epsilon_{y,y},$ 
then 
$$\omega_{\rm t}(\varphi,\pi) 
= \sup_{\Gamma\in \CQ_{\rm t}}\Tr((\Gamma\otimes{\rm id})(\xi_\pi^*\xi_\pi)P).
$$

We fix a classical-to-quantum hypergraph game $G = (\varphi, \pi)$. We let 
\begin{equation}\label{eq_bbGde}
\bb{G} 
= (\pi(x,y)\overline{\varphi(\epsilon_{x,x}\otimes\epsilon_{y,y})})_{x,y}\in 
\ell_1^X(M_A)\otimes\ell_1^Y(M_B),
\end{equation}
which will be also considered 
as an element of $\tilde{\ell}^X_1(M_A)\otimes \tilde{\ell}^Y_1(M_B)$.
Further, letting $\iota_{x,y} : M_A\otimes M_B\to \cl R_{X,A}\otimes \cl R_{Y,B}$ be 
the canonical inclusion maps, we set 
$$\tilde{\bb{G}} = \sum_{x\in X}\sum_{y\in Y} \pi(x,y)\iota_{x,y}
\left(\overline{\nph(\epsilon_{x,x}\otimes \epsilon_{y,y})}\right);$$
thus, $\tilde{\bb{G}}\in \cl R_{X,A}\otimes\cl R_{Y,B}$.

\begin{theorem}\label{qv_classtoquant}
Let $G = (\varphi, \pi)$ be a classical-to-quantum hypergraph game. 
 Then 
\begin{itemize}
\item[(i)]
$\omega_{\rm loc}(\varphi,\pi)=\|\bb{G}\|_{\ell_1^X(M_A)\otimes_{\varepsilon}\ell_1^Y(M_B)}$; 

\item[(ii)]
$\omega_{\rm q}(\varphi,\pi) = 
\|\bb{G}\|_{\ell_1^X(M_A)\otimes_{\rm min}\ell_1^Y(M_B)}
= \|\tilde{\bb G}\|_{\cl R_{X,A}\otimes_{\min} \cl R_{Y,B}}$;

\item[(iii)] 
$\omega_{\rm qc}(\varphi,\pi) =  \|\bb{G}\|_{C^*_u(\tilde{\ell}^X_1(M_A))\otimes_{\rm max}
C^*_u(\tilde{\ell}^Y_1(M_B))}
=
\|\tilde{\bb{G}}\|_{\cl R_{X,A}\otimes_{\rm c} \cl R_{Y,B}}$\\ 
$= \|\tilde{\bb{G}}\|_{\cl B_{X,A}\otimes_{\max} \cl B_{Y,B}}$;

\item[(iv)]
$\omega_{\rm ns}(\varphi,\pi)=\|\tilde{\bb{G}}\|_{\cl R_{X,A}\otimes_{\max} \cl R_{Y,B}}$.
\end{itemize}
\end{theorem}

\begin{proof}
As the minimal and Banach space injective 
tensor products are injective and $\ell^X_1(M_A)\hookrightarrow\cl S^X_1(M_A)$ completely isometrically
(see Remark \ref{r_ciembl1}), 
(i) and  the first equality in (ii) follow from  Theorem \ref{th_alterq}. The second equality in (ii) can be deduced similarly to the second equality in (iii).

(iii) 
Write $P_{x,y} = \nph(\epsilon_{x,x}\otimes\epsilon_{y,y})$, $x\in X$, $y\in Y$. 
For a semi-stochastic operator matrix $E = (E_{x,a,a'})_{x,a,a'}$ with entries in 
$\cl B(H)$, let $\phi_E : \cl R_{X,A}\to \cl B(H)$ be the unital completely positive map, 
given by $\phi_E(e_{x,a,a'}) = E_{x,a,a'}$, $x\in X$, $a,a'\in A$ (see \cite[Section 7]{tt-QNS}). 
Using Lemma \ref{l_jointdi} at the third equality, we have that
\begin{align*}
& \!\!\!\omega_{\rm qc}(\varphi,\pi)  \\
& = 
\sup_{E,F,\xi}\sum_{x\in X}\sum_{y\in Y} \sum_{a,a'\in A}\sum_{b,b'\in B}
\pi(x,y)\langle E_{x,a,a'} F_{y,b,b'}\xi,\xi\rangle 
{\Tr (P_{x,y}(\epsilon_{a,a'}\otimes\epsilon_{b,b'}))}\\
& =  
\sup_{E,F,\xi}\sum_{x\in X}\sum_{y\in Y} 
\pi(x,y)\langle (\phi_E\cdot \phi_F)(\iota_{x,y}(P_{x,y}^{\rm t}))\xi,\xi\rangle
= 
\sup_{E,F}\|(\phi_E\cdot \phi_F)(\tilde{\bb G})\|\\
& =  
\sup\{\{\|(\rho_1\cdot \rho_2)(\tilde{\bb G})\| : \rho_1, \rho_2 \mbox{ comm. *-rep. of } \cl B_{X,A} \mbox{ and } 
\cl B_{Y,B},\mbox{ resp.}\}\\
& =  
\|\tilde{\bb G}\|_{\cl B_{X,A}\otimes_{\max} \cl B_{Y,B}}.
\end{align*}
Now by \cite[Remark 7.4]{tt-QNS} (see also \cite[Lemma 2.8]{pt_QJM}), 
the norms of $\tilde{\bb G}$ in 
$\cl R_{X,A}\otimes_{\rm c} \cl R_{Y,B}$ and in
$\cl B_{X,A}\otimes_{\max} \cl B_{Y,B}$ are identical. 

To see the first equality in (iii), using the notation and the arguments from Theorem \ref{th_alterq}, 
and setting $\eta_{ab}^{a'b'}(x,y) = \eta_{ab}^{a'b'}(\epsilon_{x,x}\otimes \epsilon_{y,y})$,
we have 
 $$\omega_{\rm qc}(\varphi,\pi)
 =
 \sup_{(U,V)}\left\|\sum\pi(x,y)\,
 \overline{\eta_{ab}^{a'b'}\!(x,y)}\,\eta_{ab}^{a''b''}\!(x,y)\, U_{a',x}^*U_{a'',x} V_{b',y}^*V_{b'',y} \right\|,$$
 where the supremum is taken over all pairs $(U,V)$ of block operator isometries, 
 $U = (U_{a,x})_{a,x}$ and $V = (V_{b,y})_{b,y}$, with the property that 
 $U_{a',x}^*U_{a'',x}$ and $V_{b',y}^*V_{b'',y}$ commute for all $x, y, a', a'', b', b''$. 
 Letting as above $T_x : M_A\to \cl B(H) $, 
 $T_x(\epsilon_{a',a''}) = U_{a',x}^*U_{a'',x}$ and 
 $S_y : M_B\to\cl B(H)$, $S_y(\epsilon_{b',b''}) = V_{b',y}^*V_{b'',y}$,   
 we obtain unital completely positive maps  $T : \ell^X_1(M_A)\to\cl B(H)$
 and 
 $S:\ell^Y_1(M_B)\to\cl B(H)$ with commuting ranges, given by 
 $T(e_x\otimes t)=T_x(t)$ and $S(e_y\otimes s)=S_y(s)$, respectively.
 Let $\rho_A: C^*_u(\tilde{\ell}^X_1(M_A))\to\cl B(H)$ and $\rho_B: C^*_u(\tilde{\ell}^Y_1(M_B))\to\cl B(H)$ 
 be the $*$-homomorphisms that extend $T$ and $S$, respectively. 
 Then $\rho_A$ and $\rho_B$ have commuting ranges and therefore 
 \begin{eqnarray*}
 \omega_{\rm qc}(\varphi,\pi) 
 & = & \sup_{(T,S)}\|(T\cdot S)(\textstyle\sum\bb{G}_{x,a',a'',y,b',b''}(e_x\otimes \epsilon_{a',a''}\otimes e_y\otimes \epsilon_{b',b''}))\|\\
& = & \sup_{\rho_A,\rho_B}\|(\rho_A\cdot\rho_B)(\textstyle\sum \bb{G}_{x,a',a'',y,b',b''}
(e_x\otimes \epsilon_{a',a''}\otimes e_y\otimes \epsilon_{b',b''}))\|\\
& \leq & \|\bb{G}\|_{C^*_u(\tilde{\ell}^X_1(M_A))\otimes_{\rm max}C_u^*(\tilde{\ell}^Y_1(M_B))}.
\end{eqnarray*}
The converse inequality follows from the definition of the maximal 
operator system tensor product and its inclusion into the maximal tensor product of the universal 
covers \cite[Theorem 6.4]{kptt}. 

(iv) follows as in (iii) by replacing the terms $\langle E_{x,a,a'} F_{y,b,b'}\xi,\xi\rangle$ therein
by the values $s(e_{x,a,a'}\otimes f_{y,b,b'})$ of an arbitrary state $s$ of $\cl R_{X,A}\otimes_{\max}\cl R_{Y,B}$
at the (elementary tensors of) canonical generators. 
\end{proof}


\subsection{Classical game values}\label{ss_classical}

In this subsection, 
we recover some known tensor expressions for the 
values of $G = (X,Y, A, B,\lambda,\pi)$, a classical non-local game (with rule function $\lambda$ and 
a probability measure $\pi$ on the question set $XY$). We refer the reader to \cite[Lemma 4.1]{palazuelos-vidick} for 
a formulation of the formulas for the 
local and the quantum value of $G$; the expressions for the quantum commuting and the no-signalling value 
are folklore, but to the authors' knowledge, have not been formulated explicitly.

We recall that a \emph{classical no-signalling (NS)} correlation is a 
classical-to-quantum no-signalling correlation with values in $\cl D_{AB}$. 
Every classical NS correlation $\Gamma : \cl D_{XY}\to \cl D_{AB}$ 
can be identified with the set of values
$p(a,b|x,y) := \langle\Gamma(\epsilon_{x,x}\otimes \epsilon_{y,y}),\epsilon_{a,a}\otimes \epsilon_{b,b}\rangle$, 
$x\in X$, $y\in Y$, $a\in A$, $b\in B$; we write $\Gamma = \Gamma_p$. 
Using standard notation \cite{lmprsstw}, we write
$\cl C_{\rm t}$ for the class of classical NS correlations that, when viewed as 
classical-to-quantum correlations, lie in $\CQ_{\rm t}$. 
The ${\rm t}$-value of $G$ is defined by
$$\omega_{\rm t}(G) = \sup_{p\in \cl C_{\rm t}} 
\sum_{x\in X} \sum_{y\in Y} \sum_{a\in A} \sum_{b\in B}
\pi(x,y)p(a,b|a,y)\lambda(x,y,a,b).$$

Let $\xi_{\pi}=\sum_{(x,y)\in X Y}\sqrt{\pi(x,y)} e_x\otimes e_y\otimes e_x\otimes e_y$ and 
$P = \sum_{(x,y)\in X Y} \varphi_{\lambda}(\epsilon_{x,x}$ $\otimes\epsilon_{y,y})\otimes\epsilon_{x,x}\otimes\epsilon_{y,y}$ (see (\ref{eq_nphla})); 
thus, if $\gamma_{x,y,a,b} = \lambda(x,y,a,b)e_a\otimes e_b\otimes e_x\otimes e_y$, 
then 
$P = \sum_{(x,y)\in X Y} \sum_{(a,b)\in A B}\gamma_{x,y,a,b}\gamma_{x,y,a,b}^*$.
By the previous section and Remark \ref{r_link}, we have 
$$
\omega_{\rm t}(G) = 
\sup_{p\in\cl C_t}\Tr\big((\Gamma_p\otimes{\rm id})(\xi_\pi\xi_\pi^*)P\big).$$

Let $\ell^X_\infty(\ell_A^1) := \ell^X_\infty\otimes_{\rm min}\ell^1_A$; 
here, the space $\ell_A^1$ is equipped with its natural operator space structure by viewing is as the 
dual of $\ell_A^\infty \equiv \cl D_A$. 

Let $\ell^X_1(\ell_A^\infty) := (\ell^X_\infty(\ell_A^1))^*$ be the corresponding dual operator space, which 
now coincides with the dual operator system of $\ell_\infty^X\otimes_{\min} \ell^1_A$.
In the next corollary, we view $\bb{G}$, defined in (\ref{eq_bbGde}), as an element of 
$\ell^X_1(\ell_A^\infty)\otimes\ell^Y_1(\ell_B^\infty)$.

\begin{corollary}
Let $G=(X,Y, A,B,\pi,\lambda)$ be a non-local game 
and $\bb{G} = (\pi(x,y)\lambda(x,y,a,b))_{x,a,y,b}\in \ell_1^X(\ell_\infty^A)\otimes\ell_1^Y(\ell_\infty^B)$. 
Then

\begin{itemize}
\item[(i)]
$\omega_{\rm loc}(G) = \|\bb{G}\|_{\ell_1^X(\ell_\infty^A)\otimes_{\varepsilon}\ell_1^Y(\ell_\infty^B)}$;

\item[(ii)]
$\omega_{\rm q}(G) = \|\bb{G}\|_{\ell_1^X(\ell_\infty^A)\otimes_{\rm min}\ell_1^Y(\ell_\infty^B)}$;

\item[(iii)] 
$\omega_{\rm qc}(G) =  
\|\bb{G}\|_{\ell_1^X(\ell_\infty^A)\otimes_{\rm c}\ell_1^Y(\ell_\infty^B)}$;

\item[(iv)]
$\omega_{\rm ns}(G)=\|\tilde{\bb{G}}\|_{\ell_1^X(\ell_\infty^A)\otimes_{\rm max}\ell_1^Y(\ell_\infty^B)}$.
\end{itemize}
\end{corollary}

\begin{proof}
The claims will follow from Theorem \ref{qv_classtoquant}, once we show that 
$\iota : \ell^X_1(\ell_A^\infty)\hookrightarrow \ell^X_1(M_A)$ is 
complete isometry, and that the embedding
$\iota' : \ell^X_1(\ell_A^\infty)\hookrightarrow \tilde{\ell}^X_1(M_A)$ 
is a unital complete order isomorphism onto its range. 

Let $u\in M_n(\ell^X_1(\ell_A^\infty))$. Then
$\|u\|_{M_n(\ell^X_1(\ell_A^\infty))} = \sup_{x\in X} \|u_x\|_{\rm cb}$, where
$u_x : \ell_A^\infty\to M_n$  are the maps canonically associated to $u$.
On the other hand, 
$\|\iota^{(n)}(u)\|_{M_n(\ell^X_1(M_A))} = \sup_{x\in X} \|\tilde u_x\|_{\rm cb}$, where 
$\tilde u_x: M_A\to M_n$ is given by 
$\tilde u_x= u_x\circ \Delta$, $x\in X$. 
Noting that $\|\tilde u_x\|_{\rm cb} = \|u_x\|_{\rm cb}$, $x\in X$, 
we obtain that $\iota$ is a complete isometry.
The statement about $\iota'$ follows a similar route and the proof is omitted.
\end{proof}

\section{Examples}\label{s_examples}

In this section, we consider some examples and special cases. 


\subsection{Maximum overlap with separable states}

Let $X=A=[m]$ and $Y=B=[n]$. Fix a state $\sigma\in\cl S_1(\bb{C}^{m}\ten\bb{C}^n)$. For any separable state $\xi\ten\eta\in\bb{C}^m\ten\bb{C}^n$ we have 
$$\omega_{\rm loc}(\xi\ten\eta,\sigma)=\sup\{\tr(\rho\sigma): \rho\in\mathsf{Sep}(\bb{C}^{m}\ten\bb{C}^n)\},$$
where $\mathsf{Sep}(\bb{C}^{m}\ten\bb{C}^n)$ denotes the convex set of separable density operators on $\bb{C}^{m}\ten\bb{C}^n$.
Indeed, the orbit of $\xi\xi^*\ten\eta\eta^*$ under $\mathsf{QC}(\cl R_{\rm loc})$ is $\mathsf{Sep}(\bb{C}^{m}\ten\bb{C}^n)$, hence
\begin{align*}\omega_{\rm loc}(\xi\ten\eta,\sigma)&=\sup_{\Gamma\in \mathsf{QC}(\cl R_{\rm loc})}\tr(\Gamma(\xi\xi^*\ten\eta\eta^*)\sigma)\\
&=\sup\{\tr(\rho\sigma): \sigma\in\mathsf{Sep}(\bb{C}^{m}\ten\bb{C}^n)\}.
\end{align*}
Hence, if $\sigma$ is entangled then $\omega_{\rm loc}(\xi\ten\eta,\sigma)<1$. 

As expected, we always have $\omega_{q}(\xi\ten\eta,\sigma)=1$. Indeed, let $\Phi:M_m\ten M_m\to M_m$ be the partial trace $(\tr_m\ten\id)$ and similarly for $\Psi:M_n\ten M_n\to M_n$. Then 
$\Gamma:M_m\ten M_n\to M_m\ten M_n$ given by
$$\Gamma(\rho)=(\Phi\ten\Psi)\Sigma_{23}(\rho\ten\sigma), \ \ \ \rho\in M_m\ten M_n,$$
belongs to $\mathsf{QC}(\cl R_{\rm q})$, where $\Sigma_{23}$ denotes the swap in the second and third tensor legs. Clearly, $\Gamma(\xi\xi^*\ten\eta\eta^*)=\sigma$.  

\subsection{A finite rank quantum game}\label{ss_afrg}

Let $X=B=R=[n]$ and $Y=A=\{1\}$. Let $\xi = \frac{1}{\sqrt{n}}\sum_{i=1}^ne_i\otimes e_1\otimes e_i$ and $P=\epsilon_{1,1}\otimes Q$, where $Q$ is a projection in $M_n\otimes M_n$. Write $Q=\sum_{i,j}\alpha_{i,j}\otimes\epsilon_{i,j}$, $\alpha_{i,j}\in M_n$. Then
\begin{align*}
(\text{id}\otimes\phi_R)&(\xi\xi^*\otimes Q)\\ &=\frac{1}{n}(\text{id}\otimes\phi_R)\left(\sum_{i,j}\epsilon_{i,j}\otimes\epsilon_{1,1}\otimes\epsilon_{i,j}\otimes\epsilon_{1,1}\otimes\left(\sum_{l,m}\alpha_{l,m}\otimes\epsilon_{l,m}\right)\right)\\&=\frac{1}{n}\sum_{i,j}\epsilon_{i,j}\otimes\epsilon_{1,1}\otimes\epsilon_{1,1}\otimes \alpha_{j,i};
\end{align*}
by Theorem~\ref{alter}, 
$$\omega_{\rm loc}(\xi,P)=\frac{1}{n}\bignorm{\sum_{i,j}\epsilon_{i,j}\otimes \overline{\alpha_{j,i}}}_{\cl S_1^X\otimes_\epsilon M_B}=\frac{1}{n}\|T\|_{\cl B(M_X, M_B)},$$ where $T(x)=\sum_{i,j}\Tr(x\epsilon_{i,j})\overline{\alpha_{j,i}}$ is the completely positive map whose Choi matrix is $\overline{Q}$.  We have 
\begin{align*}\|T\|_{\cl B(M_X, M_B)}&=\|T(1)\|=\|T\|_{{\rm CB}(M_X,M_B)}\\&=\bignorm{\sum_{i,j}\epsilon_{i,j}\otimes \overline{\alpha_{j,i}}}_{\cl S_1^X\otimes_{\rm min} M_B}=n\omega_{\rm q}(\xi,P).
\end{align*}

We further claim that 
$\om_{\rm qc}(\xi,P) = \om_{\rm q}(\xi,P)$.
Indeed, since
$$\cl R_{\cl S_1(\mathbb C,\mathbb C^n)} = \cl R_{\mathbb{C}^n_c} =\mathbb{C} 
\ \mbox{ and  } \ \cl L_{\cl S_1(\mathbb C^n,\mathbb {C})}=\cl L_{\mathbb{C}^n_r}=\mathbb{C},$$ 
$\cl S_1(\mathbb C,\mathbb C^n)$ and $\cl S_1(\mathbb C^n,\mathbb {C})$ 
semi-commute, and hence, using \cite[Proposition 9.3.1]{er} at the second equality, 
\begin{eqnarray*}
\cl S_1(\mathbb C,\mathbb C^n)\ten_{\max} \cl S_1(\mathbb C^n,\mathbb {C})
& = & 
\mathbb C^n_c\ten_{\rm h} \mathbb C^n_r 
= \mathbb C^n_c\ten_{\min} \mathbb C^n_r\\
& = & 
\cl S_1(\mathbb C,\mathbb C^n)\ten_{\min} \cl S_1(\mathbb C^n,\mathbb {C}),
\end{eqnarray*}
up to a complete isometry, and the claim follows by Theorem~\ref{th_qcval}.

We now explicitly calculate the above norms for a class of projections $Q$. 
Let $\gamma_{k,i}\in \bb{C}$, $i \in [n]$, $k\in [m]$, be such that the vectors
$\gamma_k = e_1\otimes (\sum_{i=1}^n \gamma_{k,i} e_i\otimes e_i)$,
are mutually orthogonal and have norm one. Let $Q=\sum_k\gamma_k\gamma_k^*$. 
As above, we have
\begin{equation*}
  (\text{id}\otimes\phi_R)(\xi\xi^*\otimes\gamma_k\gamma_k^*)
  =\frac{1}{n}\sum_{i,j}\epsilon_{i,j}\otimes\epsilon_{1,1}\otimes\epsilon_{1,1}\otimes \gamma_{k,j}\overline{\gamma_{k,i}}\epsilon_{j,i}.
\end{equation*}
Consider the latter as an element of $\cl S_1^X(M_1)\otimes \cl S_1^1(M_B)=\cl S_1^X\otimes M_B$ and identify it with 
$\frac{1}{n}\sum_{i,j}\epsilon_{i,j}\otimes \gamma_{k,j}\overline{\gamma_{k,i}}\epsilon_{j,i}$. 
We have $$\omega_{\rm loc}(\xi,P)=\frac{1}{n}\bignorm{\sum_{k,i,j}\epsilon_{i,j}\otimes \overline{\gamma_{k,j}}\gamma_{k,i}\epsilon_{j,i}}_{\cl S_1^X\otimes_\epsilon M_B}=\frac{1}{n}\|T\|_{\cl B(M_X, M_B)},$$ where $T(x)=\sum_{i,j,k}\Tr(x\epsilon_{i,j})\overline{\gamma_{k,j}}\gamma_{k,i}\epsilon_{j,i}=\sum_kD_{k}^*xD_{k}$, where $D_{k}$ is the diagonal matrix with
$(\gamma_{k,i})_{i=1}^n$ on the diagonal. Thus, $$\|T\|_{\cl B(M_X,M_A)}=\bignorm{\sum_kD_{k}^*D_{k}}=\|T\|_{{\rm CB}(M_X,M_A)}=n\omega_{\rm q}(\xi,P).$$


\subsection{Quantum implication games}\label{ss_implg}

Let $X$ and $A$ be non-empty finite sets, and 
$P\in M_{XX}$ and $Q\in M_{AA}$ be projections. 
The \emph{quantum implication game} $P\hspace{-0.07cm}\Rightarrow \hspace{-0.07cm} Q$ has rules
requiring that, given an input state supported on $P$, the players respond with an output state 
supported on $Q$. 
We note that quantum 
implication games include as a special class the quantum graph homomorphism games \cite{bhtt, bhtt2}. We place them in the context of the present paper by fixing an orthonormal basis 
$\{\xi_i\}_{i=1}^k\subseteq \bb{C}^{XX}$ for the range of $P$, and associating with the pair $(P,Q)$ 
the quantum hypergraph game $(\nph, \pi)$, where $\pi$ is the uniform probability distribution on 
$\{\xi_i\xi_i^*\}_{i=1}^k$, and $\nph(\xi_i\xi_i^*) = Q$, $i = 1,\dots,k$. 

In \cite{tt-QNS}, the correlations $\Gamma$ satisfying condition (ii) of Corollary \ref{p_perf} below were called 
\emph{perfect strategies} for the game $P\hspace{-0.07cm}\Rightarrow \hspace{-0.07cm} Q$ (see also \cite{bhtt}). 
Corollary \ref{p_perf} thus 
establishes a metric characterisation of the existence of a perfect strategy of implication games,  
providing a link between the approaches pursued in 
\cite{junge, regev-vidick} on one hand, and \cite{bhtt, tt-QNS} on the other. 
We formulate the statement for quantum commuting values, but similar statements are true 
for the local, the one-way classical communication, or the no-signalling value,
following the natural modification in item (iii) implied by Theorem \ref{th_qcval}.

\begin{corollary}\label{p_perf}
Let $P\in M_{XX}$ (resp. $Q\in M_{AA}$) be a projection, 
$\{\xi_i\}_{i=1}^k$ (resp. $\{\eta_j\}_{j=1}^m$) be an orthonormal basis for the range of $P$ (resp. $Q$), 
and 
$\rho = [\overline{\xi_i\eta_j^*}]_{i,j}$, viewed as a column operator over $\cl B(\bb{C}^{AA},\bb{C}^{XX})$.
The following are equivalent:

\begin{itemize}
\item[(i)] $\omega_{\rm qc}(P\hspace{-0.07cm}\Rightarrow \hspace{-0.07cm} Q) = 1$;

\item[(ii)] 
there exists $\Gamma\in \cl Q_{\rm qc}$ such that for $\sigma\in M_{XX}$, we have 
$$\sigma = P\sigma P \ \Longrightarrow \ \Gamma(\sigma) = Q\Gamma(\sigma)Q;$$

\item[(iii)] $\|\rho\|_{\max} = 1$.
\end{itemize}
\end{corollary}

\begin{proof}
  (i)$\Leftrightarrow$(iii) is immediate from Theorem \ref{th_hypqhval}.

(i)$\Rightarrow$(ii) Since $\cl Q_{\rm qc}$ is closed, $\omega_{\rm qc}(P\hspace{-0.07cm}\Rightarrow \hspace{-0.07cm} Q) = 1$ implies that the supremum is attained at some $\Gamma\in \cl Q_{\rm qc}$, giving $k^{-1}\Tr(\Gamma(P)Q)=1$. Since $\frac{1}{k}P$ is a state, this implies $\Gamma(P) = Q\Gamma(P)Q$, which is equivalent to 
the condition in (ii). 


(ii)$\Rightarrow$(i) follows by a reversal of the argument in the previous paragraph. 
\end{proof}

Theorem \ref{th_alterq} has the following immediate corollary:

\begin{corollary}\label{c_PotimesQ}
  Let $P\in M_{XX}$ and $Q\in M_{AA}$ be projections. Then
\[\tightmath\omega_{\rm q}(P\hspace{-0.07cm}\Rightarrow \hspace{-0.07cm} Q) 
= 
\tfrac{1}{k}\|\overline{P}\otimes \overline{Q}\|_{\cl S_1^{A,X}\otimes_{\min}\cl S_1^{A,X}}\text{ and }
\omega_{\rm loc}(P\hspace{-0.07cm}\Rightarrow \hspace{-0.07cm} Q) 
= \tfrac{1}{k}\|\overline{P}\otimes \overline{Q}\|_{\cl S_1^{A,X}\otimes_{\varepsilon}\cl S_1^{A,X}},\]
where $k$ is the rank of $P$.
\end{corollary}


\subsection{Quantum XOR games} 

The purpose of the present subsection is to express the
value of quantum XOR games \cite{regev-vidick} through the expression (\ref{eq_omegaRP}). 
Recall the  operational interpretation of a quantum XOR game: 
the referee prepares a state 
$\xi\in H_X\otimes H_Y\otimes H_R$ and sends the $X$ and $Y$ registers to 
Alice and Bob, respectively, 
who share a state $\psi\in H\otimes K$; Alice and Bob then apply their (POVM) measurements   
$\cl E:=\{E_0,E_1\}\subseteq \cl B(H_X\otimes H)$ and $\cl F:=\{F_0, F_1\}\subseteq\cl B(H_Y\otimes H)$, respectively 
and return the outcomes $a$ and $b\in\{0,1\}$ to the referee who then 
determines if the players win or lose by applying
PVM's $(\Pi_0,1-\Pi_0)$ or $(\Pi_1,1-\Pi_1)$, depending on the parity $a\oplus b$.

The success probability of the strategy $(\cl E,\cl F,\psi)$ is
\begin{multline*}
\omega_{{\rm q},{\rm XOR}}(\cl E,\cl F,\psi) := \\
\tightmath\big\langle ((E_0\otimes F_0+E_1\otimes F_1)\otimes\Pi_0
+ 
(E_1\otimes F_0+E_0\otimes F_1)\otimes\Pi_1)(\psi\otimes\xi),\psi\otimes\xi\big\rangle.
\end{multline*}
As $\{E_0, E_1\} $ is a POVM there exist a Hilbert space $H'$ 
and an isometry $U : H_X\otimes H\to \mathbb C^2\otimes H'$, 
such that $E_a=U^*(\epsilon_{a,a}\otimes1_{H'})U$, $a\in\{0,1\}$. 
Similarly, there exist a Hilbert space $K'$ and 
an isometry $V: H_Y\otimes K\to \mathbb C^2\otimes K'$, 
such that $F_b=V^*(\epsilon_{b,b}\otimes1_{K'})V$, $b\in\{0,1\}$.
Let 
$$P = (\epsilon_{0,0}\otimes\epsilon_{0,0}+\epsilon_{1,1}\otimes\epsilon_{1,1})\otimes\Pi_0+(\epsilon_{1,1}\otimes\epsilon_{0,0}+\epsilon_{0,0}\otimes\epsilon_{1,1})\otimes\Pi_1$$ 
and 
$\Gamma = \Gamma_{E\otimes F,\psi\psi^*} : M_{XY}\to \cl D_2\otimes \cl D_2$ be the 
QNS correlation determined by the stochastic operator matrices 
$E := E_0\oplus E_1$, $F := F_0\oplus F_1$ and the unit vector $\psi$. 
Write $\xi = \sum_{x,y,r}\alpha_{x,y,r}e_x\otimes e_y\otimes e_r\in H_X\otimes H_Y\otimes H_R$ in the 
canonical orthonormal basis of $H_X\otimes H_Y\otimes H_R$. 
Then 
\begin{align*}
&\langle( P\otimes I_{H'\otimes K'})(U\otimes V\otimes 1_R)(\psi\otimes\xi), (U\otimes V\otimes I_R)(\psi\otimes\xi)\rangle\\&=
\sum_{x,y,r,x',y',r'}\sum_{a,b,a',b'}\alpha_{x,y,r}\overline{\alpha_{x',y',r'}}\tr\big((u_{a,x}\otimes v_{b,y})\psi((u_{a',x'}\otimes v_{b',y'})\psi)^*\big)\\[-2ex]
&\hspace{6cm}\times \tr(P(\epsilon_{a,a'}\otimes\epsilon_{b,b'}\otimes\epsilon_{r,r'}))\\
&=
\sum_{x,y,r,x',y',r'}\sum_{a,b,a',b'}\alpha_{x,y,r}\overline{\alpha_{x',y',r'}}\tr((u_{a,x}\otimes v_{b,y})\psi)((u_{a',x'}\otimes v_{b',y'})\psi)^*))\\[-2ex]
  &\hspace{6cm}\times \tr((\epsilon_{a',a}\otimes\epsilon_{b',b}\otimes\epsilon_{r,r'})P^{\rm t})
\\&= 
\tr((\Gamma\otimes {\rm id})(\bar\xi\bar\xi^*)P^{\rm t}).
\end{align*}
Since $P^{\rm t} = \overline{P}$, we now see that
\begin{align*}
&\omega_{{\rm q}, {\rm XOR}}(\cl E,\cl F,\psi)
= 
\langle (U^*\otimes V^*\otimes 1_R)((\epsilon_{0,0}\otimes\epsilon_{0,0}+\epsilon_{1,1}\otimes\epsilon_{1,1})\otimes\Pi_0\\
&\quad +(\epsilon_{1,1}\otimes\epsilon_{0,0}+\epsilon_{0,0}\otimes\epsilon_{1,1})\otimes\Pi_1)\otimes I_{H'\otimes K'}(U\otimes V\otimes I_R)(\psi\otimes\xi),\psi\otimes\xi\rangle\\
&=\langle( P\otimes I_{H'\otimes K'})(U\otimes V\otimes 1_R)(\psi\otimes\xi), (U\otimes V\otimes I_R)(\psi\otimes\xi)\rangle\\
&=\Tr((\Gamma\otimes\id)(\bar\xi\bar{\xi}^*)\overline{P}).
\end{align*}
Consequently,
$$\omega_{{\rm q},{\rm XOR}}(\xi):=\sup_{\cl E,\cl F,\psi}\omega_{{\rm q}, {\rm XOR}}(\cl E,\cl F,\psi)=\omega_{\rm q}(\bar{\xi},\bar{P}).$$


\end{document}